\pdfoutput=1 

\documentclass[11pt]{article}
\usepackage[margin=1in]{geometry} 

\usepackage{hyperref} 

\usepackage{amssymb, amsmath, amsthm} 
\usepackage{physics, braket} 

\usepackage{graphicx, xcolor, colordvi} 
\usepackage{tikz} 
\usetikzlibrary{shapes, arrows.meta, positioning, shapes.geometric} 

\usepackage[utf8]{inputenc} 
\usepackage[english]{babel} 
\usepackage{textcomp} 
\usepackage{dsfont} 

\usepackage[center]{caption} 
\usepackage{float} 
\usepackage{subcaption} 

\usepackage{enumitem} 

\usepackage[blocks]{authblk} 

\usepackage{array, rotating, multirow, tabularx} 

\usepackage{varioref, cleveref} 
\usepackage{wrapfig, stmaryrd, makecell} 
\usepackage{comment, url, upgreek, fancyref} 
\usepackage{autonum} 

\newtheoremstyle{newdefinition} 
{} 
{} 
{\normalfont} 
{} 
{\bfseries} 
{} 
{\newline} 
{\thmname{#1} \thmnumber{#2}\thmnote{ (#3)}} 

\newtheoremstyle{newplain} 
{} 
{} 
{\itshape} 
{} 
{\bfseries} 
{} 
{1em} 
{\thmname{#1} \thmnumber{#2}\thmnote{ (#3)}} 

\newtheoremstyle{newremark} 
{} 
{} 
{\normalfont} 
{} 
{\bfseries} 
{} 
{1em} 
{\thmname{#1}} 

\theoremstyle{newdefinition} 
\newtheorem{definition}{Definition}[section] 

\theoremstyle{newplain} 
\newtheorem{theorem}[definition]{Theorem} 
\newtheorem{lemma}[definition]{Lemma} 
\newtheorem{proposition}[definition]{Proposition} 
\newtheorem{corollary}[definition]{Corollary} 
\newtheorem{remark}[definition]{Remark} 
\newtheorem{example}[definition]{Example} 

\newtheoremstyle{myplain} 
{5pt} 
{5pt} 
{\itshape} 
{0pt} 
{\bfseries} 
{} 
{5pt plus 1pt minus 1pt} 
{} 

\theoremstyle{myplain} 
\newtheorem*{theorem*}{Theorem} 
\newtheorem*{corollary*}{Corollary} 

\DeclareMathOperator{\R}{\mathbb{R}} 
\DeclareMathOperator{\C}{\mathbb{C}} 

\DeclareMathOperator{\cH}{\mathcal{H}} 
\DeclareMathOperator{\HH}{\mathcal{H}} 
\DeclareMathOperator{\cS}{\mathcal{S}} 
\DeclareMathOperator{\cK}{\mathcal{K}} 
\DeclareMathOperator{\cC}{\mathcal{C}} 
\DeclareMathOperator{\cB}{\mathcal{B}} 
\DeclareMathOperator{\cF}{\mathcal{F}} 

\DeclareMathOperator{\1}{\mathds{1}} 

\newcommand{\coloneqq}{:=} 
\newcommand{\eqqcolon}{=:} 
\newcommand{\Qalpha}{\widetilde Q_\alpha} 
\newcommand{\Dalpha}{\widetilde D_\alpha} 
\newcommand{\alphaQ}{\widetilde Q_\alpha} 
\newcommand{\alphaI}{\widetilde I^{\uparrow}_\alpha} 
\newcommand{\twoline}[2]{\genfrac{}{}{0pt}{}{#1}{#2}} 

\definecolor{tbcolor1}{HTML}{FC7A1E} 
\definecolor{tbcolor2}{HTML}{F9C784} 
\definecolor{tbcolor3}{HTML}{485696} 

\hypersetup{
    colorlinks=true, 
    linkcolor=tbcolor3, 
    urlcolor=tbcolor3, 
    citecolor=tbcolor1, 
    hypertexnames=false 
}

\makeatletter
\def\@fnsymbol#1{\ifcase#1\or \textdagger\or \textdaggerdbl\or \S\or \P\else\@arabic{#1}\fi}
\makeatother

\begin{document}

\title{\vspace{-0.2cm}Unified framework for continuity of sandwiched R\'{e}nyi divergences}
\author[1]{Andreas Bluhm\thanks{andreas.bluhm@univ-grenoble-alpes.fr}}
\author[2,3]{{\'A}ngela Capel\thanks{angela.capel@uni-tuebingen.de}}
\author[2]{Paul Gondolf$*$\thanks{paul.gondolf@uni-tuebingen.de}}
\author[2,4]{Tim Möbus\thanks{moebustim@gmail.com}}
\affil[1]{Univ.\ Grenoble Alpes, CNRS, Grenoble INP, LIG, 38000 Grenoble, France}
\affil[2]{Fachbereich Mathematik, Universit\"at T\"ubingen, 72076 T\"ubingen, Germany}
\affil[3]{Department of Applied Mathematics and Theoretical Physics, University of Cambridge, Cambridge CB3 0WA, United Kingdom}
\affil[4]{Department of Mathematics, Technische Universit\"at München, 80333 München, Germany}
\date{\today}

\maketitle

\vspace{-0.8cm}

\begin{abstract}
    In this work, we prove uniform continuity bounds for entropic quantities related to the sandwiched R{\'e}nyi divergences such as the sandwiched R{\'e}nyi conditional entropy. We follow three different approaches: The first one is the ``almost additive approach'', which exploits the sub-/ superadditivity and joint concavity/ convexity of the exponential of the divergence. In our second approach, termed the ``operator space approach'', we express the entropic measures as norms and utilize their properties for establishing the bounds. These norms draw inspiration from interpolation space norms. We not only demonstrate the norm properties solely relying on matrix analysis tools but also extend their applicability to a context that holds relevance in resource theories. By this, we extend the strategies of Marwah and Dupuis as well as Beigi and Goodarzi employed in the sandwiched Rényi conditional entropy context. Finally, we merge the approaches into a mixed approach that has some advantageous properties and then discuss in which regimes each bound performs best. Our results improve over the previous best continuity bounds or sometimes even give the first continuity bounds available. In a separate contribution, we use the ALAFF method, developed in a previous article by some of the authors, to study the stability of approximate quantum Markov chains. 
\end{abstract}

\setcounter{tocdepth}{2}        
\renewcommand{\contentsname}{Table of Contents}
\tableofcontents 

\section{Introduction}
Entropic quantities are indispensable for classical and quantum information theory to characterize information-theoretic tasks. Examples of such quantities include various forms of (conditional) entropies, (conditional) mutual information, and many others. In applications, one often has especially convenient quantum states for which the entropic quantity of interest can be evaluated explicitly and one would therefore like to reduce nearby quantum states to this case. This is why continuity bounds for entropic quantities have become a ubiquitous tool. They provide upper bounds on 
\begin{equation}
    \sup\{|g(\rho) - g(\sigma)|: \rho, \sigma \in \mathcal S_0, \operatorname{dist}(\rho, \sigma) \leq \varepsilon\}.
\end{equation}
Here, $g$ is the entropic quantity of interest, $\mathcal S_0$ is a suitable subset of the set of quantum states $\cS(\cH)$ in a finite-dimensional Hilbert space $\mathcal H$, and $\operatorname{dist}(\cdot, \cdot)$ is an appropriate metric on $\cS(\cH)$ (often the trace distance). If the bound on the supremum only depends on $\varepsilon$ and $\mathcal S_0$, but not on $\rho$ and $\sigma$, then $g$ is \emph{uniformly continuous} on $\mathcal S_0$. 

One of the earliest continuity statements in quantum information theory is the continuity bound on the von Neumann entropy provided by Fannes in \cite{Fannes-ContinuityEntropy-1973}, which was later improved in \cite{Audenaert-ContinuityEstimateEntropy-2007, Petz2008}. Another well-known bound in quantum information theory is the Alicki-Fannes inequality for the conditional entropy \cite{AlickiFannes-2004}, for which an almost tight version was proven by Winter in \cite{Winter-AlickiFannes-2016}:
\begin{equation}
     \abs{ H_\rho( A | B ) - H_\sigma( A | B )} \leq 2 \varepsilon \log d_A + (1+ \varepsilon) h\left(\frac{\varepsilon}{1+\varepsilon}\right) ,\
\end{equation}
where $\varepsilon$ is an upper bound on the trace distance between quantum states $\rho$ and $\sigma$ and $h$ is the binary entropy. Shirokov and others \cite{mosonyi2011quantum, synak2006asymptotic} realized that the proof method used for this result does not only work for the conditional entropy but can be generalized \cite{Shirokov-AFWmethod-2020, Shirokov-ContinuityReview-2022}. Shirokov named this approach the Alicki-Fannes-Winter (AFW) method. Recently, the method was further developed in \cite{Bluhm2022ContinuityBounds} by some of the authors of the present article and applied to quantities based on the Belavkin-Staszewski relative entropy \cite{BelavkinStaszewski-BSentropy-1982}.

In this article, we focus on entropic quantities derived from the sandwiched R{\'e}nyi divergences \cite{muller2013quantum,WildeWinterYang2014SandwichedRenyiDivergences}
\begin{equation}
     \Dalpha(\rho \Vert \sigma) = \frac{1}{\alpha - 1}\log \tr[(\sigma^{\frac{1 - \alpha}{2\alpha}} \rho \sigma^{\frac{1 - \alpha}{2\alpha}})^\alpha] ,\
\end{equation}
where $\alpha \in [1/2,1) \cup (1,\infty)$. Examples of such entropic quantities include the sandwiched Rényi conditional entropy and the sandwiched Rényi mutual information:
\begin{equation}
        \widetilde H^{\uparrow}_\alpha(A|B)_\rho := \sup\limits_{\tau_B \in \cS(\cH_B)}  - \Dalpha(\rho_{AB} \Vert \1_A \otimes \tau_B) \; ,   \quad \alphaI(A:B)_\rho = \inf\limits_{\substack{\tau_A \in \cS(\cH_A),\\ \tau_B \in \cS(\cH_B)}} \widetilde D_\alpha(\rho_{AB} \Vert \tau_A \otimes \tau_B) \, .
\end{equation}
Recently, there has been increased interest in continuity bounds for the sandwiched R{\'e}nyi conditional entropy. In \cite{marwah2022}, it was shown that for $\alpha \in [1/2, 1)$,
\begin{equation} \label{eq:marwah-smaller-bound-than-1}
    \abs{\widetilde H^{\uparrow}_\alpha(A|B)_\rho - \widetilde H^{\uparrow}_\alpha(A|B)_\sigma } \leq \log(1+\varepsilon) + \frac{1}{1-\alpha} \log(1+\varepsilon^\alpha d_A^{2(1-\alpha)}-\frac{\varepsilon}{(1+\varepsilon)^{1-\alpha}})
\end{equation}
and for $\alpha \in (1, \infty)$, they used duality to infer that 
\begin{equation}\label{eq:marwah-larger-bound-than-1}
    \abs{\widetilde H^{\uparrow}_\alpha(A|B)_\rho - \widetilde H^{\uparrow}_\alpha(A|B)_\sigma } \leq \log(1+\sqrt{2\varepsilon}) + \frac{1}{1-\beta} \log(1+\sqrt{2\varepsilon}^\beta d_A^{2(1-\beta)}-\frac{\sqrt{2\varepsilon}}{(1+\sqrt{2\varepsilon})^{1-\beta}}) ,\
\end{equation}
where $\beta$ is such that $\alpha^{-1} + \beta^{-1} = 2$.

Using techniques from the interpolation of operator space, it was shown in \cite[Theorem 6.2]{Beigi2022} that for $\alpha \in (1,\infty)$
\begin{equation}\label{eq:beigi}
    \abs{\widetilde H^{\uparrow}_\alpha(A|B)_\rho - \widetilde H^{\uparrow}_\alpha(A|B)_\sigma } \leq \alpha^\prime \log(1 + 2 \varepsilon d_A^{2/\alpha^\prime}) ,\
\end{equation}
where $\alpha^\prime = \alpha/(\alpha-1)$. The authors of \cite{Beigi2022} note that their bound is better than Eq.\ \eqref{eq:marwah-larger-bound-than-1} for large $\alpha$, but that it diverges for $\alpha \to 1$.

On a high level, the proof of \cite{marwah2022} uses sub-/ superadditivity of the exponential of the sandwiched Rényi divergence, while \cite{Beigi2022} makes a connection to norms on interpolation spaces. Based on these ideas we develop a unified approach to proving not only continuity bounds for the sandwiched Rényi conditional entropy which improves or extends the ones discussed above, but further allows us to prove bounds for related entropic quantities such as the sandwiched Rényi mutual information. In particular, we introduce a new family of norms defined as optimizations over amalgamations with compact convex sets of positive operators.

Generally, we consider uniform continuity bounds on quantities of the form 
\begin{equation}
    \widetilde D_{\alpha, \mathcal C}(\rho) := \inf_{\sigma \in \mathcal C} \widetilde D_\alpha(\rho \| \sigma) \, ,
\end{equation}
where $\mathcal C$ is a compact convex subset of the quantum states containing at least one full-rank element.
Such quantities have appeared in the context of resource theories (see, e.g., \cite{anshu2018quantifying, lami2023attainability, rubboli2022new, zhu2017coherence}) and were termed entropy of resource. Resource theories provide a framework to answer questions about the interconvertability of states, using only allowed operations. Every resource theory has two main ingredients: (i) the set of free states, i.e., states that do not possess the resource (ii) the set of free operations which map the set of free states to itself, i.e., these operations do not create the resource. The best-known example of such a theory is the theory of entanglement, in which the free states are the separable states and the free operations are the local quantum operations and classical communication (LOCC). Others include, for example, the resource theories of coherence and asymmetry.

One way to quantify the resourcefulness of a quantum state in a given resource theory is to measure its distance to the set of free states. Common choices of distance measure include the relative entropy \cite{anshu2018quantifying, lami2023attainability} and the sandwiched Rényi entropies \cite{rubboli2022new, zhu2017coherence}. In the latter case, for $\mathcal C$ the set of free states, $\widetilde D_{\alpha, \mathcal C}$ is the corresponding resource measure. The results of this article can therefore be used to quantify the continuity of popular resource measures such as the Rényi relative entropies of entanglement and coherence studied, e.g., in \cite{zhu2017coherence}.

The article is organized as follows: In \Cref{sec:main-results}, we present the main results of this paper, before continuing with the necessary preliminaries on divergences and entropic quantities in \Cref{sec:preliminaries}. We will follow three different approaches in this article: an almost additive approach, an approach based on operator spaces, and one where both methods are mixed. The tools for all these approaches are developed in \Cref{sec:technical_tools}, such that all the continuity bounds will follow straightforwardly from the theorems proven in this section. In \Cref{sec:CB_Renyi}, we derive and discuss our continuity bounds on the sandwiched R{\'e}nyi divergences and their derived entropic quantities. In \Cref{sec:approxQMC}, we showcase how continuity bounds can be useful for studying approximate quantum Markov chains. Notably, the continuity bounds in this section do not stem from the three approaches mentioned previously but use the ALAFF method introduced in \cite{Bluhm2022ContinuityBounds} by some of the present authors. Finally, the paper finishes with a short discussion in \Cref{sec:discussion}. 

\section{Main results} \label{sec:main-results}

In this section, we summarize the main results --- the achieved continuity bounds for the sandwiched R{\'e}nyi conditional entropy, mutual information, and the divergence itself with a fixed second argument. All quantities are defined via the \textit{sandwiched R{\'e}nyi divergence}, i.e., 
\begin{equation}
    \Dalpha(\rho \Vert \sigma) := \frac{1}{\alpha-1}\log(\widetilde Q_\alpha(\rho \Vert \sigma)) = \frac{1}{\alpha - 1}\log \tr[(\sigma^{\frac{1 - \alpha}{2\alpha}} \rho \sigma^{\frac{1 - \alpha}{2\alpha}})^\alpha].
\end{equation}
Then, the \textit{sandwiched R{\'e}nyi conditional entropy} is defined as
\begin{equation}
    \widetilde H^{\uparrow}_\alpha(A|B)_\rho := -\sup\limits_{\tau_B \in \cS(\cH_B)} \widetilde D_\alpha(\rho_{AB} \Vert \1_A \otimes \tau_B) \, ,  
\end{equation}
the \textit{mutual information} by
\begin{equation}
    \alphaI(A:B)_\rho = \inf\limits_{\substack{\tau_A \in \cS(\cH_A),\\ \tau_B \in \cS(\cH_B)}} \widetilde D_\alpha(\rho_{AB} \Vert \tau_A \otimes \tau_B) \, ,
\end{equation}
and the \textit{divergence} itself is just considered as a function in the first argument with a fixed second argument. A precise definition of these quantities can be found in \Cref{subsec:div-entropic-quantities}.\par
To prove bounds for these quantities, we explore three different methods (see \Cref{sec:technical_tools}).
\begin{itemize}
    \item \textbf{The almost additive approach:} This approach is inspired by \cite{marwah2022} and uses joint convexity/concavity and super-/subadditivity of $\widetilde Q_\alpha$. The name is motivated by the fact that this super-/subadditivity reduces to almost concavity of $\rho \mapsto D(\rho \Vert \sigma)$, respectively almost convexity of the von Neuman entropy if $\alpha \to 1$. The main result in this approach, from which the respective continuity bounds follow straightforwardly, is a continuity bound on the distance to a compact convex subset of the quantum states in Theorem \ref{theo:minimial-distance-cc-set-almost-additive-approach}.
    \item  \textbf{The operator space approach:} It is inspired by \cite{Beigi2022} and relates $\widetilde Q_\alpha$ to a norm. In this context, showing the norm properties poses a challenge, yet yields continuity bounds directly through the triangle inequality. In particular, we define new quantities for $1 \le q' \le p' \le \infty$ such that $\frac{1}{r} = \frac{1}{q'} - \frac{1}{p'}$, namely
    \begin{equation}
       \norm{\cdot}_{\cC, p', q'}^*:\cB(\cH) \to [0, \infty), \quad  X \mapsto \norm{X}_{\cC, p', q'}^* := \inf\limits_{c \in \cC, c > 0} \norm{c^{-\frac{1}{2r}} X c^{-\frac{1}{2r}}}_{p'}
    \end{equation}
    where $X \in \cB(\cH)$ and $\norm{\cdot}_{p'} = \left(\tr[|\cdot|^{p'}]\right)^{\frac{1}{p'}}$. Usually, $\mathcal C$ is considered a subalgebra of $\cB(\cH)$. Our proof, however, extends this to compact convex subsets of $\mathcal B(\mathcal H)$ consisting of positive semidefinite operators containing at least one full-rank state. We prove, without having to resort to interpolation theory, that the dual quantity is subadditive on positive semi-definite elements, i.e.,
    \begin{equation}
        \norm{X + Y}_{\cC, p', q'}^* \le \norm{X}_{\cC, p', q'}^* + \norm{Y}_{\cC, p', q'}^*
    \end{equation}
    for $X, Y \in \cB_{\ge 0}(\cH)$. This yields another continuity bound on the distance to a compact convex subset of the quantum states in Theorem \ref{theo:minimial-distance-cc-set-operator-approach}, which allows us to directly infer all our continuity bounds based on this approach in the following. 
    \item  \textbf{The mixed approach:} It combines the previous two approaches. The main theorem in this approach is Theorem \ref{theo:minimial-distance-cc-set-mixed-approach}.
\end{itemize}

The table below illustrates the various upper bounds that our work encompasses with each bound having a range in $\alpha$ on which it is superior. Note that one bound was already previously proven in \cite{marwah2022}, which we marked in the table with the reference. All other bounds are new. 

\begin{theorem*}\label{thm:main-continuity-bounds}
     Let $\rho, \sigma, \tau \in \cS(\cH_A \otimes \cH_B)$\footnote{Note that this includes the case $S(\cH)$ by choosing $\cH_A=\cH$ and $\cH_B=\C$.} with $\ker \tau \subseteq \ker \rho \cap \ker \sigma$ and $\tfrac{1}{2}\norm{\rho - \sigma}_1 \le \varepsilon$, define $\widetilde m_\tau$ to be the minimal non-zero eigenvalue of $\tau$, and $m = \min\{d_A, d_B\}$. Then, the following continuity bounds for the sandwiched Rényi conditional entropy, mutual information, and divergence in the first input (\ref{subsec:div-entropic-quantities}) hold, i.e., for an entropic quantity $g_\alpha$, one finds upper bounds on $|g_\alpha(\rho) - g_\alpha(\sigma)|$:
\end{theorem*}

\newcolumntype{Y}{>{\centering\arraybackslash}X}
\renewcommand{\arraystretch}{2}
\begin{flushleft}
    \begin{tabularx}{\textwidth}{|@{}l|c|Y|c|} 
        \hline
        & \textbf{Approach} & \textbf{Continuity bound} & \textbf{$\alpha$} \\ \hline \hline
        \multirow{4}{*}{\hspace{0.2cm}\begin{sideways}\parbox{2.6cm}{Conditional \\ Entropy (\ref{sec:sand-cond-entropy})}\end{sideways}} 
        & \multirow{2}{*}{A.~additive} 
        & $\log(1 + \varepsilon) + \frac{1}{1 - \alpha}\log(1 + \varepsilon^\alpha d_A^{2(1 - \alpha)} - \frac{\varepsilon}{(1 + \varepsilon)^{1 - \alpha}})$ \cite{marwah2022} 
        & $[\tfrac{1}{2},1)$ \\ \cline{3-4}
        & 
        & $\log(1 + \varepsilon) + \frac{1}{\alpha - 1}\log(1 + \varepsilon d_A^{2(\alpha - 1)} - \frac{\varepsilon^\alpha}{(1 + \varepsilon)^{\alpha - 1}})$ 
        & $(1,\infty)$ \\ \cline{2-4}
        & Op.~space  
        & $\frac{\alpha}{\alpha - 1}\log(1 + \varepsilon d_A^{2\frac{\alpha - 1}{\alpha}})$ 
        & $(1,\infty)$ \\ \cline{2-4}
        & Mixed 
        & $\log(1 + \varepsilon) + \frac{\alpha}{\alpha - 1}\log(1 + \varepsilon d_A^{2 \frac{\alpha - 1}{\alpha}} - \frac{\varepsilon^{2 - \frac{1}{\alpha}}}{(1 + \varepsilon)^{\frac{\alpha - 1}{\alpha}}})$ 
        & $(1,\infty)$\\ \hline
        
        \multirow{2}{*}{\hspace{0.2cm}\begin{sideways}\parbox{1.9cm}{Mutual\\Info.~(\ref{sec:sand-mutual-info})}\end{sideways}} 
        & \multirow{2}{*}{A.~additive}  
        & $ 2\log(1 + \varepsilon^{\frac{1}{\alpha}}) + \frac{1}{1 - \alpha}\log(1 + \varepsilon^\alpha m^{2(1 - \alpha)} - \frac{\varepsilon^{\frac{1}{\alpha}}}{(1 + \varepsilon^\frac{1}{\alpha})^{2(1 - \alpha)}})$ 
        & $[\tfrac{1}{2},1)$ \\ \cline{3-4}
        & 
        & $2\log(1 + \varepsilon^{\frac{1}{\alpha}}) + \frac{1}{\alpha - 1}\log(1 + \varepsilon^{\frac{1}{\alpha}} m^{2(\alpha - 1)} - \frac{\varepsilon^\alpha}{(1 + \varepsilon^{\frac{1}{\alpha}})^{2(\alpha - 1)}})$ 
        & $(1,\infty)$ \\ \hline
        
        \multirow{4}{*}{\hspace{0.2cm}\begin{sideways}\parbox{2.8cm}{$1^{\text{st}}$ Entry of \\ Divergence (\ref{sec:sand-first-arg})}\end{sideways}} 
        & \multirow{2}{*}{A.~additive\footnote{The continuity bound directly implies the same bound as a divergence bound \ref{subsec:divergence-bound} by choosing $\sigma=\tau$. \label{foot:div-bound}} } 
        & $\log(1 + \varepsilon) + \frac{1}{1 - \alpha}\log(1 + \varepsilon^\alpha \widetilde m_\tau^{\alpha - 1} - \frac{\varepsilon}{(1 + \varepsilon)^{1 - \alpha}})$
        & $[\tfrac{1}{2},1)$ \\ \cline{3-4}
        & 
        & $\log(1 + \varepsilon) + \frac{1}{\alpha - 1}\log(1 + \varepsilon \widetilde m_\tau^{1 - \alpha} - \frac{\varepsilon^\alpha}{(1 + \varepsilon)^{\alpha - 1}})$
        & $(1,\infty)$ \\ \cline{2-4}
        & Op.~space\footref{foot:div-bound} 
        & $\frac{\alpha}{\alpha - 1} \log(1 + \varepsilon \widetilde m_\tau^{\frac{1 - \alpha}{\alpha}})$
        & $(1,\infty)$ \\ \cline{2-4}
        & Mixed\footref{foot:div-bound} 
        & $\log(1 + \varepsilon) + \frac{\alpha}{\alpha - 1}\log(1 + \varepsilon \widetilde m_\tau^{\frac{1 - \alpha}{\alpha}} - \frac{\varepsilon^{2 - \frac{1}{\alpha}}}{(1 + \varepsilon)^{\frac{\alpha - 1}{\alpha}}})$
        & $(1,\infty)$\\ \hline
    \end{tabularx}
\end{flushleft}

For the various results presented in the table, it is important to note that none is universally preferable across the entire interval of $\alpha$. Rather, each method possesses its strengths and limitations. In the following, we compare our results using the sandwiched R{\'e}nyi conditional entropy. The accompanying diagram (\Cref{fig:plot-main-results}) demonstrates the {almost additive method} is particularly well-suited for $\alpha$ close to $1$ and small dimensions. The mixed method, which as the
\begin{wrapfigure}{l}{0.54\linewidth}
    \centering
    \includegraphics[width=\linewidth]{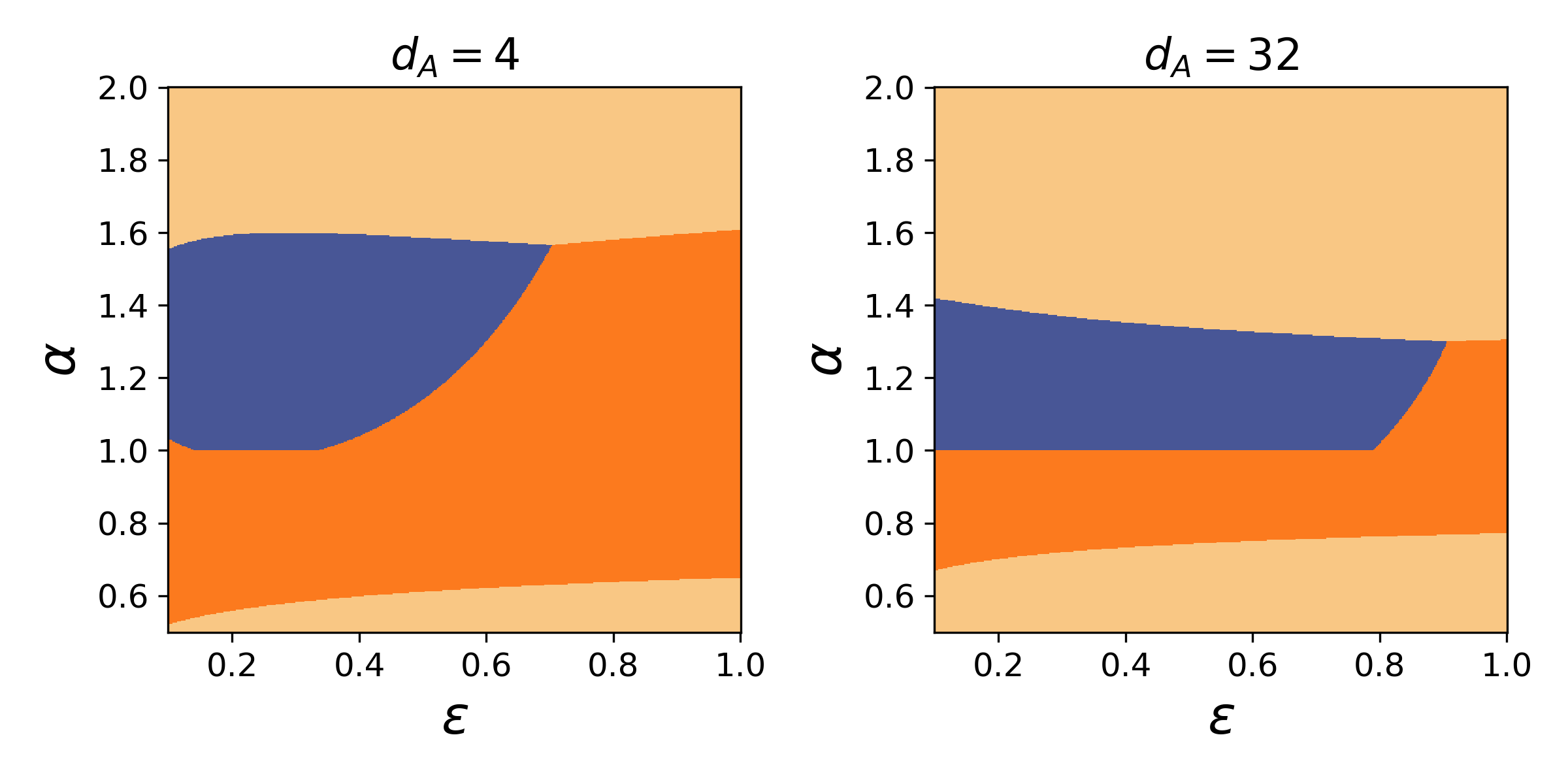}
    \vspace{-1cm}
    \caption{Continuity bounds for $\tilde{H}^{\uparrow}_\alpha(A|B)_\rho$ proven by the \textcolor{tbcolor1}{\rule{0.3cm}{0.3cm} almost additive}, \textcolor{tbcolor2}{\rule{0.3cm}{0.3cm} operator space}, and \textcolor{tbcolor3}{\rule{0.3cm}{0.3cm} mixed} approach, where the visible colour indicates the tightest bound.}
    \label{fig:plot-main-results}
\end{wrapfigure}
name suggests is a combination of the almost additive and operator space method, can even improve this property in large dimensions, while the operator space method is better suited for large $\alpha$. Even though the mixed method results in a slightly weaker bound, it performs well across the entire regime $\alpha \in (1,\infty)$. The continuity bounds proven for the mutual information and the first entry of the divergence come with similar scaling. In comparison to the existing results mentioned in Eqs.\ \eqref{eq:marwah-smaller-bound-than-1}, \eqref{eq:marwah-larger-bound-than-1}, and \eqref{eq:beigi}, our bound proven by the almost additive approach performs in the regime $\alpha\in(1,\infty)$ by a power of two better than the result in \cite{marwah2022} for small $\varepsilon$. The operator space approach improves the result in \cite{Beigi2022} by an order of two. Furthermore, \cite{marwah2022, Beigi2022} treat only the case of the sandwiched R{\'e}nyi conditional entropy, whereas our versions can also be used for other entropic quantities such as the sandwiched R{\'e}nyi mutual information and the sandwiched R{\'e}nyi divergence with fixed second argument.

For even more general quantities, where none of the inputs of the sandwiched R{\'e}nyi divergence is fixed, we apply the ALAFF method introduced in \cite{Bluhm2022ContinuityBounds} to estimate the continuity bounds for the $\alphaQ$ and related quantities of the sandwiched Rényi divergences. 
The analysis of the bounds is worked out in \Cref{sec:approxQMC} and applied in the context of approximate quantum Markov chains.

 As covered in the introduction, more results exist for the limiting case $\alpha\rightarrow1$. Here, the sandwiched R{\'e}nyi conditional entropy converges to the usual quantum conditional entropy, the sandwiched R{\'e}nyi mutual information to the usual quantum mutual information and
 the divergence itself to the relative entropy. Similarly, for the limit $\alpha\rightarrow\infty$, the quantities converge respectively to the min-conditional entropy, the max-mutual information and the max-divergences. Moving forward, we will focus on exploring the limits to compare our methods with each other but also with already established continuity bounds.
\begin{corollary*}
     Let $\rho, \sigma, \tau, \varepsilon, \widetilde m_\tau$, and $m$ be defined as in the beginning of \Cref{thm:main-continuity-bounds}. Then, the obtained continuity bounds converge for $\alpha\rightarrow1$ or $\alpha\rightarrow\infty$ in the following way:
\end{corollary*}
\begin{center}
    \begin{tabularx}{\textwidth}{|@{}l | c | c | Y|} 
        \hline
         & \textbf{Approach} & $\alpha\rightarrow1$ & $\alpha\rightarrow\infty$ \\\hline\hline
        \multirow{3}{*}{\hspace{0.3cm}\parbox{2.3cm}{Conditional \\ Entropy (\ref{sec:sand-cond-entropy})}} 
        & A.~additive & $2 \varepsilon \log d_A + (1+ \varepsilon) h\big(\frac{\varepsilon}{1+\varepsilon}\big)$ & $\log(1 + \varepsilon) + 2 \log{d_A}$ \\\cline{2-4}
        & Op.~space & $\infty$ & $\log(1 + \varepsilon d_A^2)$ \\\cline{2-4}
        & Mixed & $2 \varepsilon \log d_A + (1+ \varepsilon) h\big(\frac{\varepsilon}{1+\varepsilon}\big)$ & $\log(1+\varepsilon d_A^2 + \varepsilon(1 + \varepsilon(d_A^{2}-1)))$\\\hline
        
        \hspace{0.5cm}\parbox{1.8cm}{Mutual\\Info.~(\ref{sec:sand-mutual-info})}
        & A.~additive  & $2 \varepsilon \log{m} + 2 (1+\varepsilon) h\big(\frac{\varepsilon}{1+\varepsilon}\big)$ & $\log{4 m^2}$\\\hline
        
        \multirow{3}{*}{\hspace{0.2cm}\parbox{2.8cm}{$1^{\text{st}}$ Entry of \\ Divergence (\ref{sec:sand-first-arg})}}
        & A.~additive & $\varepsilon\log(\widetilde m_\tau^{-1}) + (1+\varepsilon) h\big(\frac{\varepsilon}{1+\varepsilon}\big)$ & $\log(1+\varepsilon) + \log(\widetilde m_\tau^{-1})$\\\cline{2-4}
        &  Op.~space & $\infty$ & $\log(1 + \varepsilon \widetilde m_\tau^{-1})$\\\cline{2-4}
         & Mixed & $\varepsilon\log(\widetilde m_\tau^{-1}) + (1+ \varepsilon) h\big(\frac{\varepsilon}{1+\varepsilon}\big)$ & $\log(1+\varepsilon \frac{1}{\widetilde m_\tau} + \varepsilon(1 + \varepsilon(\frac{1}{\widetilde m_\tau}-1)))$\\\hline
    \end{tabularx}~\\[2ex]
\end{center}

The corollary shows that the almost additive and mixed approach for the sandwiched Rényi conditional entropy converge for $\alpha\rightarrow1$ to the established bound by Winter \cite{Winter-AlickiFannes-2016} and thereby recover the best-known results. For the other quantities, we achieve similar bounds. Moreover, it is noteworthy that both the operator space and mixed method converge to a similar limit, which is markedly superior to that of the almost additive method, as the latter does not vanish for $\varepsilon\rightarrow0$. Furthermore, the almost additive and mixed approach applied to the first entry of the divergence reduce for $\alpha\rightarrow 1$ to the bound proven in \cite{Bluhm2022ContinuityBounds}.

Thus, our three new approaches have allowed us to prove good continuity bounds for many quantities of interest related to sandwiched R{\'e}nyi divergences. Additionally, the ALAFF method allows us to derive slightly worse, albeit more general, continuity bounds with applications in the context of approximate quantum Markov chains.

\section{Preliminaries} \label{sec:preliminaries}
\subsection{Notation and basic concepts}

We start by fixing the notation used in this paper. We denote by $\mathcal H$ a complex Hilbert space with an inner product linear in the second argument. Throughout this paper $\cH$ is finite-dimensional with dimension $d \in \mathbb N$. For a bipartite or tripartite system, we will always use indices with capital letters to refer to the different subsystems. If we have, for example, the bipartite space $\mathcal H = \mathcal H_A \otimes \mathcal H_B$ and consider the dimension of $\mathcal H_A$, we write $d_A$.\par
The set of bounded linear operators on the Hilbert space $\mathcal H$ will be denoted by $\mathcal B(\mathcal H)$ and its convex subset of positive semi-definite operators with trace one, i.e.~the quantum states, by $\mathcal S(\mathcal H)$. 
\par 
We use $\tr[\,\cdot\,]$ for the usual matrix trace and $\norm{\,\cdot\,}_1$ and $\norm{\,\cdot\,}_\infty$ to denote the trace norm and the spectral norm on $\mathcal B(\mathcal H)$, respectively. More generally, we set $\| X\|_p :=\tr[ |X|^p]^{1/p}$, which coincides with the Schatten $p$-norms for $p \in [1,\infty]$.
\par 
Moreover, for a state $\rho$ on a bipartite system $\mathcal H_{AB} = \mathcal H_A \otimes \mathcal H_B$, we set $\rho_A\in\cB(\cH_A)$ to be the output-state of the partial trace. The partial trace is a completely positive trace-preserving (CPTP) map. Finally, we denote by $\mathds{1}_A$ the identity matrix on $A$ and, with a slight abuse of notation, we denote by $\tr_A[\cdot]$ both the partial trace of $A$ as well as the corresponding map on $\mathcal H_{AB}$ by tensorizing with $\mathds{1}_A$. In the first case, we mean just the usual definition of the partial trace as a map from $\cB(\cH_{AB})$ to $\cB(\cH_{B})$ while in the second case we mean $\1_A \otimes \tr_A[\cdot]$

Throughout the text, we will use the logarithm in natural basis and will denote it by $\log$.

\subsection{Divergences and entropic quantities}\label{subsec:div-entropic-quantities}
In this section, we will introduce the entropic quantities considered in the present article. We start with the sandwiched R{\'e}nyi divergence, which is the base for all subsequently defined quantities.
\begin{definition}[Sandwiched Rényi divergence]
    Let $X$, $Y \in \mathcal B(\mathcal H)$ be positive semi-definite operators with $\ker X \supseteq \ker Y$. For $\alpha \in [1/2, 1) \cup (1, \infty)$, we define
    \begin{equation}
        \Qalpha(X \Vert Y) := \tr[(Y^{\frac{1 - \alpha}{2\alpha}} X Y^{\frac{1 - \alpha}{2\alpha}})^\alpha] = \| Y^{\frac{1 - \alpha}{2\alpha}} X Y^{\frac{1 - \alpha}{2\alpha}}\|_\alpha^\alpha.
    \end{equation}
    In case the power in $Y^{\frac{1 - \alpha}{2\alpha}}$ becomes negative these quantities have to be understood as positive powers of the pseudo inverse of $Y$. Then, the \emph{sandwiched Rényi divergence} of $X$ and $Y$ is 
    \begin{equation}
        \Dalpha(X \Vert Y) := \frac{1}{\alpha - 1}\log \Qalpha(X \Vert Y) \, .
    \end{equation}
\end{definition}

\begin{remark} \label{rem:change-order-in-norm}
    Alternatively, we can write 
    \begin{equation}
        \Qalpha(X \Vert Y) = \tr[(X^{\frac{1}{2}}Y^{\frac{1 - \alpha}{\alpha}}X^{\frac{1}{2}})^\alpha] = \| X^{\frac{1}{2}}Y^{\frac{1 - \alpha}{\alpha}}X^{\frac{1}{2}}\|_\alpha^\alpha,
    \end{equation}
    because the operator $f(A^\ast A)$ for $f(0) = 0$ is defined by applying the continuous real function $f$ on the singular values of $A$. The simple trick $A^\ast Av=\lambda v\implies AA^\ast Av=\lambda Av$ shows that the singular values of $A^\ast A$ and $AA^\ast$ are equal such that $\tr[f(A^\ast A)]=\tr[f(AA^\ast)]$.
\end{remark}
It is known that the sandwiched R{\'e}nyi divergences converge in the limits $\alpha\rightarrow1$ and $\alpha\rightarrow\infty$ to well-known quantities \cite[Section 4.3.2]{Tomamichel_2016}:
\begin{proposition}\label{prop:limits}
    Let $\rho$, $\sigma \in \mathcal S(\mathcal H)$. Then, $\widetilde D_\alpha(\rho \Vert \sigma)$ converges to
    \begin{itemize}
        \item the relative entropy $D(\rho\Vert\sigma):= \tr[\rho(\log \rho - \log \sigma )]$ for $\alpha \to 1$.
        \item the max-entropy $D_{\infty}(\rho \| \sigma):= \inf \{ \lambda >0 \, : \, \rho \leq e^\lambda \sigma \} \,$ (see also Eq.\ \eqref{definition:max-divergence})  for $\alpha \to \infty$. 
    \end{itemize}
    For $\alpha = \frac{1}{2}$, it holds that $\widetilde D_{1/2}(\rho \Vert \sigma) = - \log F(\rho, \sigma)$, where $F$ is the fidelity $F(\rho, \sigma)=\left(\tr[|\sqrt{\rho} \sqrt{\sigma}|]\right)^2$.
\end{proposition}
Next, we define the sandwiched R{\'e}nyi conditional entropy in the same spirit as the conditional entropy in terms of the relative entropy.
\begin{definition}[Sandwiched R{\'e}nyi conditional entropy]\label{definition:conditional-sandwiched-renyi-divergences}
    Let $\rho_{AB} \in \cS(\cH_A \otimes \cH_B)$. Then, for $\alpha \in [1/2, 1) \cup (1, \infty)$, the \emph{sandwiched R{\'e}nyi conditional entropy} is given by 
    \begin{equation}
        \widetilde H^{\uparrow}_\alpha(A|B)_\rho := \sup\limits_{\sigma_B \in \cS(\cH_B)} \frac{1}{1 - \alpha} \log \widetilde Q_\alpha(\rho_{AB} \Vert \1_A \otimes \sigma_B) \, .
    \end{equation}
\end{definition}
Mimicking how the mutual information arises from the relative entropy, we arrive at the sandwiched R{\'e}nyi mutual information:
\begin{definition}[Sandwiched R{\'e}nyi mutual information]
    Let $\rho_{AB} \in \mathcal S(\mathcal H_A \otimes \mathcal H_B)$. Then, for $\alpha \in [1/2, 1) \cup (1, \infty)$, we define
    the \emph{sandwiched R{\'e}nyi mutual information} as
\begin{equation}
    \alphaI(A:B)_\rho := \inf\limits_{\sigma_A, \sigma_B} \widetilde D_\alpha(\rho_{AB} \Vert \sigma_A \otimes \sigma_B),
\end{equation}
where the infimum is taken over $\sigma_A \in \mathcal S(\mathcal H_A)$ and $\sigma_B \in \mathcal S(\mathcal H_B)$.
\end{definition}

Finally, we define the sandwiched R{\'e}nyi conditional mutual information in the same spirit as the conditional mutual information in terms of the relative entropy. In this case, we base our definition on the difference between (sandwiched R{\'e}nyi) conditional entropies.

\begin{definition}[Sandwiched R{\'e}nyi conditional mutual information]\label{definition:sandwiched-renyi-CMI}
    Let $\rho_{ABC} \in \cS(\cH_A \otimes \cH_B \otimes \cH_C)$. Then, for $\alpha \in [1/2, 1) \cup (1, \infty)$, the \emph{sandwiched R{\'e}nyi conditional mutual information} is given by 
    \begin{equation}
        \widetilde I^{\uparrow}_\alpha(A:C|B)_\rho := \widetilde H^{\uparrow}_\alpha(C|B)_\rho - \widetilde H^{\uparrow}_\alpha(C|AB)_\rho \, .
    \end{equation}
\end{definition}
Note that the infimum in the definition of the mutual information and the conditional entropy is attained at the reduced state of $\rho_{AB}$ (see, e.g., \cite[Section 5.1]{Tomamichel_2016}). We will often use that the sandwiched R{\'e}nyi conditional entropy is bounded by \cite[Lemma 5.2]{Tomamichel_2016}
\begin{equation} \label{eq:bound-sandwiched-cond-renyi}
    - \log \min\{d_A, d_B\} \leq \widetilde H^\uparrow_\alpha(A|B)_\rho \leq \log d_A.
\end{equation}

As already observed, the entropic quantities defined for the sandwiched R{\'e}nyi divergences converge for $\alpha \to 1$ to the ones defined by the Umegaki relative entropy. We recall its definition:
\begin{definition}[Umegaki relative entropy]
     Let $\rho$, $\sigma \in \mathcal S(\mathcal H)$ be quantum states with $\ker{\sigma} \subseteq \ker{\rho}$. Then, the \emph{Umegaki relative entropy} of $\rho$ and $\sigma$ is given by 
    \begin{equation}
         D(\rho \| \sigma):= \tr[\rho \log \rho - \rho \log \sigma].
    \end{equation}
\end{definition}
We can furthermore define the quantum conditional entropy
\begin{definition}[Quantum conditional entropy]
Let $\rho_{AB} \in \cS(\cH_A \otimes \cH_B)$. Then, the \emph{quantum conditional entropy} is given by 
    \begin{equation}
         H(A|B)_\rho := \sup\limits_{\sigma_B \in \cS(\cH_B)} -D(\rho_{AB}\Vert \1_A \otimes \sigma_B) = -  D(\rho_{AB}\Vert \1_A \otimes \rho_B)  \, .
    \end{equation}
\end{definition}
and the quantum mutual information
\begin{definition}[Quantum mutual information]\label{def:quantum-mutual-information}
Let $\rho_{AB} \in \cS(\cH_A \otimes \cH_B)$. Then, the \emph{quantum mutual information} is given by 
    \begin{equation}
         I(A:B)_\rho := \inf\limits_{\sigma_A,\sigma_B} D(\rho_{AB}\Vert \sigma_A \otimes \sigma_B) =  D(\rho_{AB}\Vert \rho_A \otimes \rho_B)  \, ,
    \end{equation}
    where the infimum is taken over $\sigma_A \in \mathcal S(\mathcal H_A)$ and $\sigma_B \in \mathcal S(\mathcal H_B)$.
    \end{definition}

Finally, we consider explicitly the limit of $\alpha \to \infty$ of the sandwiched R{\'e}nyi divergences:
\begin{definition}[Max-divergence]\label{definition:max-divergence}
     Let $X$, $Y \in \mathcal B(\mathcal H)$ be positive semi-definite operators with $X \neq 0$. Then, the \emph{max-divergence} of $X$ and $Y$ is given by 
    \begin{equation}
         D_{\infty}(X \| Y):= \text{inf} \{ \lambda >0 \, : \, X \leq \operatorname{e}^\lambda Y \} \, .
    \end{equation}
    Note that it admits the following equivalent representation:
    \begin{equation}
         D_{\infty}(X \| Y):= \log \norm{Y^{-1/2} X Y^{-1/2}}_\infty \, .
    \end{equation}
\end{definition}

\begin{definition}[Min-conditional entropy]\label{definition:max-conditional-divergence}
Let $\rho_{AB} \in \cS(\cH_A \otimes \cH_B)$. Then, the \emph{min-conditional entropy} is given by 
    \begin{equation}
         H^{\uparrow}_{\infty}(A|B)_\rho := \sup\limits_{\sigma_B \in \cS(\cH_B)}  - D_\infty(\rho_{AB}\Vert \1_A \otimes \sigma_B)  \, .
    \end{equation}
\end{definition}

Likewise, we can define a max-mutual information:
\begin{definition}[Max-mutual information]
    Let $\rho_{AB} \in \mathcal S(\mathcal H_A \otimes \mathcal H_B)$. Then, we define
    the \emph{max-mutual information} as
\begin{equation}
    I^{\uparrow}_\infty(A:B)_\rho := \inf\limits_{\sigma_A, \sigma_B} D_\infty(\rho_{AB} \Vert \sigma_A \otimes \sigma_B),
\end{equation}
where the infimum is taken over $\sigma_A \in \mathcal S(\mathcal H_A)$ and $\sigma_B \in \mathcal S(\mathcal H_B)$.
\end{definition}

It has been shown that $\alpha \to \widetilde D_\alpha(\rho \Vert \sigma)$ is monotonically increasing \cite[Corollary 4.2]{Tomamichel_2016}. Thus, in particular,
\begin{equation} \label{eq:max-upper-bound-on-sandwiched}
    \widetilde D_\alpha(\rho \Vert \sigma) \leq D_{\infty}(\rho\Vert\sigma) \qquad \forall \alpha \in [1/2,1) \cup (1,\infty).
\end{equation}

\section{Technical tools and main theorems}\label{sec:technical_tools}
In this section, we will introduce the technical tools and prove the main theorems that form the cornerstones for our proofs of the continuity bounds in \Cref{sec:CB_Renyi}, which are just corollaries of the former.
    
\subsection{Almost additive approach}

We start by reviewing some tools that are related to the almost additive approach to sandwiched R{\'e}nyi divergences (see \cite{muller2013quantum, Tomamichel_2016} for an overview). One property is that $\Qalpha$ is jointly convex/ concave:

\begin{lemma}[Joint convexity and concavity of $\Qalpha$]\label{lemma:Q-alpha-joint-convexity/concavity}
    For quantum states $\rho$, $\sigma \in \mathcal S(\mathcal H)$, the function  
    \begin{equation}
        (\rho, \sigma) \mapsto \Qalpha(\rho \Vert \sigma)
    \end{equation}
    is jointly concave for $\alpha \in [1/2, 1)$ and jointly convex for $\alpha \in (1, \infty)$.
\end{lemma}
\begin{proof}
    For a proof, see \cite[Proposition 4.7 and Theorem 4.1]{Tomamichel_2016}. 
\end{proof}

Another useful property of the $\Qalpha$ is that it behaves nicely under addition, unlike the sandwiched R{\'e}nyi divergences:
\begin{lemma}[Sub- and Superadditivity of $\Qalpha$]\label{lemma:Q-alpha-super/subadditivity}
    For $X_1$, $X_2$, $Y \in \mathcal B(\mathcal H)$ positive semi-definite with $\ker Y \subseteq \ker X_1 \cap \ker X_2$, we find that for $\alpha \in (0, 1)$
    \begin{equation}
        \Qalpha(X_1 + X_2 \Vert Y) \le \Qalpha(X_1 \Vert Y) + \Qalpha(X_2 \Vert Y)
    \end{equation}
    and further for $\alpha \in (1, \infty)$
    \begin{equation}
        \Qalpha(X_1 \Vert Y) + \Qalpha(X_2 \Vert Y) \le \Qalpha(X_1 + X_2 \Vert Y) \, .
    \end{equation}
\end{lemma}
\begin{proof}
    The proof of the first claim can be found in \cite{marwah2022}, based on an inequality from \cite{mccarthy1967cp}. For the second, we use that one can write 
    \begin{align}
        \Qalpha(X_1 + X_2 \Vert Y) = \tr[(X_1^\prime + X_2^\prime)^\alpha] = \norm{(X^\prime_1 + X^\prime_2)^\alpha}_1 \, ,
    \end{align}
    with $X^\prime_{i} := Y^{\frac{1 - \alpha}{2\alpha}} X_{i} Y^{\frac{1 - \alpha}{2\alpha}}$, $i = 1, 2$. As before $Y^{-1}$ is the pseudo inverse of $Y$. Since $\norm{\cdot}_1$ is unitarily invariant and since the map $\R_+\ni x \mapsto x^\alpha$ is convex for $\alpha > 1$ and vanishes at zero, we can apply \cite[Theorem 1.2]{Bourin2007} to conclude
    \begin{equation}
        \norm{(X_1^\prime + X_2^\prime)^\alpha}_1 \ge \norm{(X_1^\prime)^\alpha + (X_2^\prime)^\alpha}_1 = \tr[(X_1^\prime)^\alpha + (X_2^\prime)^\alpha] = \Qalpha(X_1 \Vert Y) + \Qalpha(X_2 \Vert Y) \, . 
    \end{equation}
    The trace and $1$-norm agree as all involved operators are positive semi-definite.
\end{proof}

One might be tempted to conjecture that super- and subadditivity holds more generally for the analogues of the $\Qalpha$, in the case of Petz or geometric Renyi divergences. However, this is not the case. One can relatively easily construct counterexamples:
\begin{example}
    In the following, we present an example which contradicts the superadditivity of the Petz and the geometric Rényi divergence. For $\rho$, $\tau \in \mathcal S(\mathcal H)$ with $\ker{\rho} \supseteq \ker{\tau}$ and $\alpha \in (0,1) \cup (1, \infty)$,
    \begin{equation}
        \overline{Q}_{\alpha}(\rho\Vert\tau) := \tr[\rho^{\alpha}\tau^{1-\alpha}] \; , \quad \overline{D}_{\alpha}(\rho\Vert\tau) := \frac{1}{\alpha -1} \log \overline{Q}_{\alpha}(\rho\Vert\tau)
    \end{equation}
    is the Petz Rényi divergence and for $\alpha \in (1,2]$, let
        \begin{equation}
        \widehat{Q}_{\alpha}(\rho\Vert\tau) := \tr[\tau^{\frac{1}{2}}(\tau^{-\frac{1}{2}}\rho \tau^{-\frac{1}{2}})^{\alpha} \tau^{\frac{1}{2}}] \; , \quad \widehat{D}_{\alpha}(\rho\Vert\tau) := \frac{1}{\alpha -1} \log \widehat{Q}_{\alpha}(\rho\Vert\tau) \, 
    \end{equation}
    be the geometric (maximal) Rényi divergence. For the density matrices
    \begin{equation}
        \rho_1=\begin{pmatrix}0.8 & 0.3\\0.3 & 0.2 \end{pmatrix},\quad\rho_2=\begin{pmatrix}0.1 & 0.2\\0.2 & 0.9\end{pmatrix}, \quad\tau=\begin{pmatrix}0.45 & 0.49\\0.49 & 0.55\end{pmatrix},
    \end{equation}
    one can calculate
    \begin{equation}
        \overline{Q}_{1.5}(\rho_1\Vert\tau)+\overline{Q}_{1.5}(\rho_2\Vert\tau)>6>5.9>\overline{Q}_{1.5}(\rho_1+\rho_2\Vert\tau)
    \end{equation}
    and
    \begin{equation}
        \widehat{Q}_{1.5}(\rho_1\Vert\tau)+\widehat{Q}_{1.5}(\rho_2\Vert\tau)>9>6>\widehat{Q}_{1.5}(\rho_1+\rho_2\Vert\tau)
    \end{equation}
    which contradicts the superadditivity.
\end{example}

We conclude the subsection on the almost additive approach by proving our fundamental technical result from which all our continuity bounds following the almost additive approach will be obtained.

\begin{theorem}[Distance to a convex, compact set]\label{theo:minimial-distance-cc-set-almost-additive-approach}
    Let $\cC \subseteq \cS(\cH)$ be a convex, compact set that contains at least one positive definite state. Then the map 
    \begin{equation}
        \widetilde D_{\alpha, \cC}:\cS(\cH) \to \R, \qquad \rho \mapsto \widetilde D_{\alpha, \cC}(\rho) :=  \inf\limits_{\tau \in \cC} \widetilde D_\alpha(\rho \Vert \tau)
    \end{equation}
    is uniformly continuous (cf.\ \cite[Definition 4.18]{RudinAnalysis}) for $\alpha \in [1/2, 1) \cup (1, \infty)$. For $\rho$, $\sigma \in \cS(\cH)$ with $\tfrac{1}{2}\norm{\rho - \sigma}_1 \le \varepsilon$, $\alpha \in [1/2, 1)$ and $\kappa$ (see \Cref{rem:kappa}) such that $\sup\limits_{\rho \in \cS(\cH)} \widetilde D_{\alpha, \cC}(\rho)\leq\log(\kappa)<\infty$ we get
    \begin{equation}
        |\widetilde D_{\alpha, \cC}(\rho) - \widetilde D_{\alpha, \cC}(\sigma)| \le \log(1 + \varepsilon) + \frac{1}{1 - \alpha}\log(1 + \varepsilon^\alpha \kappa^{1 - \alpha} - \frac{\varepsilon}{(1 + \varepsilon)^{1 - \alpha}})
    \end{equation}
    Further for $\alpha \in (1, \infty)$ and $\kappa$ (see \Cref{rem:kappa}) such that $\sup\limits_{\rho \in \cS(\cH)} \widetilde D_{\alpha, \cC}(\rho)\leq\log(\kappa)<\infty$
    \begin{equation}
         |\widetilde D_{\alpha, \cC}(\rho) - \widetilde D_{\alpha, \cC}(\sigma)| \le \log(1 + \varepsilon) + \frac{1}{\alpha - 1}\log(1 + \varepsilon \kappa^{\alpha - 1} - \frac{\varepsilon^\alpha}{(1 + \varepsilon)^{\alpha - 1}}) \, . 
    \end{equation}
\end{theorem}
\begin{proof}
    We will only demonstrate here the proof for the second inequality ($\alpha>1$), as the proof for the first ($\alpha<1$) is almost completely analogous. The proof is inspired by \cite{marwah2022}. 
    Without loss of generality, we can assume that $\tfrac{1}{2}\norm{\rho - \sigma}_1 = \varepsilon$, as the bound is monotone in $\varepsilon$. We have
    \begin{equation}
        |\widetilde D_{\alpha, \cC}(\rho) - \widetilde D_{\alpha, \cC}(\sigma)| = \frac{1}{\alpha - 1} \left|\log \frac{\inf\limits_{\tau \in \cC} \alphaQ(\rho \Vert \tau)}{\inf\limits_{\tau \in \cC} \alphaQ(\sigma \Vert \tau)}\right|
    \end{equation}
    using the monotonicity of the logarithm. Hence the result reduces to an upper bound on
    \begin{equation}
        \frac{\inf\limits_{\tau \in \cC} \alphaQ(\rho \Vert \tau)}{\inf\limits_{\tau \in \cC} \alphaQ(\sigma \Vert \tau)} \quad \text{and} \quad \frac{\inf\limits_{\tau \in \cC} \alphaQ(\sigma \Vert \tau)}{\inf\limits_{\tau \in \cC} \alphaQ(\rho \Vert \tau)} \, . 
    \end{equation}
    We only upper bound the first fraction as the bound on the second is achieved by swapping the roles of $\rho$ and $\sigma$. First, there are orthonormal quantum states $\mu, \nu$ such that $\rho + \varepsilon\mu = \sigma + \varepsilon\nu$. Using the superadditivity of $\alphaQ$ (\Cref{lemma:Q-alpha-super/subadditivity}) and the $\alpha$-homogenity in the first argument gives
    \begin{equation}
        \widetilde Q_{\alpha, \cC}(\rho) + \varepsilon^\alpha \widetilde Q_{\alpha, \cC}(\mu) \le \inf\limits_{\tau \in \cC} \left(\alphaQ(\rho \Vert \tau) + \varepsilon^\alpha\alphaQ(\mu \Vert \tau)\right) \le \widetilde Q_{\alpha, \cC}(\rho + \varepsilon \mu)
    \end{equation}
    where we set $\widetilde Q_{\alpha, \cC}(X) := \inf\limits_{\tau \in \cC} \widetilde Q_\alpha(X \Vert \tau)$ for a positive semi-definite operator $X$. Joint convexity of the $\alphaQ$ and their $\alpha$-homogenity further allows us to write
    \begin{align}
        \widetilde Q_{\alpha, \cC}(\sigma + \varepsilon\nu) &= (1 + \varepsilon)^{\alpha} \inf\limits_{\tau_1, \tau_2 \in \cC} \alphaQ\Big(\tfrac{1}{1 + \varepsilon} \sigma + \tfrac{\varepsilon}{1 + \varepsilon}\nu \big\Vert \tfrac{1}{1 + \varepsilon} \tau_1 + \tfrac{\varepsilon}{1 + \varepsilon} \tau_2\Big) \\
        &\le (1 + \varepsilon)^{\alpha - 1} \left(\inf\limits_{\tau_1\in \cC} \alphaQ(\sigma \Vert \tau_1) + \inf\limits_{\tau_2 \in \cC}\varepsilon\alphaQ(\nu\Vert\tau_2)\right)\\
        &= (1 + \varepsilon)^{\alpha - 1}\left(\widetilde Q_{\alpha, \cC}(\sigma) + \varepsilon \widetilde Q_{\alpha, \cC}(\nu)\right)
    \end{align}
    where we split the infimum in the second argument into an equivalent infimum over a convex combination $\frac{1}{1 + \varepsilon} \tau_1 + \frac{\varepsilon}{1 + \varepsilon} \tau_2$ which allows us to split the infimum later. Moreover, it holds that $1 \le \widetilde Q_{\alpha, \cC}(\nu) \le \kappa^{\alpha - 1}$ for any state $\nu$, where the lower bound stems from the non-negativity of the sandwiched $\alpha$-R{\'e}nyi divergences on quantum states and the upper bound holds by assumption. Putting these estimates together, using $\rho + \varepsilon\mu = \sigma + \varepsilon\nu$, we find
    \begin{equation}
        \widetilde Q_{\alpha, \cC}(\rho) \le (1 + \varepsilon)^{\alpha - 1}\left( \widetilde Q_{\alpha, \cC}(\sigma) + \varepsilon \kappa^{\alpha - 1} - \frac{\varepsilon^\alpha}{(1 + \varepsilon)^{\alpha - 1}}\right) \, . 
    \end{equation}
    We therefore obtain
    \begin{align}
        \frac{\widetilde Q_{\alpha, \cC}(\rho)}{\widetilde Q_{\alpha, \cC}(\sigma)} \le (1 + \varepsilon)^{\alpha - 1}\left(1 + \frac{\varepsilon\kappa^{\alpha - 1} - \frac{\varepsilon^\alpha}{(1 + \varepsilon)^{\alpha - 1}}}{\widetilde Q_{\alpha, \cC}(\sigma)}\right)
        \le (1 + \varepsilon)^{\alpha - 1}\left(1 + \varepsilon\kappa^{\alpha - 1} - \frac{\varepsilon^\alpha}{(1 + \varepsilon)^{\alpha - 1}}\right)
    \end{align}
    where in the second inequality we lower bound $\widetilde Q_{\alpha, \cC}(\sigma)$ by $1$. This is valid, since for $\varepsilon \in (0, 1)$, $\varepsilon \kappa^{\alpha - 1} - \frac{\varepsilon^\alpha}{(1 + \varepsilon)^{\alpha - 1}} \ge \varepsilon - \frac{\varepsilon^\alpha}{(1 + \varepsilon)^{\alpha - 1}} \ge 0$. Repeating the procedure for the other fraction gives the same bound and hence concludes the claim. 
\end{proof}

\begin{remark}[Existence of (uniform) $\kappa$]\label{rem:kappa}
    A $\kappa < \infty$ upper bound on $\widetilde D_{\alpha, \cC}(\cdot)$ and uniform in all $\alpha \in [1/2, 1) \cup (1, \infty)$ always exists, since we can upper bound $\widetilde D_{\alpha, \cC}(\cdot) \le \log \norm{\tau^{-1}}_\infty$ independently of $\alpha$, where $\tau$ is a full-rank state in $\cC$ (which by assumption exists). For the continuity bounds the $\kappa$ can, however, also depend on $\alpha$.
\end{remark}

Before considering the limit cases, we prove that limits can be exchanged with the infimum in the above definition of $\widetilde D_{\alpha, \cC}(\cdot)$:
\begin{lemma}\label{lem:limit-change}
    Let $\cC \subseteq \cS(\cH)$ be a compact set that contains at least one positive definite state. Then the following identities hold:
    \begin{equation}
        \lim_{\alpha\rightarrow 1}\widetilde D_{\alpha, \cC}(\rho)=\inf\limits_{\tau \in \cC} \lim_{\alpha\rightarrow 1}\widetilde D_\alpha(\rho \Vert \tau)=\inf\limits_{\tau \in \cC} D(\rho \Vert \tau)\eqqcolon\widetilde D_{1, \cC}(\rho)
    \end{equation}
    and 
    \begin{equation}
        \lim_{\alpha\rightarrow \infty}\widetilde D_{\alpha, \cC}(\rho)=\inf\limits_{\tau \in \cC} \lim_{\alpha\rightarrow \infty}\widetilde D_\alpha(\rho \Vert \tau)=\inf\limits_{\tau \in \cC} D_\infty(\rho \Vert \tau)\eqqcolon\widetilde D_{\infty, \cC}(\rho)\,.
    \end{equation}
\end{lemma}
\begin{proof}
Let $\eta\in\cC$ be positive definite, which exists by assumption. We note first that the infimum $\inf_{\tau \in \mathcal C}\widetilde D_\alpha(\rho \Vert \tau)$ is attained \cite[Theorem 2.43]{Aliprantis2006} since $\mathcal C$ is compact and $\sigma \mapsto \widetilde D(\rho\Vert\sigma)$ is lower semi-continuous for any fixed $\rho \in \cS(\cH)$ \cite[Proposition 4.5]{Hiai2022}.
Then, 
\begin{equation}
    \widetilde D_\infty(\rho \Vert \eta) =: c_\infty < \infty.
\end{equation}
Moreover, 
\begin{equation}
    \inf_{\tau \in \mathcal C}\widetilde D_\alpha(\rho \Vert \tau) \leq \widetilde D_\alpha(\rho \Vert \eta) \leq c_\infty
\end{equation}
for all $\alpha\in(1,\infty]$, since the sandwiched Rényi divergences are monotonically increasing in $\alpha$ \cite[Corollary 4.2]{Tomamichel_2016}. Next, we define $\cC(\rho)=\{\tau\in\cC\,|\,\widetilde D_\infty(\rho \Vert \tau)\leq c_\infty\}$ which satisfies by the above
    \begin{equation}
        \inf\limits_{\tau \in \cC} \widetilde D_\alpha(\rho \Vert \tau)=\inf\limits_{\tau \in \cC(\rho)} \widetilde D_\alpha(\rho \Vert \tau)
    \end{equation}
for all $\alpha\in(1,\infty]$. The set $\cC (\rho)$ is compact, because the preimage of $(-\infty, c_\infty]$ under a lower semi-continuous function is closed. Hence, also the infimum on the right hand side is attained. Moreover, the function $\cC (\rho)\ni\tau\mapsto D_\alpha(\rho \Vert \tau)$ is continuous for all $\alpha\in(1,\infty]$. Therefore, Dini's theorem \cite[Theorem 2.66]{Aliprantis2006} shows that $\cC (\rho) \ni\tau\mapsto D_\alpha(\rho \Vert \tau)$ converges uniformly for $\alpha \to 1$, $\alpha \to \infty$, since we have pointwise convergence by \cite[Section 4.3.2]{Tomamichel_2016} (see also \Cref{prop:limits}). Thus, the assertion follows from \cite[Lemma I.7.6]{dunford1958linear}.

\end{proof}

\begin{lemma}[Limits]\label{lem:limit-almost-additive}
    Let $\rho, \sigma \in \cS(\cH)$ with $\frac{1}{2}\norm{\rho - \sigma}_1 \le \varepsilon$ and $\kappa$ a bound on $\widetilde D_{\alpha, \cC}$ independent of $\alpha$ (see Remark \ref{rem:kappa}), then the limit $\alpha \to 1$ of the bounds obtained in \Cref{theo:minimial-distance-cc-set-almost-additive-approach} gives 
    \begin{equation}
        |\widetilde D_{1, \cC}(\rho) - \widetilde D_{1, \cC}(\sigma)| \le \varepsilon\log \kappa \,  + (1 + \varepsilon) h\Big(\frac{\varepsilon}{1 + \varepsilon}\Big)
    \end{equation}
    where $h(\cdot)$ is the binary entropy. Unless $\varepsilon$ is trivial, i.e.~$\varepsilon = 0$, we find for $\alpha \to \infty$
    \begin{equation}
        |\widetilde D_{\infty, \cC}(\rho) - \widetilde D_{\infty, \cC}(\sigma)| \le \log(1 + \varepsilon) + \log \kappa
    \end{equation}
    which is no longer a continuity bound
\end{lemma}
\begin{proof}
    Using l'Hospital's rule, we find that
    \begin{align}
        &\lim_{\alpha\,\nearrow\,1}\left[\frac{1}{1-\alpha}\log(1 + \varepsilon^\alpha \kappa^{(1-\alpha)}- \frac{\varepsilon}{(1 + \varepsilon)^{1-\alpha}}) \right]\\
        &=-\lim_{\alpha\,\nearrow\,1}\left[\frac{\log(\varepsilon)\varepsilon^\alpha \kappa^{(1-\alpha)}-\varepsilon^\alpha \kappa^{(1-\alpha)} \log{\kappa}-\log(1+\varepsilon)\frac{\varepsilon}{(1 + \varepsilon)^{1-\alpha}}}{1 + \varepsilon^\alpha \kappa^{(1-\alpha)}- \frac{\varepsilon}{(1 + \varepsilon)^{1-\alpha}}} \right] \\
        &=-\varepsilon\log(\varepsilon)+ \varepsilon\log{\kappa} + \varepsilon\log(1 + \varepsilon)
    \end{align}
    and similarly
    \begin{align}
        &\lim_{\alpha\,\searrow\,1}\left[\frac{1}{\alpha - 1}\log(1 + \varepsilon \kappa^{(\alpha - 1)}- \frac{\varepsilon^\alpha}{(1 + \varepsilon)^{\alpha - 1}}) \right]\\
        &= \lim_{\alpha\,\searrow\,1}\left[\frac{\varepsilon \kappa^{(\alpha - 1)} \log{\kappa} - \log(\varepsilon)\frac{\varepsilon^\alpha}{(1 + \varepsilon)^{\alpha - 1}}+\log(1+\varepsilon)\frac{\varepsilon^\alpha}{(1 + \varepsilon)^{\alpha - 1}}}{1 + \varepsilon \kappa^{(\alpha - 1)}- \frac{\varepsilon^\alpha}{(1 + \varepsilon)^{\alpha - 1}}} \right] \\
        &= \varepsilon\log{\kappa} - \varepsilon\log(\varepsilon) + \varepsilon \log(1+\varepsilon)
    \end{align}
    which proves the limit $\alpha\rightarrow1$. Another use of l'Hospital's rule shows that 
    \begin{equation}
        \lim_{\alpha \to \infty}\left[\frac{1}{\alpha - 1}\log(1 + \varepsilon \kappa^{(\alpha - 1)}- \frac{\varepsilon^\alpha}{(1 + \varepsilon)^{\alpha - 1}}) \right]= \log{\kappa}.
    \end{equation}
\end{proof}

\subsection{Operator spaces approach}
In this section, we will construct a family of norms inspired by the norms on interpolation spaces over von Neumann algebras defined in \cite[Theorem 4.5]{Beigi2022}. These norms have an explicit characterization given by a supremum of amalgamations with elements from a different von Neumann subalgebra in a Schatten p-norm. In this section, we show that in the finite-dimensional setting, this construction is not limited to amalgamation with a von Neumann subalgebra, but can be generalised to a convex compact set of positive semi-definite operators, which contains at least one positive definite state. We further give a direct proof of a triangle-like inequality for a map arising from these norms, previously only shown in an abstract and more restrictive setting. To be more precise, our goal is to show that for positive semi-definite operators the map 
\begin{equation}
    X \mapsto \inf\limits_{c \in \cC, c > 0} \norm{c^{-\frac{1}{2r}} X c^{-\frac{1}{2r}}}_p
\end{equation}
satisfies a triangle inequality and some other properties that come in handy when proving continuity bounds later. Here $r$ is implicitly constrained by the explicit $p$ and an implicit $q$ satisfying $p \ge q \ge 1$. The choice to not use $r$ as a defining parameter will become clear later when we establish duality relations and $p$ and $q$ transform to their Hölder conjugates while $r$ is unaffected.  $\cC$ is a convex, compact set of positive semi-definite operators containing at least one positive definite element.

We begin with the construction of the auxiliary norms.

\begin{definition}[The $\cC, p, q$ norm]\label{def:Kpr-norm}
    Let $\cC \subset \cB_{\ge 0}(\cH)$ be a convex, compact set containing at least one positive definite state. Then for $1 \le p \le q \le \infty$, $\frac{1}{r} := \frac{1}{p} - \frac{1}{q}$, we define
    \begin{equation}
       \norm{\cdot}_{\cC, p, q}:\cB(\cH) \to [0, \infty), \quad  X \mapsto \norm{X}_{\cC, p, q} := \sup\limits_{c \in \cC} \norm{c^{\frac{1}{2r}} X c^{\frac{1}{2r}}}_{p}
    \end{equation}
    where $\norm{\cdot}_p = (\tr[|\cdot|^p])^{\frac{1}{p}}$ is the Schatten-$p$-norm. 
\end{definition}

\begin{lemma}
    $\norm{\cdot}_{\cC, p, q}:\cB(\cH) \to [0, \infty)$ defines a norm on $\cB(\cH)$.
\end{lemma}
\begin{proof}
    The map is finite for all $X \in \cB(\cH)$, because $\cC$ is compact and $c \mapsto c^{\frac{1}{2r}}X c^{\frac{1}{2r}}$ continuous on positive semi-definite matrices. Further, it is positive definite since $\cC$ contains a positive-definite state, by assumption. Finally, it satisfies positive homogeneity and triangle inequality because the Schatten-p-norm has this property and the supremum is subadditive.
\end{proof}

We will proceed to define the maps we are interested in:
\begin{definition}\label{def:Kp'r-maps}
    Let $\cC \subset \cB_{\ge 0}(\cH)$ a convex, compact set containing at least one positive definite state. Then for $1 \le q' \le p' \le \infty$ and $\frac{1}{r} = \frac{1}{q'} - \frac{1}{p'}$, we define
    \begin{equation}
       \norm{\cdot}_{\cC, p', q'}^*:\cB(\cH) \to [0, \infty), \quad  X \mapsto \norm{X}_{\cC, p', q'}^* := \inf\limits_{c \in \cC, c > 0} \norm{c^{-\frac{1}{2r}} X c^{-\frac{1}{2r}}}_{p'}
    \end{equation}
    where $X \in \cB(\cH)$ and $\norm{\cdot}_{p'} = \left(\tr[|\cdot|^{p'}]\right)^{\frac{1}{p'}}$. 
\end{definition}

\begin{remark}
    For the above maps, the norm properties are no longer apparent. We get that the map is finite, positive homogenous, and further positive definite (requires a bit more work but relates back to the compactness of $\cC$). However, the triangle inequality we can only prove for positive semi-definite operators. 
\end{remark}

We will now show that on the positive semi-definite operators $\norm{\cdot}_{\cC, p', q'}^*$ satisfies triangle inequality by showing that $\norm{\cdot}_{\cC, p', q'}^*$ agrees with the dual norm of $\norm{\cdot}_{\cC, p, q}$ on this set. Here $p$ is given by $1 = \frac{1}{p'} + \frac{1}{p}$ and $q$ by $1 = \frac{1}{q} + \frac{1}{q'}$. Note that the latter norm is indeed consistent with \Cref{def:Kpr-norm} (as choosing $q$ via $1 = \frac{1}{q} + \frac{1}{q'}$ leaves $r$ invariant and satisfies the requirements, i.e.~$1 \le p \le q \le \infty$). As a first step, we derive a Hölder-inequality for these maps.

\begin{lemma}[Hölder-inequality]\label{lem:hoelder-inequality}
    Let $\norm{\cdot}_{\cC, p, q}$ as defined in \Cref{def:Kpr-norm} be given, then for all $X, Y \in \cB(\cH)$
    \begin{equation}
        |\tr[X Y]| \le \norm{X}_{\cC, p, q} \norm{Y}_{\cC, p', q'}^*
    \end{equation}
    where $\norm{\cdot}_{\cC, p', q'}^*$ as defined in \Cref{def:Kp'r-maps} with $p'$ given by $\frac{1}{p} + \frac{1}{p'} = 1$ and $q'$ via $\frac{1}{q} + \frac{1}{q'} = 1$.
\end{lemma}
\begin{proof}
    For $c \in \cC$ with $c > 0$, using Hölder inequality on Schatten-$p$-norms, we have
    \begin{align}
        |\tr[XY]| &= \left|\tr[c^{\frac{1}{2r}} X c^{\frac{1}{2r}} c^{-\frac{1}{2r}} Y c^{-\frac{1}{2r}}]\right|\\
        &\le \norm{c^\frac{1}{2r} X c^{\frac{1}{2r}}}_p \norm{c^{-\frac{1}{2r}} Y c^{-\frac{1}{2r}}}_{p'} \, . 
    \end{align}
    Now using that $\norm{c^\frac{1}{2r} X c^{\frac{1}{2r}}}_p \le \norm{X}_{\cC, p, q}$ we get
    \begin{equation}
        |\tr[XY]| \le \norm{X}_{\cC, p, q} \norm{c^{-\frac{1}{2r}} Y c^{-\frac{1}{2r}}}_{p'} \, . 
    \end{equation}
    Finally taking the infimum over all $c \in \cC$, $c > 0$ of the above inequality
    \begin{equation}
         |\tr[XY]| \le \norm{X}_{\cC, p, q} \norm{Y}_{\cC, p', q'}^* \, . 
    \end{equation}
\end{proof}

We can now characterise $\norm{\cdot}_{\cC, p, q}$ via $\norm{\cdot}_{\cC, p', q'}^*$ as follows.

\begin{lemma}
     Let $\norm{\cdot}_{\cC, p, q}$ as defined in \Cref{def:Kpr-norm} be given. Then we obtain
     \begin{equation}
         \sup\limits_{Y \in \cB(\cH), \,\norm{Y}_{\cC, p', q'}^* \le 1} |\tr[X Y]| = \norm{X}_{\cC, p, q}
     \end{equation}
     where $\norm{\cdot}_{\cC, p', q'}^*$ is the map from \Cref{def:Kp'r-maps} with $p'$ defined via $\frac{1}{p} + \frac{1}{p'} = 1$ and $q'$ via $\frac{1}{q} + \frac{1}{q'} = 1$.
\end{lemma}
\begin{proof}
    Let $X \in \cB(\cH)$ be given. Due to the continuity of $c \mapsto c^{\frac{1}{2r}} X c^{\frac{1}{2r}}$ on $\cC$ and compactness of $\cC$, there exists an $c_* \in \cC$ such that 
    \begin{equation}
        \norm{X}_{\cC, p, q} = \norm{c_*^{\frac{1}{2r}} X c_*^{\frac{1}{2r}}}_p
    \end{equation}
    Due to the convexity, the existence of a positive definite state $\tau$ in $\cC$ and the continuity of $c \mapsto c^{\frac{1}{2r}} X c^{\frac{1}{2r}}$, we have that for all $\varepsilon > 0$ there exists a $\delta \in [0, 1]$ such that $c_\delta = (1 - \delta) c_* + \delta \tau \in \cC$, $c_\delta > 0$ and 
    \begin{equation}
        \norm{X}_{\cC, p, q} \le \norm{c_\delta^{\frac{1}{2r}} X c_\delta^{\frac{1}{2r}}}_p + \varepsilon \, . 
    \end{equation}
    Regarding $\norm{\cdot}_{p'}$ as dual norm to $\norm{\cdot}_p$ (see e.g.~\cite[Lemma 3.3]{Tomamichel_2016}) where $1 = \frac{1}{p'} + \frac{1}{p}$, we find
    \begin{equation}
        \begin{aligned}
            \norm{X}_{\cC, p, q} &\le \sup\limits_{W \in \cB(\cH), \norm{W}_{p'} \le 1} \left|\tr[c_\delta^{\frac{1}{2r}} X c_\delta^{\frac{1}{2r}} W]\right| + \varepsilon\\
            &= \sup\limits_{Y \in \cB(\cH), \Vert c_\delta^{-\frac{1}{2r}} Y c_\delta^{-\frac{1}{2r}}\Vert_{p'} \le 1} |\tr[X Y]| + \varepsilon\\
            &\le \sup\limits_{Y \in \cB(\cH), \norm{Y}_{\cC, p', q'}^* \le 1} |\tr[X Y]| + \varepsilon \, ,
        \end{aligned}
    \end{equation}
    where the last inequality holds, since $\Vert c_\delta^{-\frac{1}{2r}} Y c_\delta^{-\frac{1}{2r}}\Vert_{p'} \le 1$ implies $\norm{Y}_{\cC, p', q'}^* \le 1$. Now taking $\varepsilon \to 0$ gives 
    \begin{equation}
         \norm{X}_{\cC, p, q} \le \sup\limits_{Y \in \cB(\cH), \norm{Y}_{\cC, p', q'}^* \le 1} |\tr[X Y]| \, . 
    \end{equation}
    The reverse inequality is a direct consequence of \Cref{lem:hoelder-inequality}. We, hence, conclude the claim.
\end{proof} 

Now we come to the theorem which is at the heart of our continuity bounds. 
\begin{theorem}[A dual formula for $\norm{\cdot}_{\cC, p', q'}^*$]\label{theo:dual-formula}
    Let $\norm{\cdot}_{\cC, p', q'}^*$ be as defined in \Cref{def:Kp'r-maps} and $X \ge 0$. Then
    \begin{equation}
        \sup\limits_{Y \in \cB_{\ge 0}(\cH), \norm{Y}_{\cC, p, q} \le 1} |\tr[X Y]| = \norm{X}_{\cC, p', q'}^* \, ,
    \end{equation}
    for $p$ given by $\frac{1}{p} + \frac{1}{p'} = 1$, $q'$ via $\frac{1}{q} + \frac{1}{q'} = 1$ and $\norm{Y}_{\cC, p, q}$ as in \Cref{def:Kpr-norm}. That is, on positive semi-definite states, $\norm{\cdot}_{\cC, p', q'}^*$ agrees with the dual norm of $\norm{\cdot}_{\cC, p, q}$.
\end{theorem}
\begin{proof}
    Let $X \in \cB_{\ge 0}(\cH)$. Then we immediately get that 
    \begin{equation}
         \sup\limits_{Y \in \cB_{\ge 0}(\cH), \norm{Y}_{\cC, p, q} \le 1} |\tr[X Y]| \le \norm{X}_{\cC, p', q'}^*
    \end{equation}
    by \Cref{lem:hoelder-inequality}. This inequality furthermore holds for $X \in \cB(\cH)$. To prove the reverse inequality, we define the auxiliary function
    \begin{equation}
        \begin{aligned}
            f_X:\cB_{\ge 0}(\cH) \times \cC &\to \R, \\
            (Y, c) &\mapsto f_X(Y, c) = p' \tr[X Y] - \frac{p'}{p} \tr[(c^{\frac{1}{2r}} Y c^{\frac{1}{2r}})^p] \, ,
        \end{aligned}
    \end{equation}
    where $p$ is given by $\frac{1}{p} + \frac{1}{p'} = 1$. Note that the function is lower semi-continuous in $c$ and upper semi-continuous in $Y$ (as it is continuous in both of its arguments). In addition, we will now show, that it is concave in $Y$ and convex in $c$. The concavity in $Y$ is straightforward. We have that $Y \mapsto \tr[YX]$ is linear in $Y$, hence in particular concave. Further for a fixed $c$, we have that $Y \mapsto -\frac{p'}{p}\tr[(c^{\frac{1}{2r}} Y c^{\frac{1}{2r}})^p]$ is concave, since $x \mapsto - x^p$ for $p \ge 1$ is. To show convexity in $c$, we have to show that for a fixed $Y$, $c \mapsto \tr[(c^{\frac{1}{2r}} Y c^{\frac{1}{2r}})^p]$ is concave. We first rewrite
    \begin{equation}
        \tr[(c^{\frac{1}{2r}} Y c^{\frac{1}{2r}})^p] = \tr[(\sqrt{Y} c^{\frac{1}{r}} \sqrt{Y})^p]
    \end{equation}
    and then use \cite[Eq. (3.16)]{Tomamichel_2016} to obtain
    \begin{equation}
        \tr[(\sqrt{Y} c^{\frac{1}{r}} \sqrt{Y})^p] = \sup\limits_{Z \ge 0} p \tr[\sqrt{Y} c^{\frac{1}{r}} \sqrt{Y} Z] - \frac{p}{p'}\tr[Z^{p'}] = \sup\limits_{Z \ge 0} p \tr[\sqrt{Y} c^{\frac{1}{r}} \sqrt{Y} Z^{\frac{1}{p'}}] - \frac{p}{p'}\tr[Z] \, . 
    \end{equation}
    Since $0 \le \frac{1}{p'}, \frac{1}{r} \le 1$ and $\frac{1}{p'} + \frac{1}{r} = \frac{1}{p'} + \frac{1}{q'} - \frac{1}{p'} = \frac{1}{q'} \le 1$, Lieb's concavity theorem \cite{Lieb.1973} gives us joint concavity of the map 
    \begin{equation}
        (Z, c) \mapsto p \tr[\sqrt{Y} c^{\frac{1}{r}} \sqrt{Y} Z^{\frac{1}{p'}}] - \frac{p}{p'}\tr[Z]
    \end{equation}
    and therefore concavity of 
    \begin{equation}
        c \mapsto  \sup\limits_{Z \ge 0} p \tr[\sqrt{Y} c^{\frac{1}{r}} \sqrt{Y} Z^{\frac{1}{p'}}] - \frac{p}{p'}\tr[Z] = \tr[(\sqrt{Y} c^{\frac{1}{r}} \sqrt{Y})^p] =  \tr[(c^{\frac{1}{2r}} Y c^{\frac{1}{2r}})^p] \, . 
    \end{equation}
    Knowing that $f_X$ is a map from the Cartesian product of a convex set with a convex compact set to the reals, being upper semi-continuous and convex in its first and lower semi-continuous and concave in its second argument, we can employ Sions minimax theorem \cite[Theorem 2.18]{khatri2020principles} to find that 
    \begin{equation}
        \sup\limits_{Y \in \cB_{\ge 0}(\cH)} \inf\limits_{c \in \cC} f_X(Y, c) = \inf\limits_{c \in \cC} \sup\limits_{Y \in \cB_{\ge 0}(\cH)} f_X(Y, c) \, . 
    \end{equation}
    Rewriting the RHS of the above equation gives for $\ker c \subseteq \ker X$
    \begin{equation}
        \begin{aligned}
            \sup\limits_{Y \in \cB_{\ge 0}(\cH)} f_X(Y, c) &= \sup\limits_{Y \in \cB_{\ge 0}(\cH)} p' \tr[X Y] - \frac{p'}{p} \tr[(c^{\frac{1}{2r}} Y c^{\frac{1}{2r}})^p]\\
            &= \sup\limits_{Y \in \cB_{\ge 0}(\cH)} p' \tr[c^{-\frac{1}{2r}}X c^{-\frac{1}{2r}}Y] - \frac{p'}{p} \tr[Y^p]\\
            &= \norm{c^{-\frac{1}{2r}} X c^{-\frac{1}{2r}}}_{p'}^{p'}
        \end{aligned}
    \end{equation}
    with $\cdot^{-1}$ in this context being the pseudo inverse, and $\infty$ if $ \ker c \not\subseteq \ker X$. Using the lower semi-continuity of $c \mapsto \norm{c^{-\frac{1}{2r}} X c^{-\frac{1}{2r}}}_{p'}^{p'}$ and that for every $c \in \cC$ we can approximate it via $c_\delta = (1 - \delta)c + \delta \tau \in \cC$ a positive definite sequence, due to $\tau \in \cC$ chosen positive definite and $\delta \in (0, 1)$, we have finally 
    \begin{equation}\label{eq:sion-minmax-equality}
        \inf\limits_{c \in \cC} \sup\limits_{Y \in \cB_{\ge 0}(\cH)} f_X(Y, c) =  \inf\limits_{c \in \cC, \ker c \subseteq \ker X} \norm{c^{-\frac{1}{2r}} X c^{-\frac{1}{2r}}}_{p'}^{p'} = {\norm{X}_{\cC, p', q'}^*}^{p'} \, . 
    \end{equation}
    The LHS becomes
    \begin{equation}
        \begin{aligned}
            \inf\limits_{c \in \cC} f_X(Y, c) &= p'\tr[XY] - \frac{p'}{p} \sup\limits_{c \in \cC} \tr[(c^{\frac{1}{2r}} Y c^{\frac{1}{2r}})^p]\\
            &= p'\tr[XY] - \frac{p'}{p} \norm{Y}_{\cC, p, q}^p
        \end{aligned}
    \end{equation}
    As of Equation \eqref{eq:sion-minmax-equality} we have that for all $\varepsilon > 0$ there exists $Y \in \cB_{\ge 0}(\cH)$, s.t.
    \begin{equation}
        p'\tr[XY] - \frac{p'}{p} \norm{Y}_{\cC, p, q}^p + \varepsilon \ge {\norm{X}_{\cC, p', q'}^*}^{p'}
    \end{equation}
    This immediately\footnote{To see this we introduce a non negative parameter $\lambda$ and take $g(\lambda) =  p'\tr[X(\lambda \cdot Y)] - \frac{p'}{p} \norm{(\lambda \cdot Y)}_{\cC, p, q}^p$. Clearly $\sup\limits_{\lambda \ge 0} g(\lambda) \ge p'\tr[XY] - \frac{p'}{p} \norm{Y}_{\cC, p, q}^p$ with the supremum being achieved at $\lambda^{p - 1} = \frac{\tr[XY]}{\norm{Y}_{\cC, p, q}^p}$. Inserting this immediately gives Equation \eqref{eq:inequality}.} gives that for all $\varepsilon > 0$ there exists $Y \in \cB_{\ge 0}(\cH)$ (w.l.o.g. $Y \ne 0$), s.t.
    \begin{equation}\label{eq:inequality}
        \left(\frac{\tr[XY]}{\norm{Y}_{\cC, p, q}}\right)^{p'} + \varepsilon \ge {\norm{X}_{\cC, p', q'}^*}^{p'} \, . 
    \end{equation}
    and hence taking an upper bound on the LHS,
    \begin{equation}
        \left(\sup\limits_{Y \in \cB_{\ge 0}(\cH), \norm{Y}_{\cC, p, q} \le 1} |\tr[X Y]|\right)^{p'} + \varepsilon \ge {\norm{X}_{\cC, p', q'}^*}^{p'}
    \end{equation}
    Letting $\varepsilon \to 0$ and taking the $p'$th root on both sides concludes the claim.
\end{proof}

The following corollary is a consequence.
\begin{corollary}\label{cor:triangle-inequality-Kp'r-maps}
    Let $X, Y \in \cB_{\ge 0}(\cH)$, then 
    \begin{equation}
        \norm{X + Y}_{\cC, p', q'}^* \le \norm{X}_{\cC, p', q'}^* + \norm{Y}_{\cC, p', q'}^*
    \end{equation}
\end{corollary}
\begin{proof}
    This follows directly from the formula derived in \Cref{theo:dual-formula} and the subadditivity of the supremum.
\end{proof}

We further have the following lemma. 

\begin{lemma}\label{lem:montonicity-Kp'r-maps}
    For $X, Y \in \cB_{\ge 0}(\cH)$ with $X \le Y$, we have that
    \begin{equation}
        \norm{X}_{\cC, p', q'}^* \le \norm{Y}_{\cC, p', q'}^* \, . 
    \end{equation}
\end{lemma}
\begin{proof}
    We have that for every $c \in \cC$, $c > 0$
    \begin{equation}
        c^{-\frac{1}{2r}} X c^{-\frac{1}{2r}} \le c^{-\frac{1}{2r}} Y c^{-\frac{1}{2r}}
    \end{equation}
    and hence, since the Schatten $p$-norms preserve order  on positive semi-definite operators, we have
    \begin{equation}
        \norm{c^{-\frac{1}{2r}} X c^{-\frac{1}{2r}}}_p \le \norm{c^{-\frac{1}{2r}} Y c^{-\frac{1}{2r}}}_p
    \end{equation}
    concluding the claim.
\end{proof}

Putting together the above results, we can now show the following.

\begin{theorem}[Distance to convex, compact set]\label{theo:minimial-distance-cc-set-operator-approach}
    Let $\cC \subseteq \cS(\cH)$ be a convex, compact set that contains at least one positive definite state. Then the map 
    \begin{equation}
        \widetilde D_{\alpha, \cC}:\cS(\cH) \to \R, \qquad \rho \mapsto \widetilde D_{\alpha, \cC}(\rho) :=  \inf\limits_{\tau \in \cC} \widetilde D_\alpha(\rho \Vert \tau)
    \end{equation}
    is uniformly continuous (cf.\ \cite[Definition 4.18]{RudinAnalysis}) for $\alpha \in [1/2, 1) \cup (1, \infty)$. For $\rho$, $\sigma \in \cS(\cH)$ with $\tfrac{1}{2}\norm{\rho - \sigma}_1 \le \varepsilon$, $\alpha \in [1/2, 1)$ and $\kappa$ (see \Cref{rem:kappa}) such that $\sup\limits_{\rho \in \cS(\cH)} \widetilde D_{\alpha, \cC}(\rho)\leq\log(\kappa)<\infty$, 
    \begin{equation}
         |\widetilde D_{\alpha, \cC}(\rho) - \widetilde D_{\alpha, \cC}(\sigma)| \le \frac{1}{1 - \alpha} \log(1 + \varepsilon^\alpha \kappa^{1 - \alpha}) \, . 
    \end{equation}
    Further for $\alpha \in (1, \infty)$ and $\kappa$ (see \Cref{rem:kappa}) such that $\sup\limits_{\rho \in \cS(\cH)} \widetilde D_{\alpha, \cC}(\rho)\leq\log(\kappa)<\infty$ we have 
    \begin{equation}
         |\widetilde D_{\alpha, \cC}(\rho) - \widetilde D_{\alpha, \cC}(\sigma)| \le \frac{\alpha}{\alpha - 1} \log(1 + \varepsilon \kappa^{\frac{\alpha - 1}{\alpha}})\, . 
    \end{equation}
\end{theorem}
\begin{proof}
    The proof strategy is inspired by \cite{Beigi2022}. Let $\rho, \sigma \in \cS(\cH)$. Without loss of generality, we can assume that $\frac{1}{2}\norm{\rho - \sigma}_1 = \varepsilon$, as both bounds are monotonically increasing in $\varepsilon$. We will begin with the first bound, i.e.~the bound for $\alpha < 1$ and note that 
    \begin{equation}
        |\widetilde D_{\alpha, \cC}(\rho) - \widetilde D_{\alpha, \cC}(\sigma)| \le \frac{1}{1 - \alpha} \left|\log(\frac{\sup\limits_{c \in \cC} \alphaQ(\rho \Vert c)}{\sup\limits_{c \in \cC} \alphaQ(\sigma \Vert c)})\right| \, . 
    \end{equation}
    Now, we can use that there exists $\nu, \mu \in \cS(\cH)$ such that $\rho + \varepsilon \nu = \sigma + \varepsilon \mu$ and hence $\rho \le \sigma + \varepsilon \mu$. Furthermore $\alphaQ(\cdot \Vert c)$ is monotone for all $c \in \cC$ and subadditive (c.f. \Cref{lemma:Q-alpha-super/subadditivity}), which gives us
    \begin{equation}
        \sup\limits_{c \in \cC} \alphaQ(\rho \Vert c) \le \sup\limits_{c \in \cC} \alphaQ(\sigma + \varepsilon \mu \Vert c) \le \sup\limits_{c \in \cC} \alphaQ(\sigma \Vert c) + \varepsilon^\alpha \sup\limits_{c \in \cC} \alphaQ(\mu \Vert c) \, . 
    \end{equation}
    If we use that $\kappa^{\alpha - 1} \le \sup\limits_{c \in \cC} \alphaQ(\mu \Vert c)$ and  $\sup\limits_{c \in \cC} \alphaQ(\sigma \Vert c) \le 1$, we find
    \begin{equation}
        \frac{\sup\limits_{c \in \cC} \alphaQ(\rho \Vert c)}{\sup\limits_{c \in \cC} \alphaQ(\sigma \Vert c)} \le 1 + \varepsilon^\alpha \kappa^{1 - \alpha} \, . 
    \end{equation}
    Repeating the same steps for the inverse fraction gives the claim. \par
    For $\alpha > 1$, we find 
    \begin{equation}
        |\widetilde D_{\alpha, \cC}(\rho) - \widetilde D_{\alpha, \cC}(\sigma)| = \frac{\alpha}{\alpha - 1} \Big|\log(\frac{\norm{\rho}_{\cC, \alpha, 1}^*}{\norm{\sigma}_{\cC, \alpha, 1}^*})\Big| \, . 
    \end{equation}
    Using now that there exist $\nu, \mu \in \cS(\cH)$ such that $\rho + \varepsilon \nu = \sigma + \varepsilon \mu$ and hence $\rho \le \sigma + \varepsilon \mu$, we can employ \Cref{lem:montonicity-Kp'r-maps} and \Cref{cor:triangle-inequality-Kp'r-maps}
    \begin{equation}
        \begin{aligned}
            \frac{\norm{\rho}_{\cC, \alpha, 1}^*}{\norm{\sigma}_{\cC, \alpha, 1}^*} &\le \frac{\norm{\sigma + \varepsilon \mu}_{\cC, \alpha, 1}^*}{\norm{\sigma}_{\cC, \alpha, 1}^*} \\
            &\le \frac{\norm{\sigma}_{\cC, \alpha, 1}^* + \varepsilon \norm{\mu}_{\cC, \alpha, 1}^*}{\norm{\sigma}_{\cC, \alpha, 1}^*} \\
            &\le 1 + \varepsilon\frac{\norm{\mu}_{\cC, \alpha, 1}^*}{\norm{\sigma}_{\cC, \alpha, 1}^*}
        \end{aligned}
    \end{equation}
    We can now bound $\norm{\mu}_{\cC, \alpha, 1}^*$ from above with $\kappa^{\frac{\alpha - 1}{\alpha}}$ and lower bound $\norm{\sigma}_{\cC, \alpha, 1}^*$ with $1$ from below, as $\widetilde D_{\alpha, \cC}(\cdot) \ge 0$ on quantum states. We hence obtain 
    \begin{equation}
        \frac{\norm{\rho}_{\cC, \alpha, 1}^*}{\norm{\sigma}_{\cC, \alpha, 1}^*} \le 1 + \varepsilon \kappa^{\frac{\alpha - 1}{\alpha}}
    \end{equation}
    and the same bound for the inverse quotient. This proves the claim.
\end{proof}

\begin{remark}
    In \cite{Rubboli2022}, the authors prove a continuity bound for $\alpha \in [1/2, 1)$ and $0 \leq \varepsilon \leq (\widetilde Q_{\alpha, \cC}(\rho))^{1/\alpha}$ given as
    \begin{equation}
         |\widetilde D_{\alpha, \cC}(\rho) - \widetilde D_{\alpha, \cC}(\sigma)| \le \frac{1}{\alpha-1} \log(1 - \frac{\varepsilon^\alpha}{\widetilde Q_{\alpha, \cC}(\rho)} ) \, . 
    \end{equation}
    This form is convenient for them as they work in the context of resource theories and are looking for dimension independent bounds. This bound immediately implies a continuity bound of the kind that we are looking for, namely 
        \begin{equation}
         |\widetilde D_{\alpha, \cC}(\rho) - \widetilde D_{\alpha, \cC}(\sigma)| \le \frac{1}{\alpha-1} \log(1 - \varepsilon^\alpha \kappa^{1-\alpha}) \, . 
    \end{equation}
    for $0 \leq \varepsilon \leq \kappa^{(\alpha-1)/\alpha}$. Using the inequality $ \log(1+x) \leq -\log{1-x}$ for $0 \leq x \leq 1$, it can be seen that this is worse than the bound we obtain in \Cref{theo:minimial-distance-cc-set-operator-approach}. Altering the proof in \cite{Rubboli2022} slightly one could, however, also derive the bound in \Cref{theo:minimial-distance-cc-set-operator-approach}, since both proofs are almost identical. 
\end{remark}

In the following lemma, we investigate the limiting behaviour for the bounds derived above.

\begin{lemma}[Limits] \label{lem:limit-operator}
    Let $\rho, \sigma \in \cS(\cH)$ with $\frac{1}{2}\norm{\rho - \sigma}_1 \leq \varepsilon$ and $\kappa$ (see \Cref{rem:kappa}) a uniform bound on $\widetilde D_{\alpha, \cC}(\cdot)$ independent of $\alpha$, then the limit $\alpha \to \infty$ of the bound derived in \Cref{theo:minimial-distance-cc-set-operator-approach} is stable and we find that 
    \begin{equation}
        |\widetilde D_{\infty, \cC}(\rho) - \widetilde D_{\infty, \cC}(\sigma)| \le \log(1 + \varepsilon \kappa) \, . 
    \end{equation}
    For $\alpha \to 1$ the bound diverges unless it is trivial, i.e.~$\varepsilon = 0$. 
\end{lemma}
\begin{proof}
    Both conclusions are obtained straightforwardly.
\end{proof}

\subsection{Mixed approach}

For $\alpha > 1$ we notice that both approaches have one limit ($\alpha \to 1$ or $\alpha \to \infty$) in which they perform well, while for the other one they either diverge or do not give a continuity bound anymore. The purpose of this section, therefore, is to combine both approaches and to obtain a bound which performs well in both limits. 

\begin{theorem}[Distance to convex, compact set]\label{theo:minimial-distance-cc-set-mixed-approach}
    Let $\cC \subseteq \cS(\cH)$ be a convex, compact set that contains at least one positive definite state. Then the map 
    \begin{equation}
        \widetilde D_{\alpha, \cC}:\cS(\cH) \to \R, \qquad \rho \mapsto \widetilde D_{\alpha, \cC}(\rho) :=  \inf\limits_{\tau \in \cC} \widetilde D_\alpha(\rho \Vert \tau)
    \end{equation}
    is uniformly continuous (cf.\ \cite[Definition 4.18]{RudinAnalysis}) for $\alpha \in (1, \infty)$. For $\rho$, $\sigma \in \cS(\cH)$ satisfying $\tfrac{1}{2}\norm{\rho - \sigma}_1 \le \varepsilon$ and $\kappa$ (see \Cref{rem:kappa}) such that $\sup\limits_{\rho \in \cS(\cH)} \widetilde D_{\alpha, \cC}(\rho)\leq\log(\kappa)<\infty$ we find
    \begin{equation}
         |\widetilde D_{\alpha, \cC}(\rho) - \widetilde D_{\alpha, \cC}(\sigma)| \le  \log(1 + \varepsilon) + \frac{\alpha}{\alpha - 1}\log(1 + \varepsilon \kappa^{\frac{\alpha - 1}{\alpha}} - \frac{\varepsilon^{\frac{2\alpha - 1}{\alpha}}}{(1 + \varepsilon)^{\frac{\alpha - 1}{\alpha}}}) \, . 
    \end{equation}
\end{theorem}
\begin{proof}
    Let $\rho, \sigma \in \cS(\cH)$. Without loss of generality, we can assume that $\frac{1}{2}\norm{\rho - \sigma}_1 = \varepsilon$, as the bound is monotonically increasing in $\varepsilon$. Since the logarithm and the map $x \mapsto x^\frac{1}{\alpha}$ are monotone, we find 
    \begin{equation}
        |\widetilde D_{\alpha, \cC}(\rho) - \widetilde D_{\alpha, \cC}(\sigma)| = \frac{\alpha}{\alpha - 1} \Bigg|\log(\frac{\norm{\rho}_{\cC, \alpha, 1}^*}{\norm{\sigma}_{\cC, \alpha, 1}^*})\Bigg| \, . 
    \end{equation}
    We further get $\mu, \nu \in \cS(\cH)$, such that $\rho + \varepsilon \mu = \sigma + \varepsilon\nu$. Due to \Cref{cor:triangle-inequality-Kp'r-maps} and $\norm{\nu}_{\cC, \alpha, 1}^* \le \kappa^{\frac{\alpha - 1}{\alpha}}$, we have 
    \begin{equation}
        \norm{\sigma + \varepsilon\nu}_{\cC, \alpha, 1}^* \le \norm{\sigma}_{\cC, \alpha, 1}^* + \varepsilon \norm{\nu}_{\cC, \alpha, 1}^* \le \norm{\sigma}_{\cC, \alpha, 1}^* + \varepsilon \kappa^{\frac{\alpha - 1}{\alpha}}\, . 
    \end{equation}
    We then also have that for $X \ge 0$, $\norm{X}_{\cC, \alpha, 1}^* = \Big(\inf\limits_{c \in \cC} \alphaQ (X \Vert c)\Big)^{\frac{1}{\alpha}}$, which gives
    \begin{equation}
        \begin{aligned}
             \norm{\rho + \varepsilon \mu}_{\cC, \alpha, 1}^* &= \Big(\inf\limits_{c \in \cC} \alphaQ (\rho + \varepsilon\mu \Vert c)\Big)^{\frac{1}{\alpha}} \\
             &\ge \Big(\inf\limits_{c \in \cC} \alphaQ (\rho \Vert c) + \varepsilon^\alpha \inf\limits_{c \in \cC} \alphaQ (\mu \Vert c)\Big)^{\frac{1}{\alpha}}\\
             &= (1 + \varepsilon)^{\frac{1}{\alpha}}\Big(\frac{1}{1 + \varepsilon} {\norm{\rho}_{\cC, \alpha, 1}^*}^\alpha + \frac{\varepsilon}{1 + \varepsilon} \varepsilon^{\alpha - 1} {\norm{\mu}_{\cC, \alpha, 1}^*}^\alpha\Big)^{\frac{1}{\alpha}}\\
             &\ge (1 + \varepsilon)^{\frac{1 - \alpha}{\alpha}}\Big(\norm{\rho}_{\cC, \alpha, 1}^* + \varepsilon^{\frac{2\alpha - 1}{\alpha}} \norm{\mu}_{\cC, \alpha, 1}^*\Big)\\
             &\ge (1 + \varepsilon)^{\frac{1 - \alpha}{\alpha}}\Big(\norm{\rho}_{\cC, \alpha, 1}^* + \varepsilon^{\frac{2\alpha - 1}{\alpha}}\Big)
        \end{aligned}
    \end{equation}
     where we used \Cref{lemma:Q-alpha-super/subadditivity}, the concavity of $x\mapsto x^{\frac{1}{\alpha}}$ and that $\norm{\mu}_{\cC, \alpha, 1}^* \ge 1$, since $\widetilde D_{\alpha, \cC}(\cdot) \ge 0$ on quantum states. Combining these two bounds, we find
    \begin{equation}
        \norm{\rho}_{\cC, \alpha, 1}^* \le (1 + \varepsilon)^{\frac{\alpha - 1}{\alpha}}\Bigg(\norm{\sigma}_{\cC, \alpha, 1}^* + \varepsilon\kappa^{\frac{\alpha - 1}{\alpha}} - \frac{\varepsilon^{\frac{2\alpha - 1}{\alpha}}}{(1 + \varepsilon)^{\frac{\alpha - 1}{\alpha}}}\Bigg)
    \end{equation}
    and as a consequence (since again $\norm{\sigma}_{\cC, \alpha, 1}^* \ge 1$ and $\varepsilon\kappa^{\frac{\alpha - 1}{\alpha}} - \frac{\varepsilon^{\frac{2\alpha - 1}{\alpha}}}{(1 + \alpha)^{\frac{\alpha - 1}{\alpha}}} > 0$)
    \begin{equation}
        \frac{\norm{\rho}_{\cC, \alpha, 1}^*}{\norm{\sigma}_{\cC, \alpha, 1}^*} \le (1 + \varepsilon)^{\frac{\alpha - 1}{\alpha}}\Bigg(1 + \varepsilon \kappa^{\frac{\alpha - 1}{\alpha}} - \frac{\varepsilon^{\frac{2\alpha - 1}{\alpha}}}{(1 + \varepsilon)^{\frac{\alpha - 1}{\alpha}}}\Bigg) \, . 
    \end{equation}
    Repeating the same steps for the inverse fraction proves the claimed bound.
\end{proof}

\begin{lemma}[Limits]\label{lem:limit-mixed}
    Let $\rho, \sigma \in \cS(\cH)$ with $\frac{1}{2}\norm{\rho - \sigma}_1 \le \varepsilon$, and $\kappa$ (see \Cref{rem:kappa}) be a bound on $\widetilde D_{\alpha, \cC}(\cdot)$ independent of $\alpha \in (1, \infty)$. Then, the limit $\alpha \to 1$ of the bounds obtained in \Cref{theo:minimial-distance-cc-set-mixed-approach} gives
    \begin{equation}
        |\widetilde D_{1, \cC}(\rho) - \widetilde D_{1, \cC}(\sigma)| \le \varepsilon \log \kappa \,  + (1 + \varepsilon) h\Big(\frac{\varepsilon}{1 + \varepsilon}\Big)
    \end{equation}
    and for $\alpha \to \infty$ 
    \begin{equation}
        |\widetilde D_{\infty, \cC}(\rho) - \widetilde D_{\infty, \cC}(\sigma)| \le \log((1 + \varepsilon)(1 + \varepsilon \kappa) - \varepsilon^2)
    \end{equation}
\end{lemma}
\begin{proof}
     With $\beta = \frac{\alpha - 1}{\alpha}$ l'Hospital's rule lets us infer that
    \begin{align}
        \lim_{\beta \to 0}(\beta)^{-1}\log(1 + \varepsilon \kappa^{ \beta} - \frac{\varepsilon^{1+ \beta}}{(1 + \varepsilon)^{\beta}})&=\lim_{\beta \to 0}\frac{\varepsilon \kappa^{ \beta} \log{\kappa} - \log(\varepsilon)\frac{\varepsilon^{1+ \beta}}{(1 + \varepsilon)^{\beta}}+\log(1+\varepsilon)\frac{\varepsilon^{1+ \beta}}{(1 + \varepsilon)^{\beta}}}{1 + \varepsilon \kappa^{\beta} - \frac{\varepsilon^{1+ \beta}}{(1 + \varepsilon)^{\beta}}}\\
        &= \varepsilon \log{\kappa} - \varepsilon \log{\varepsilon} + \varepsilon \log(1+\varepsilon) \, ,
    \end{align}
    and
    \begin{equation}
        \lim_{\alpha \to \infty} \frac{\alpha}{\alpha - 1}\log(1 + \varepsilon \kappa^{ \frac{\alpha - 1}{\alpha}} - \frac{\varepsilon^{2 - \frac{1}{\alpha}}}{(1 + \varepsilon)^{\frac{\alpha - 1}{\alpha}}}) = \log(1 + \varepsilon \kappa - \frac{\varepsilon^{2}}{(1 + \varepsilon)}) \, .
    \end{equation}
\end{proof}

\section{Continuity bounds for sandwiched Rényi divergences}\label{sec:CB_Renyi}

\subsection{Continuity bounds for the sandwiched Rényi conditional entropy}\label{sec:sand-cond-entropy}
In this section, we prove and compare the bounds for the sandwiched Rényi conditional entropy using the three different approaches: almost-additive, operator space, and mixed. We will start with a discussion of these bounds for the example of the sandwiched Rényi conditional entropies, but our conclusions carry over to the other quantities in the following sections as well.

The comparison is shown in \Cref{fig:comparison-of-bounds}. As we discussed in \Cref{sec:main-results}, the strength and weaknesses of each bound depend on the combination of parameters $\alpha$, $d_A$, and $\varepsilon$.
Our analysis reveals that the almost additive approach performs best in the low $d_A$ regime with small $\alpha$, followed by a region where the mixed approach is superior, and then with increasing $\alpha$, the operator-space approach outperforms the other two. As we increase $d_A$, the almost additive approach becomes progressively weaker compared to the mixed and operator-space approaches. Specifically, in the low $\alpha$ regime, the mixed approach dominates, while in the high $\alpha$ regime, the operator-space approach performs best. This improvement in performance with increasing $d_A$ for the mixed and operator-space approaches is due to their scaling with $d_A^{2\frac{\alpha - 1}{\alpha}}$, which is more favourable than the scaling with $d_A^{2(\alpha - 1)}$ of the bound we got using the almost additive approach.
Notably, the superior scaling of the mixed and operator-space approaches allows us to take the limit $\alpha \to \infty$ and obtain a continuity bound. However, the almost additive approach fails to have this property and does not vanish for $\varepsilon \to 0$ in this limit. 
\begin{figure}[ht!]
    \centering
    \includegraphics[width=\linewidth]{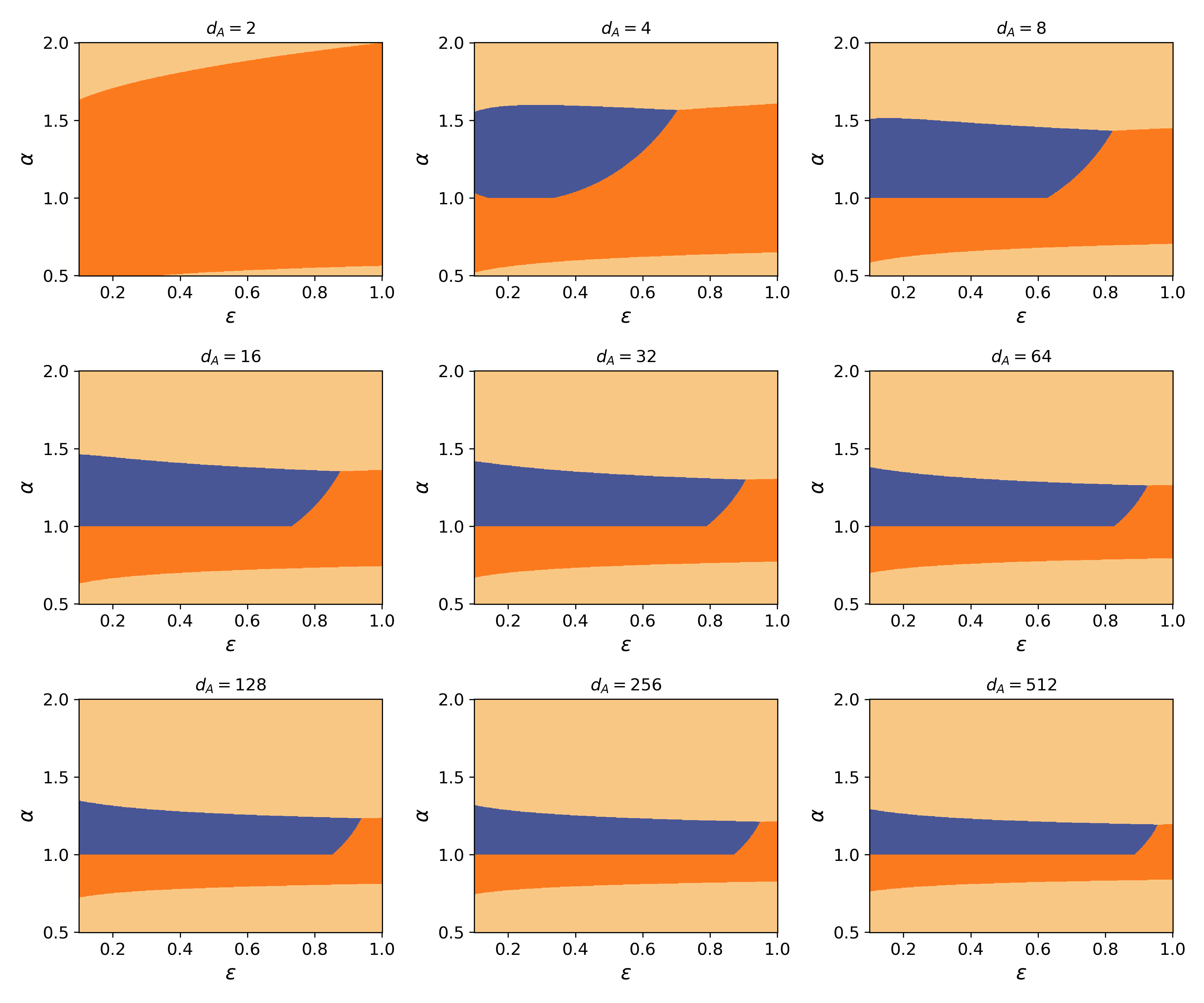}
    \caption{A comparison of the continuity bounds for $\tilde{H}_\alpha(A|B)_\rho$ proven by the \textcolor{tbcolor1}{\rule{0.3cm}{0.3cm} almost-additive}, \textcolor{tbcolor2}{\rule{0.3cm}{0.3cm} operator space}, and \textcolor{tbcolor3}{\rule{0.3cm}{0.3cm} mixed} approach. The value of $d_A$ is according to the title of each plot. The visible colour indicates where the respective bound outperforms (is tighter than) the others.}
    \label{fig:comparison-of-bounds}
\end{figure}

\begin{corollary}\label{cor:continuity-bounds-renyi-conditional-entropy}
    Let $\rho, \sigma \in \cS(\cH_A \otimes \cH_B)$, with $\frac{1}{2}\norm{\rho - \sigma} \le \varepsilon$, then for $\alpha \in [1/2, 1)$
    \begin{equation}\label{eq:bounds-renyi-conditional-entropy-alpha<1}
          |\widetilde H^\uparrow_\alpha(A|B)_\rho - \widetilde H^\uparrow_\alpha(A|B)_\sigma| \le \log(1 + \varepsilon) + \frac{1}{1 - \alpha}\log(1 + \varepsilon^\alpha d_A^{2(1-\alpha)} - \frac{\varepsilon}{(1 + \varepsilon)^{1-\alpha}}) \, ,
    \end{equation}
     which is the bound from \cite{marwah2022} (see Eq.\ \ref{eq:marwah-smaller-bound-than-1}), and for $\alpha \in (1, \infty)$
    \begin{equation}\label{eq:bounds-renyi-conditional-entropy-alpha>1}
         |\widetilde H^\uparrow_\alpha(A|B)_\rho - \widetilde H^\uparrow_\alpha(A|B)_\sigma|  \le \min \begin{cases}
             \log(1 + \varepsilon) + \frac{1}{\alpha - 1}\log(1 + \varepsilon d_A^{2(\alpha - 1)}- \frac{\varepsilon^\alpha}{(1 + \varepsilon)^{\alpha - 1}}) \, ,\\
             \frac{\alpha}{\alpha - 1}\log(1 + \varepsilon d_A^{2\frac{\alpha - 1}{\alpha}}) \, ,\\
             \log(1 + \varepsilon) + \frac{\alpha}{\alpha - 1}\log(1 + \varepsilon d_A^{2 \frac{\alpha - 1}{\alpha}} - \frac{\varepsilon^{2 - \frac{1}{\alpha}}}{(1 + \varepsilon)^{\frac{\alpha - 1}{\alpha}}})  \, . 
         \end{cases}
    \end{equation}
\end{corollary}
\begin{proof}
    We have that $\cC = \{d_A^{-1} \1_A \otimes \sigma_B \;:\; \sigma_B \in \cS(\cH_B)\}$ is clearly a convex and compact set, containing a positive definite state. Using this definition and the $(1 - \alpha)$-homogeneity of the $\alphaQ$ in the second argument, we get
    \begin{equation}
        \widetilde H^\uparrow_\alpha(A|B)_\rho = - \widetilde D_{\alpha, \cC}(\rho) + \log d_A
    \end{equation}
    and hence
    \begin{equation}
        |\widetilde H^\uparrow_\alpha(A|B)_\rho - \widetilde H^\uparrow_\alpha(A|B)_\sigma| = |\widetilde D_{\alpha, \cC}(\rho) - \widetilde D_{\alpha, \cC}(\sigma)|
    \end{equation}
    We further have that $\sup_{\rho \in \cS(\cH)} \widetilde D_{\alpha, \cC}(\rho) \le 2\log d_A$ using Eq.\ \eqref{eq:bound-sandwiched-cond-renyi} and hence can directly apply \Cref{theo:minimial-distance-cc-set-almost-additive-approach}, \Cref{theo:minimial-distance-cc-set-operator-approach}, and \Cref{theo:minimial-distance-cc-set-mixed-approach} to obtain the bounds in the assertion.
\end{proof}

We can now study the limits $\alpha \to 1$ and $\alpha \to \infty$ of the new bound in \Cref{cor:continuity-bounds-renyi-conditional-entropy}. We find that the former limit coincides with the bounds by Alicki, Fannes and Winter and the latter bound with the one found by Marwah \& Dupuis in the appendix of \cite[Theorem 2]{marwah2022}.

\begin{corollary}[Limits]\label{cor:limits-renyi-conditional-entropy}
    Let $\rho, \sigma \in \cS(\cH_A \otimes \cH_B)$, with $\tfrac{1}{2}\norm{\rho - \sigma}_1 \le \varepsilon$. Then, the limit $\alpha \to 1$ of Equation \eqref{eq:bounds-renyi-conditional-entropy-alpha<1} yields
    \begin{equation}
         \abs{H( A | B )_\rho - H( A | B )_\sigma} \leq  2 \varepsilon \log d_A + (1+ \varepsilon) h\left(\frac{\varepsilon}{1+\varepsilon}\right) \, ,
    \end{equation}
    and of Equation \eqref{eq:bounds-renyi-conditional-entropy-alpha>1}
    \begin{equation}
             \abs{H( A | B )_\rho - H( A | B )_\sigma} \le \min \begin{cases}
                 2 \varepsilon \log d_A + (1+ \varepsilon) h\left(\frac{\varepsilon}{1+\varepsilon}\right) \, ,\\
                 \infty \, ,\\
                 2 \varepsilon \log d_A + (1+ \varepsilon) h\left(\frac{\varepsilon}{1+\varepsilon}\right) \, ,
             \end{cases} 
    \end{equation}
    where $H_\rho( A | B )$ is the usual quantum conditional entropy, which is the bound derived in \cite{Winter-AlickiFannes-2016}. The limit $\alpha \to \infty$ of \Cref{eq:bounds-renyi-conditional-entropy-alpha>1} yields
    \begin{equation}
         |H^\uparrow_\infty(A|B)_\rho -  H^\uparrow_\infty(A|B)_\sigma| \le \min\begin{cases}
             \log(1 + \varepsilon) + 2 \log d_A \, , \\
              \log(1 + \varepsilon d_A^2) \, ,\\
              \log((1 + \varepsilon)(1 + \varepsilon d_A^{2}) - \varepsilon^{2}) \, .
         \end{cases} 
    \end{equation}
\end{corollary}
\begin{proof}
    The proof is a direct application of \Cref{lem:limit-almost-additive}, \Cref{lem:limit-operator}, \Cref{lem:limit-mixed} in the context of \Cref{cor:continuity-bounds-renyi-conditional-entropy}.
\end{proof}

\subsection{Continuity bounds for the sandwiched Rényi mutual information}\label{sec:sand-mutual-info}

Now, we can bring our techniques to bear on the sandwiched Rényi mutual information. We, unfortunately, cannot employ the theorems established in the beginning as $\{ \rho_A \otimes \rho_B \;:\; \rho_A \in \cS(\cH_A), \rho_B \in \cS(\cH_B)\}$ is not a convex set, however, will use techniques very similar to those presented already.

\begin{corollary}
    Let $\rho, \sigma \in \cS(\cH_A \otimes \cH_B)$ with $\tfrac{1}{2}\norm{\rho - \sigma}_1 \le \varepsilon$. Then, for $ \alpha \in [1/2, 1)$ we find
    \begin{equation}\label{eq:almost-additive-MI-smaller}
         |\alphaI(A:B)_\rho - \alphaI(A:B)_\sigma| \le 2\log(1 + \varepsilon^{\frac{1}{\alpha}}) + \frac{1}{1 - \alpha}\log(1 + \varepsilon^\alpha m^{2(1 - \alpha)} - \frac{\varepsilon^{\frac{1}{\alpha}}}{(1 + \varepsilon^\frac{1}{\alpha})^{2(1 - \alpha)}}) \, ,
    \end{equation}
    and for $\alpha \in (1, \infty)$ we have
    \begin{equation}\label{eq:almost-additive-MI-bigger}
        |\alphaI(A:B)_\rho - \alphaI(A:B)_\sigma| \le 2\log(1 + \varepsilon^{\frac{1}{\alpha}}) + \frac{1}{\alpha - 1}\log(1 + \varepsilon^{\frac{1}{\alpha}} m^{2(\alpha - 1)} - \frac{\varepsilon^\alpha}{(1 + \varepsilon^{\frac{1}{\alpha}})^{2(\alpha - 1)}}) \, ,
    \end{equation}
    where in both bounds $m = \min\{d_A, d_B\}$.
\end{corollary}
\begin{proof}
    We will only demonstrate the proof for $\alpha \in (1, \infty)$ as the case $\alpha \in [1/2,1)$ is almost completely analogous up to reversing the inequalities. We further will only cover the case where $\tfrac{1}{2}\norm{\rho - \sigma}_1 = \varepsilon$, since the proposed bounds are monotone in $\varepsilon$. Let $\mu, \nu$ be orthogonal quantum states such that $\rho + \varepsilon\mu = \sigma + \varepsilon\nu$. It is straightforward to see that deriving the claimed bound boils down to deriving upper bounds on 
    \begin{equation}
        \frac{\inf\limits_{\tau_A, \tau_B} \alphaQ(\rho \Vert \tau_A \otimes \tau_B)}{\inf\limits_{\tau_A, \tau_B} \alphaQ(\sigma \Vert \tau_A \otimes \tau_B)} \quad \text{and} \quad \frac{\inf\limits_{\tau_A, \tau_B} \alphaQ(\sigma \Vert \tau_A \otimes \tau_B)}{\inf\limits_{\tau_A, \tau_B} \alphaQ(\rho \Vert \tau_A \otimes \tau_B)} \, .
    \end{equation}
    Since the proof for both fractions is exactly the same, we will only demonstrate it for the first one here. Employing joint convexity of the $\alphaQ$ (c.f. \Cref{lemma:Q-alpha-joint-convexity/concavity}), we get that 
    \begin{equation}
        \inf\limits_{\tau_A, \tau_B} \alphaQ(\sigma + \varepsilon \nu \Vert \tau_A \otimes \tau_B) \le (1 + \varepsilon)^{\alpha - 1} (\inf\limits_{\tau_A} (\inf_{\tau_B} \alphaQ(\sigma \Vert \tau_A \otimes \tau_B) + \varepsilon\inf\limits_{\tau_B} \alphaQ(\nu \Vert \tau_A \otimes \tau_B)))
    \end{equation}
    where we could split the first infimum by writing the optimization as an optimization over $\tau_{B, 1}$ and $\tau_{B, 2}$, replacing $\tau_B = \frac{1}{1 + \varepsilon}\tau_{B, 1} + \frac{\varepsilon}{1 + \varepsilon} \tau_{B, 2}$ and then apply joint convexity. The second infimum is divided in $\tau_A = \frac{1}{1 + \varepsilon^{\frac{1}{\alpha}}} \tau_{A, 1} + \frac{\varepsilon^{\frac{1}{\alpha}}}{1 + \varepsilon^{\frac{1}{\alpha}}} \tau_{A, 2}$ and bounded via the anti-monotonicity of $\alphaQ$ in the second argument (see \cite[Lemma 4.10]{Tomamichel_2016}) and $\tau_A \ge  \frac{1}{1 + \varepsilon^{\frac{1}{\alpha}}} \tau_{A, 1}$, $\tau_A \ge \frac{\varepsilon^{\frac{1}{\alpha}}}{1 + \varepsilon^{\frac{1}{\alpha}}} \tau_{A, 2}$. This gives
    \begin{align}
        (1 + \varepsilon)^{\alpha - 1} (\inf\limits_{\tau_A} (\inf\limits_{\tau_B} &\alphaQ(\sigma \Vert \tau_A \otimes \tau_B) + \varepsilon\inf\limits_{\tau_B} \alphaQ(\nu \Vert \tau_A \otimes \tau_B))) \\
        &\le (1 + \varepsilon)^{\alpha - 1}(1 + \varepsilon^{\frac{1}{\alpha}})^{\alpha - 1}(\inf\limits_{\tau_A, \tau_B} \alphaQ(\sigma \Vert \tau_A \otimes \tau_B) + \varepsilon^{1 + \frac{1 - \alpha}{\alpha}}\inf\limits_{\tau_A, \tau_B}\alphaQ(\nu \Vert \tau_A \otimes \tau_B))\\
        &\le (1 + \varepsilon^{\frac{1}{\alpha}})^{2(\alpha - 1)}(\inf\limits_{\tau_A, \tau_B }\Qalpha(\sigma \Vert \tau_A \otimes \tau_B) + \varepsilon^{\frac{1}{\alpha}} \inf\limits_{\tau_A, \tau_B} \Qalpha(\nu \Vert \tau_A \otimes \tau_B))\,.
    \end{align}
    Now, using again the superadditivity (c.f. \Cref{lemma:Q-alpha-super/subadditivity}), we get
    \begin{align}
         \inf\limits_{\tau_A, \tau_B} \alphaQ(\sigma + \varepsilon \nu \Vert \tau_A \otimes \tau_B) &= \inf\limits_{\tau_A, \tau_B} \alphaQ(\rho + \varepsilon \mu \Vert \tau_A \otimes \tau_B)\\
         &\ge \inf\limits_{\tau_A, \tau_B} \alphaQ(\rho \Vert \tau_A \otimes \tau_B) + \varepsilon^\alpha \inf\limits_{\tau_A, \tau_B} \alphaQ( \mu \Vert \tau_A \otimes \tau_B)
    \end{align}
    It holds that for a quantum state $\xi$, $1 \le \inf\limits_{\tau_A, \tau_B}\alphaQ(\xi\Vert \tau_A \otimes \tau_B) \le m^{2(\alpha - 1)}$, where $m = \min \{ d_A, d_B\}$. Indeed, assuming without loss of generality $m = d_A$, we can estimate $\widetilde I_\alpha(A:B)_\xi \leq - \widetilde H(A|B)_\xi + \log d_A$ and use the bound in Eq.\ \eqref{eq:bound-sandwiched-cond-renyi}. Combining both estimates with this bound, we get
    \begin{equation}
        \inf\limits_{\tau_A, \tau_B} \alphaQ(\rho \Vert \tau_A \otimes \tau_B) \le (1 + \varepsilon^{\frac{1}{\alpha}})^{2(\alpha - 1)}(\inf\limits_{\tau_A, \tau_B} \alphaQ(\sigma \Vert \tau_A \otimes \tau_B) + \varepsilon^{\frac{1}{\alpha}}m^{2(\alpha - 1)}) - \varepsilon^\alpha \,. 
    \end{equation}
    Subsequently, we divide by $\inf\limits_{\tau_A, \tau_B} \alphaQ(\sigma \Vert \tau_A \otimes \tau_B)$ and use again that for $\xi$ a quantum state $1 \le \inf\limits_{\tau_A, \tau_B} \alphaQ(\xi \Vert \tau_A \otimes \tau_B)$. In addition, we use that
    $(1 + \varepsilon^{\frac{1}{\alpha}})^{2(\alpha - 1)}\varepsilon^{\frac{1}{\alpha}}m^{2(\alpha - 1)} - \varepsilon^\alpha \geq \varepsilon^{ 1/\alpha}- \varepsilon^{\alpha} \geq 0$.
    This leads to
      \begin{align}
        \frac{\inf\limits_{\tau_A, \tau_B} \alphaQ(\rho \Vert \tau_A \otimes \tau_B)}{\inf\limits_{\tau_A, \tau_B} \alphaQ(\sigma \Vert \tau_A \otimes \tau_B)} &\le(1 + \varepsilon^{\frac{1}{\alpha}})^{2(\alpha - 1)}\left(1 + \varepsilon^{\frac{1}{\alpha}}m^{2(\alpha - 1)} - \frac{\varepsilon^\alpha}{(1 + \varepsilon^{\frac{1}{\alpha}})^{2(\alpha - 1)}}\right)
    \end{align}
    Applying the logarithm, multiplying with $\frac{1}{\alpha - 1}$ and then repeating the whole procedure for the other fraction gives the claimed result.
\end{proof}

\begin{corollary}[Limits]
Let $\rho, \sigma \in \cS(\cH_A \otimes \cH_B)$ with $\tfrac{1}{2}\norm{\rho - \sigma}_1 \le \varepsilon$. Then, taking $\alpha \to 1$ in  Eqs.\ \eqref{eq:almost-additive-MI-smaller} and \eqref{eq:almost-additive-MI-bigger} yields
    \begin{equation}
        |I(A:B)_\rho - I(A:B)_\sigma| \leq 2 \varepsilon \log{m} + 2 (1+\varepsilon) h\left(\frac{\varepsilon}{1+\varepsilon}\right),
    \end{equation}
    where $I(A:B)_\rho$ is the usual mutual information. This is the bound from \cite{Shirokov_Mutual_Information}. Taking the limit $\alpha \to \infty$ leads to 
        \begin{equation}
        |I^\uparrow_\infty(A:B)_\rho - I^\uparrow_\infty(A:B)_\sigma| \leq \log{4 m^2},
    \end{equation}
    which is no longer a continuity bound.
\end{corollary}
\begin{proof}
    The first limit can be obtained using l'Hospital's rule as in \Cref{lem:limit-mixed}. The second limit can be obtained by isolating $m^{2(\alpha -1)}$ in the logarithm.
\end{proof}

\subsection{Continuity bounds for the sandwiched Rényi conditional mutual information}\label{sec:sand-cond-mutual-info}

Let us recall that given a tripartite space $\cH_{ABC}=\cH_A \otimes \cH_B \otimes \cH_C$ and  $\rho_{ABC} \in  \cS(\cH_{ABC}) $, the sandwiched R{\'e}nyi conditional mutual information of $\rho_{ABC}$ is given by 
\begin{equation}
    \widetilde I^\uparrow_\alpha(A:C|B)_\rho := \widetilde H^\uparrow_\alpha(A|B)_\rho - \widetilde H^\uparrow_\alpha(A|BC)_\rho \, .
\end{equation}
We also define the max conditional mutual information 
\begin{equation}
    \widetilde I^\uparrow_\infty(A:C|B)_\rho := \widetilde H^\uparrow_\infty(A|B)_\rho - \widetilde H^\uparrow_\infty(A|BC)_\rho \, .
\end{equation}
and the quantum conditional mutual information
\begin{equation} \label{eq:cond-mutual-information}
    \widetilde I(A:C|B)_\rho := \widetilde H(A|B)_\rho - \widetilde H(A|BC)_\rho \, .
\end{equation}

As a consequence of the definition, we can derive continuity bounds for the sandwiched Rényi conditional mutual information in terms of the continuity bounds for sandwiched Rényi conditional entropies obtained in Section \ref{sec:sand-cond-entropy}. This is the content of the next result. 
\begin{corollary}[Continuity bound for sandwiched Rényi conditional mutual information]\label{cor:continuity-bound-renyi-CMI}
    Let $\rho$, $\sigma \in \mathcal S(\mathcal H_{ABC})$ be quantum states, with $\tfrac{1}{2}\norm{\rho - \sigma}_1 \le \varepsilon$.
    Then, for $\alpha \in [1/2,1)$ we have 
    \begin{align} \label{eq:CB-sand-CMI-smaller1}
        |\widetilde I^\uparrow_\alpha(A:C|B)_\rho  - \widetilde I^\uparrow_\alpha(A:C|B)_\sigma| \le 2\log(1 + \varepsilon) + \frac{2}{1 - \alpha}\log(1 + \varepsilon^\alpha d_A^{2(1 - \alpha)} - \frac{\varepsilon}{(1 + \varepsilon)^{1 - \alpha}}) \, , 
    \end{align}
    and for $\alpha \in (1, \infty)$,
    \begin{align} \label{eq:CB-sand-CMI-larger1}
        |\widetilde I^\uparrow_\alpha(A:C|B)_\rho  - \widetilde I^\uparrow_\alpha(A:C|B)_\sigma| \le \min \begin{cases}
            2 \log(1 + \varepsilon) + \frac{2}{\alpha - 1}\log(1 + \varepsilon d_A^{2(\alpha - 1)} - \frac{\varepsilon^\alpha}{(1 + \varepsilon)^{\alpha - 1}}) \, , \\
            \frac{2\alpha}{\alpha - 1}\log(1 + \varepsilon d_A^{2\frac{\alpha - 1}{\alpha}}) \, ,\\
            2 \log(1 + \varepsilon) + \frac{2 \alpha}{\alpha - 1}\log(1 + \varepsilon d_A^{2 \frac{\alpha - 1}{\alpha}} - \frac{\varepsilon^{2 - \frac{1}{\alpha}}}{(1 + \varepsilon)^{\frac{\alpha - 1}{\alpha}}}) \, . 
        \end{cases} 
    \end{align}
\end{corollary}
\begin{proof}
    Note that
    \begin{align}
        | \widetilde I^\uparrow_\alpha(A:C|B)_\rho  - \widetilde I^\uparrow_\alpha(A:C|B)_\sigma | \leq | \widetilde H^\uparrow_\alpha(A|B)_\rho - \widetilde H^\uparrow_\alpha(A|B)_\sigma| + | \widetilde H^\uparrow_\alpha(A|BC)_\rho - \widetilde H^\uparrow_\alpha(A|BC)_\sigma| \, .
    \end{align}
    Thus, a continuity bound for the sandwiched Rényi conditional mutual information follows as a continuity bound for the sandwiched Rényi conditional entropy, with a factor of $2$. We conclude using the bounds from \Cref{cor:continuity-bounds-renyi-conditional-entropy}.
\end{proof}

For the sake of consistency, we state in the following the limits $\alpha\rightarrow1$ and $\alpha\rightarrow\infty$ which directly follows from \Cref{cor:limits-renyi-conditional-entropy} by the same argument as in the proof of the above corollary:
\begin{corollary}[Limits]\label{cor:limits-renyi-conditional-information}
    Let $\rho$, $\sigma \in \mathcal S(\mathcal H_{ABC})$ be quantum states, with $\tfrac{1}{2}\norm{\rho - \sigma}_1 \le \varepsilon$. Then, the limit $\alpha \to 1$ of Equation \eqref{eq:CB-sand-CMI-smaller1} yields
    \begin{equation}
         |I(A:C|B)_\rho  - I(A:C|B)_\sigma| \leq  4 \varepsilon \log d_A + 2(1+ \varepsilon) h\left(\frac{\varepsilon}{1+\varepsilon}\right) \, ,
    \end{equation}
    and the limit of Equation \eqref{eq:CB-sand-CMI-larger1}
    \begin{equation}
             |I(A:C|B)_\rho  -  I(A:C|B)_\sigma| \le \min \begin{cases}
                 4 \varepsilon \log d_A + 2(1+ \varepsilon) h\left(\frac{\varepsilon}{1+\varepsilon}\right) \, ,\\
                 \infty \, ,\\
                 4 \varepsilon \log d_A + 2(1+ \varepsilon) h\left(\frac{\varepsilon}{1+\varepsilon}\right) \, ,
             \end{cases} 
    \end{equation}
    where $I(A:C|B)_\rho$ is the usual quantum conditional information (see Eq. \eqref{eq:cond-mutual-information}). The limit $\alpha \to \infty$ of equation \eqref{eq:CB-sand-CMI-larger1} yields
    \begin{equation}
         |I^\uparrow_\infty(A:C|B)_\rho  -  I^\uparrow_\infty(A:C|B)_\sigma| \le \min\begin{cases}
                2\log(1 + \varepsilon) + 4 \log d_A \, , \\
                2\log(1 + \varepsilon d_A^2) \, ,\\
                2\log((1 + \varepsilon)(1 + \varepsilon d_A^{2}) - \varepsilon^{2}) \, .
         \end{cases} 
    \end{equation}
\end{corollary}
\begin{proof}
    The proof is a direct application of \Cref{cor:limits-renyi-conditional-entropy}.
\end{proof}

\subsection{Continuity bounds in the first argument for sandwiched Rényi divergences}\label{sec:sand-first-arg}
\begin{corollary}[Continuity bound in the first argument]\label{cor:continuity-bound-renyi-first-argument}
    Let $\rho$, $\sigma$, $\tau \in \mathcal S(\mathcal H)$ be quantum states, with $\ker \tau \subseteq \ker \rho \cap \ker \sigma$, $\tfrac{1}{2}\norm{\rho - \sigma}_1 \le \varepsilon$ and $\alpha \in [1/2, 1)$. Then 
    \begin{equation} \label{eq:continuity-bounds-renyi-first-alpha<1}
        |\Dalpha(\rho \Vert \tau) - \Dalpha(\sigma \Vert \tau)| \le \log(1 + \varepsilon) + \frac{1}{1 - \alpha}\log(1 + \varepsilon^\alpha \widetilde m_\tau^{\alpha - 1} - \frac{\varepsilon}{(1 + \varepsilon)^{1 - \alpha}}) \, , 
    \end{equation}
    where $\widetilde m_\tau$ is the smallest non-zero eigenvalue of $\tau$. For $\alpha \in (1, \infty)$ we find 
    \begin{equation} \label{eq:continuity-bounds-renyi-first-alpha>1}
         |\Dalpha(\rho \Vert \tau) - \Dalpha(\sigma \Vert \tau)| \le \min \begin{cases}
             \log(1 + \varepsilon) + \frac{1}{\alpha - 1}\log(1 + \varepsilon \widetilde m_\tau^{1 - \alpha} - \frac{\varepsilon^\alpha}{(1 + \varepsilon)^{\alpha - 1}}) \, , \\
             \frac{\alpha}{\alpha - 1} \log(1 + \varepsilon \widetilde m_\tau^{\frac{1 - \alpha}{\alpha}}) \, ,\\
             \log(1 + \varepsilon) + \frac{\alpha}{\alpha - 1}\log(1 + \varepsilon \widetilde m_\tau^{\frac{1 - \alpha}{\alpha}} - \frac{\varepsilon^{2 - \frac{1}{\alpha}}}{(1 + \varepsilon)^{\frac{\alpha - 1}{\alpha}}}) \, . 
         \end{cases} 
    \end{equation}
\end{corollary}
\begin{proof}
    We first restrict $\cH$ to the support of $\tau$ which gives us a Hilbert space $\widetilde \cH$ on which $\tau$ is positive definite. Clearly $\rho$, $\sigma \in \cS(\widetilde \cH)$ as $\ker \tau \subseteq \ker \rho \cap \ker\sigma$. We further have that $\cC := \{\tau\}$ is a convex, compact set containing a positive definite state and 
    \begin{equation}
        \sup\limits_{\rho \in \cS(\widetilde \cH)} \widetilde D_{\alpha, \cC}(\rho) \le \log \frac{1}{\widetilde m_\tau}\,,
    \end{equation}
    which follows from $\widetilde D_\alpha(\nu\Vert \tau) \le D_{\infty}(\nu \Vert \tau) \le -\log \widetilde m_\tau$ for $\nu \in \cS(\widetilde \cH)$ (see Eq.\ \eqref{eq:max-upper-bound-on-sandwiched}). Lastly, we find 
    \begin{equation}
        \widetilde D_\alpha(\eta\Vert \tau) = \widetilde D_{\alpha, \cC}(\eta) \qquad \eta \in \{\rho, \sigma\} \, , 
    \end{equation}
    allowing us to use \Cref{theo:minimial-distance-cc-set-almost-additive-approach}, \Cref{theo:minimial-distance-cc-set-operator-approach}, \Cref{theo:minimial-distance-cc-set-mixed-approach}. This yields both assertions.
\end{proof}

We can again consider the limits $\alpha \to 1$ and $\alpha \to \infty$. 

\begin{corollary}[Limits]
    Let $\rho, \sigma, \tau \in \cS(\cH)$ be quantum states, with $\ker \tau \subseteq \ker \rho \cap \ker \sigma$, $\frac{1}{2}\norm{\rho - \sigma}_1 \le \varepsilon$. For the limit $\alpha \to 1$ in Equation \eqref{eq:continuity-bounds-renyi-first-alpha<1}, we obtain
    \begin{equation}
        |D(\rho \Vert \tau) - D(\sigma \Vert \tau)| \le \varepsilon\log(\widetilde m_\tau^{-1}) + (1+\varepsilon) h\left(\frac{\varepsilon}{1+\varepsilon}\right) \, ,
    \end{equation}
    where $\widetilde m_\tau$ is the smallest non-zero eigenvalue of $\tau$, and for \Cref{eq:continuity-bounds-renyi-first-alpha>1}
    \begin{equation}
        |D(\rho \Vert \tau) - D(\sigma \Vert \tau)| \le \min \begin{cases}
            \varepsilon\log(\widetilde m_\tau^{-1}) + (1+\varepsilon) h\left(\frac{\varepsilon}{1+\varepsilon}\right) \, ,\\
            \infty \, , \\
            \varepsilon\log(\widetilde m_\tau^{-1}) + (1+\varepsilon) h\left(\frac{\varepsilon}{1+\varepsilon}\right) \, ,
        \end{cases}
    \end{equation}
    which is exactly the continuity bound in \cite{Bluhm2022ContinuityBounds}. For $\alpha \to \infty$ Equation \eqref{eq:continuity-bounds-renyi-first-alpha>1} gives 
    \begin{equation}
        |D_\infty(\rho\Vert \tau) - D_\infty(\sigma \Vert \tau)| \le \min 
        \begin{cases}
            \log(1+\varepsilon) + \log(\widetilde m_\tau^{-1}) \, ,\\
            \log(1 + \varepsilon \widetilde m_\tau^{-1}) \, ,\\
            \log((1 + \varepsilon)(1 + \varepsilon \widetilde m_\tau^{-1}) - \varepsilon^{2}) \, . 
        \end{cases}
    \end{equation}
\end{corollary}
\begin{proof}
    The proof is a direct application of \Cref{lem:limit-almost-additive}, \Cref{lem:limit-operator}, \Cref{lem:limit-mixed} in the context of \Cref{cor:continuity-bound-renyi-first-argument}.
\end{proof}

\subsection{Divergence bounds for sandwiched Rényi divergences}\label{subsec:divergence-bound}
The continuity bounds in the first argument directly give us divergence bounds of the sandwiched R{\'e}nyi divergences.
\begin{corollary}[Divergence bound]
    Let $\rho$, $\tau \in \mathcal S(\mathcal H)$ with $\ker \tau \subseteq \ker \rho$ and $\tfrac{1}{2}\norm{\rho - \tau}_1 \le \varepsilon$ and $\alpha \in [1/2, 1)$
    \begin{equation} \label{eq:divergence-bounds-renyialpha<1}
        \Dalpha(\rho \Vert \tau) \le \log(1 + \varepsilon) + \frac{1}{1 - \alpha}\log(1 + \varepsilon^\alpha \widetilde m_\tau^{\alpha - 1} - \frac{\varepsilon}{(1 + \varepsilon)^{1 - \alpha}}) \, , 
    \end{equation}
    where $\widetilde m_\tau$ is the smallest non-zero eigenvalue of $\tau$. For $\alpha \in (1, \infty)$ we find 
    \begin{equation} \label{eq:divergence-bounds-renyi-first-alpha>1}
         \Dalpha(\rho \Vert \tau) \le \min\begin{cases}
             \log(1 + \varepsilon) + \frac{1}{\alpha - 1}\log(1 + \varepsilon \widetilde m_\tau^{1 - \alpha} - \frac{\varepsilon^\alpha}{(1 + \varepsilon)^{\alpha - 1}}) \, , \\
             \frac{\alpha}{\alpha - 1} \log(1 + \varepsilon \widetilde m_\tau^{\frac{1 - \alpha}{\alpha}}) \, ,\\
             \log(1 + \varepsilon) + \frac{\alpha}{\alpha - 1}\log(1 + \varepsilon \widetilde m_\tau^{\frac{1 - \alpha}{\alpha}} - \frac{\varepsilon^{2 - \frac{1}{\alpha}}}{(1 + \varepsilon)^{\frac{\alpha - 1}{\alpha}}})  \, . 
         \end{cases}
    \end{equation}
\end{corollary}
\begin{proof}
    The proof is a direct application of \Cref{cor:continuity-bound-renyi-first-argument}, with $\sigma = \tau$.
\end{proof}

\begin{corollary}[Limits]
  Let $\rho$, $\tau \in \mathcal S(\mathcal H)$ with $\ker \tau \subseteq \ker \rho$ and $\tfrac{1}{2}\norm{\rho - \tau}_1 \le \varepsilon$. For the limit $\alpha \to 1$ in Equation \eqref{eq:divergence-bounds-renyialpha<1}, we obtain
    \begin{equation}
        D(\rho \Vert \tau) \le \varepsilon\log(\widetilde m_\tau^{-1}) + (1+\varepsilon) h\left(\frac{\varepsilon}{1+\varepsilon}\right) \, ,
    \end{equation}
    where $\widetilde m_\tau$ is the smallest non-zero eigenvalue of $\tau$, and for \Cref{eq:divergence-bounds-renyi-first-alpha>1}
    \begin{equation}
        D(\rho \Vert \tau) \le \min \begin{cases}
            \varepsilon\log(\widetilde m_\tau^{-1}) + (1+\varepsilon) h\left(\frac{\varepsilon}{1+\varepsilon}\right) \, ,\\
            \infty \, , \\
            \varepsilon\log(\widetilde m_\tau^{-1}) + (1+\varepsilon) h\left(\frac{\varepsilon}{1+\varepsilon}\right) \, ,
        \end{cases}
    \end{equation}
    which is exactly the divergence bound in \cite{Bluhm2022ContinuityBounds}. For $\alpha \to \infty$ in Equation \eqref{eq:divergence-bounds-renyi-first-alpha>1} gives 
    \begin{equation}
        D_\infty(\rho\Vert \tau) \le  \min
        \begin{cases}
            \log(1+\varepsilon) + \log(\widetilde m_\tau^{-1}) \, ,\\
            \log(1 + \varepsilon \widetilde m_\tau^{-1}) \, ,\\
            \log((1 + \varepsilon)(1 + \varepsilon \widetilde m_\tau^{-1}) - \varepsilon^{2}) \, .
        \end{cases}
    \end{equation}
\end{corollary}
\begin{proof}
    The proof is a direct application of \Cref{lem:limit-almost-additive}, \Cref{lem:limit-operator}, \Cref{lem:limit-mixed} in the context of \Cref{cor:continuity-bound-renyi-first-argument}.
\end{proof}

\subsection{Distance to separable states} \label{sec:sep-states}

We consider the distance to the set of separable states in terms of the sandwiched R{\'e}nyi divergence and prove that it is uniformly continuous in the spirit of \cite{DonaldHorodecki1999continuityREentanglement, Winter-AlickiFannes-2016}. Note that the $\alpha \to 1$ limit recovers the Corollary 8 of \cite{Winter-AlickiFannes-2016}.

\begin{corollary}[Distance to the set of separable states]\label{cor:minimal-distance-seperable-states}\label{cor:continuity-bound-renyi-min-dist-sep-states}
    Let $\mathrm{SEP}_{AB} \subseteq \cS(\cH_A \otimes \cH_B)$ be the set of separable states. Then the map 
    \begin{equation}
        \widetilde D_{\alpha, \mathrm{SEP}}:\cS(\cH_A \otimes \cH_B) \to \R, \qquad \rho \mapsto \widetilde D_{\alpha, \mathrm{SEP}}(\rho) :=  \inf\limits_{\tau \in \mathrm{SEP}_{AB}} \widetilde D_\alpha(\rho \Vert \tau)
    \end{equation}
    is uniformly continuous for $\alpha \in [1/2, 1) \cup (1, \infty)$ and for $\rho$, $\sigma \in \cS(\cH_A \otimes \cH_B)$ with $\tfrac{1}{2}\norm{\rho - \sigma}_1 \le \varepsilon$, and $\alpha \in [1/2,1)$
    \begin{equation}
        |\widetilde D_{\alpha, \mathrm{SEP}}(\rho) - \widetilde D_{\alpha, \mathrm{SEP}}(\sigma)| \le \log(1 + \varepsilon) + \frac{1}{1 - \alpha}\log(1 + \varepsilon^\alpha m^{1 - \alpha} - \frac{\varepsilon}{(1 + \varepsilon)^{1 - \alpha}})
    \end{equation}
    and for $\alpha \in (1, \infty)$ 
    \begin{equation}
         |\widetilde D_{\alpha, \mathrm{SEP}}(\rho) - \widetilde D_{\alpha, \mathrm{SEP}}(\sigma)| \le \min\begin{cases}
             \log(1 + \varepsilon) + \frac{1}{\alpha - 1}\log(1 + \varepsilon m^{\alpha - 1} - \frac{\varepsilon^\alpha}{(1 + \varepsilon)^{\alpha - 1}}) \, , \\
             \frac{\alpha}{\alpha - 1} \log(1 + \varepsilon m^{\frac{\alpha - 1}{\alpha}}) \, ,\\
             \log(1 + \varepsilon) + \frac{\alpha}{\alpha - 1}\log(1 + \varepsilon m^{\frac{\alpha - 1}{\alpha}} - \frac{\varepsilon^{2 - \frac{1}{\alpha}}}{(1 + \varepsilon)^{\frac{\alpha - 1}{\alpha}}}) \, , 
         \end{cases}
    \end{equation}
    where $m = \min \{d_A, d_B\}$.
\end{corollary}
\begin{proof}
    The set of separable states is known to be compact and convex and contains the maximally mixed state, which is positive definite. Consider $\alpha > 1$. Due to the joint convexity of $\widetilde Q_\alpha$ we can reduce to pure states and further using \Cref{rem:change-order-in-norm} shows
    \begin{align}
        \inf\limits_{\tau \in \mathrm{SEP}_{AB}} \widetilde Q_\alpha(\rho \Vert \tau) &\leq \sup\limits_{\ket{\psi}} \inf\limits_{\tau \in  \mathrm{SEP}_{AB}} \widetilde Q_\alpha(\ketbra{\psi}{\psi} \Vert \tau)\\
        &=\sup\limits_{\ket{\psi}} \inf\limits_{\tau \in  \mathrm{SEP}_{AB}}  (\bra{\psi} \tau^{\frac{1-\alpha}{\alpha}}\ket{\psi})^\alpha,
    \end{align}
    where $\ketbra{\psi}{\psi}$ is a rank-$1$ projection. Let
    \begin{equation}
        \ket{\psi} = \sum_{i = 1}^{m} \lambda_i \ket{e_i}_A \ket{f_i}_B 
    \end{equation}
    be the Schmidt decomposition of $\ket{\psi}$ such that in particular $\lambda_i \geq 0$ for all $i \in \{1, \ldots, m\}$ and $\sum_i \lambda_i^2 =1$. Both $\{\ket{e_i}\}_i$ and $\{\ket{f_j}\}_j$ are orthonormal sets. Choose
    \begin{equation}
        \tau_0 := \frac{1}{m} \sum_{i = 1}^{m} \dyad{e_i} \otimes \dyad{f_i}.
    \end{equation}
    Then,
        \begin{align}
        \inf\limits_{\tau \in \mathrm{SEP}_{AB}} \widetilde Q_\alpha(\rho \Vert \tau) &\leq \sup\limits_{\ket{\psi}} (\bra{\psi} \tau_0^{\frac{1-\alpha}{\alpha}}\ket{\psi})^\alpha \\
        &=\sup\limits_{\ket{\psi}} m^{\alpha -1} \left(\sum_{i = 1}^{m} \lambda_i^2\right)^\alpha \\
        &= m^{\alpha -1} \, .
    \end{align}
    
     The proof for $\alpha \in [1/2,1)$ proceeds analogously, using the same $\tau_0$. Employing \Cref{theo:minimial-distance-cc-set-almost-additive-approach}, \Cref{theo:minimial-distance-cc-set-operator-approach} and \Cref{theo:minimial-distance-cc-set-mixed-approach} gives the claimed bounds.
\end{proof}

We can conclude the limits $\alpha \to 1$ and $\alpha \to \infty$ as before using \Cref{lem:limit-almost-additive}, \Cref{lem:limit-operator} and \Cref{lem:limit-mixed}.

\begin{corollary}[Limits]
  Let $\rho$, $\sigma \in \mathcal S(\mathcal H)$  and $\tfrac{1}{2}\norm{\rho - \sigma}_1 \le \varepsilon$. For the limit $\alpha \to 1$ in Equation \eqref{eq:divergence-bounds-renyialpha<1}, we obtain
    \begin{equation}
        |D_{ \mathrm{SEP}}(\rho ) - D_{ \mathrm{SEP}}(\sigma )| \le \min \begin{cases}
            \varepsilon\log(m) + (1+\varepsilon) h\left(\frac{\varepsilon}{1+\varepsilon}\right) \, ,\\
            \infty \, , \\
            \varepsilon\log(m) + (1+\varepsilon) h\left(\frac{\varepsilon}{1+\varepsilon}\right) \, ,
        \end{cases}
    \end{equation}
    which is exactly the bound in \cite{Winter-AlickiFannes-2016}. For $\alpha \to \infty$ we infer
    \begin{equation}
        |D_{\infty, \mathrm{SEP}}(\rho) - D_{\infty, \mathrm{SEP}}(\sigma)| \le  \min
        \begin{cases}
            \log(1+\varepsilon) + \log( m^{-1}) \, ,\\
            \log(1 + \varepsilon  m) \, ,\\
            \log((1 + \varepsilon)(1 + \varepsilon  m) - \varepsilon^{2}) \, .
        \end{cases}
    \end{equation}
\end{corollary}

\subsection{Distance measures in resource theories}

While the previous section already gave an application to entanglement theory of our tools developed in \Cref{sec:technical_tools}, we can apply our techniques to any resource theory that includes a set of free states $\mathcal F$ that is compact and convex and contains a positive definite state. Then $\widetilde D_{\alpha, \cF}(\rho)$ quantifies the resourcefulness of the state $\rho$ in terms of its distance to the set of free states (see, e.g., \cite{zhu2017coherence}). This resource measure is known as the \emph{Rényi relative entropy of resource}. Especially, the requirements for $\mathcal F$ hold for the resource theories falling into the framework of \cite{anshu2018quantifying}, requiring the sets of free states to be compact and convex and further contain the maximally mixed state. These theories include the resource theories of 
\begin{enumerate}
    \item entanglement (where the free states are the separable states)
    \item coherence (where the free states are the states diagonal in a fixed basis)
    \item asymmetry (where the free states are those invariant under some group)
    \item nonuniformity and purity (where the only free state is the maximally mixed state)
    \item thermodynamics (where the only free state is the Gibbs state for a fixed Hamiltonian and temperature; for this to be the maximally mixed state, we can take the temperature to be infinite)
    \item contextuality (where the set of free states are non-contextual probability distributions)
    \item stabilizer computation (where the free states are the convex hull of states that can be produced with a Clifford unitary from a standard state such as $\ket{0}$)
\end{enumerate}
For more details on the resource theories in question, we refer the reader to \cite{anshu2018quantifying} and the references therein as well as to the textbook \cite{gour2024resources}. We now present the continuity bounds we obtain on resource measures:

\begin{corollary} (Distance to the set of free states)
    Let $\mathcal F \subset \mathcal S(\mathcal H)$ be the free states of a resource theory such that $\1/d \in \mathcal F$. Then the map 
    \begin{equation}
        \widetilde D_{\alpha, \mathcal F}:\cS(\cH) \to \R, \qquad \rho \mapsto \widetilde D_{\alpha, \mathcal F}(\rho) :=  \inf\limits_{\tau \in \mathcal F} \widetilde D_\alpha(\rho \Vert \tau)
    \end{equation}
    is uniformly continuous for $\alpha \in [1/2, 1) \cup (1, \infty)$ and for $\rho$, $\sigma \in \cS(\cH)$ with $\tfrac{1}{2}\norm{\rho - \sigma}_1 \le \varepsilon$, and $\alpha \in [1/2,1)$
    \begin{equation}
        |\widetilde D_{\alpha, \cF}(\rho) - \widetilde D_{\alpha, \cF}(\sigma)| \le \log(1 + \varepsilon) + \frac{1}{1 - \alpha}\log(1 + \varepsilon^\alpha d^{1 - \alpha} - \frac{\varepsilon}{(1 + \varepsilon)^{1 - \alpha}})
    \end{equation}
    and for $\alpha \in (1, \infty)$ 
    \begin{equation}
         |\widetilde D_{\alpha, \cF}(\rho) - \widetilde D_{\alpha, \cF}(\sigma)| \le \min\begin{cases}
             \log(1 + \varepsilon) + \frac{1}{\alpha - 1}\log(1 + \varepsilon d^{\alpha - 1} - \frac{\varepsilon^\alpha}{(1 + \varepsilon)^{\alpha - 1}}) \, , \\
             \frac{\alpha}{\alpha - 1} \log(1 + \varepsilon d^{\frac{\alpha - 1}{\alpha}}) \, ,\\
             \log(1 + \varepsilon) + \frac{\alpha}{\alpha - 1}\log(1 + \varepsilon d^{\frac{\alpha - 1}{\alpha}} - \frac{\varepsilon^{2 - \frac{1}{\alpha}}}{(1 + \varepsilon)^{\frac{\alpha - 1}{\alpha}}}) \, .
         \end{cases}
    \end{equation}
\end{corollary}
\begin{proof}
    We observe that, using $\1/d \in \cF$,
    \begin{equation}
        \sup_{\rho \in \cS(\cH)}\widetilde D_{\alpha, \cF}(\rho) \leq \log(d) + \sup_{\rho \in \cS(\cH)} \frac{1}{\alpha-1} \log  \tr[\rho^\alpha] \leq \log(d).
    \end{equation}
    Then, the claimed bounds follow from \Cref{theo:minimial-distance-cc-set-almost-additive-approach}, \Cref{theo:minimial-distance-cc-set-operator-approach}, and \Cref{theo:minimial-distance-cc-set-mixed-approach}.
\end{proof}

\begin{remark}
    The bounds presented here use no special knowledge of $\mathcal F$ other than that it contains the maximally mixed state. Note that this state could be replaced with another full-rank state of choice. As seen in \Cref{sec:sep-states} better bounds are in some cases possible for a given resource theory using other states in $\cF$ to estimate the $\kappa$ in \Cref{theo:minimial-distance-cc-set-almost-additive-approach}, \Cref{theo:minimial-distance-cc-set-operator-approach}, and \Cref{theo:minimial-distance-cc-set-mixed-approach}. Similar calculations can also provide continuity bounds on the resource measure in the resource theory of thermodynamics at finite temperature. We leave further explorations of concrete resource theories to future work.
\end{remark}

\subsection{Generalized sandwiched Rényi mutual information}

In this section, we are interested in continuity bounds for the \emph{generalized sandwiched R{\'e}nyi mutual information} defined as
\begin{equation}\label{eq:gen-sandwiched-MI}
    \widetilde I^\uparrow_\alpha(\rho_{AB}\|\tau_A) :=  \inf_{\sigma_B \in \mathcal S(\mathcal H_B)}\widetilde D_\alpha(\rho_{AB}\|\tau_A \otimes \sigma_B)
\end{equation}
for $\rho_{AB} \in \mathcal S(\mathcal H_{AB})$ and $\tau_A \in \mathcal S(\mathcal H_A)$ such that $\ker(\tau_A) \subseteq \ker{\rho_A}$. 

This quantity appears in the context of hypothesis testing \cite{hayashi2016correlation}. More concretely, in appears in the strong converse exponent of the hypothesis test where the null hypothesis is that the state is $\rho_{AB}^{\otimes n}$ and the alternative hypothesis that the state is $\tau_A^{\otimes n} \otimes \sigma_{B^n}$ for some state $\sigma_{B^n}$, not necessarily product. In this case, we impose that the error of the second kind goes to zero exponentially fast with a rate exceeding the mutual information $I(\rho_{AB}\| \tau_A)$. Here, $I(\rho_{AB}\| \tau_A)$ is defined as in Eq.\ \eqref{eq:gen-sandwiched-MI}, but with the Umegaki relative entropy. Then, the authors of \cite{hayashi2016correlation} prove that the error of the first kind converges to $1$ exponentially fast and the rate is determined by $\widetilde I^\uparrow_\alpha(\rho_{AB}\|\tau_A)$ with $\alpha > 1$. Recently, the generalized sandwiched R{\'e}nyi mutual information has also appeared in the context of convex splitting \cite{cheng2023tight, cheng2023quantum}.

For this quantity, we can prove the following continuity bounds:

\begin{corollary}
    Let $\rho_{AB}$, $\sigma_{AB} \in \mathcal S(\mathcal H_{AB})$ and $\tau_A \in \mathcal S(\mathcal H_A)$ such that $\ker\tau_A \subseteq \ker{\rho_A} \cap \ker \sigma_A$. Moreover, let $\frac{1}{2}\|\rho_{AB} - \sigma_{AB}\| \leq \epsilon$ and let $\widetilde m_{\tau}$ be the minimal non-zero eigenvalue of $\tau_A$. Then the function $\eta_{AB} \mapsto \widetilde I^\uparrow_\alpha(\eta_{AB}\|\tau_A)$ is uniformly continuous on $\mathcal S(\mathcal H_{AB})$ with the following continuity bounds: for $\alpha \in [1/2, 1)$, we find
    \begin{equation}
         |\widetilde I^\uparrow_\alpha(\rho_{AB}\|\tau_A) - \widetilde I^\uparrow_\alpha(\sigma_{AB}\|\tau_A)| \le \log(1 + \varepsilon) + \frac{1}{1 - \alpha} \log(1 + \varepsilon^\alpha \left(\frac{m}{\widetilde m_\tau}\right)^{1 - \alpha} - \frac{\varepsilon}{(1 + \varepsilon)^{1 - \alpha}})
    \end{equation}
    and for $\alpha \in (1, \infty)$
    \begin{equation}
         |\widetilde I^\uparrow_\alpha(\rho_{AB}\|\tau_A) - \widetilde I^\uparrow_\alpha(\sigma_{AB}\|\tau_A)| \le \min \begin{cases}
             \log(1 + \varepsilon) + \frac{1}{\alpha - 1}\log(1 + \varepsilon \left(\frac{m}{\widetilde m_\tau}\right)^{\alpha - 1} - \frac{\varepsilon}{(1 + \varepsilon)^{\alpha - 1}}) \, ,\\
             \frac{\alpha}{\alpha - 1} \log(1 + \varepsilon \left(\frac{m}{\widetilde m_\tau}\right)^{\frac{\alpha - 1}{\alpha}}) \, , \\
             \log(1 + \varepsilon) + \frac{\alpha}{\alpha - 1} \log(1 + \varepsilon \left(\frac{m}{\widetilde m_\tau}\right)^{\frac{\alpha - 1}{\alpha}} - \frac{\varepsilon^{2 - \frac{1}{\alpha}}}{(1 + \varepsilon)^{\frac{\alpha - 1}{\alpha}}}) \, , 
         \end{cases}
    \end{equation}
    where $m = \min\{d_A, d_B\}$.
\end{corollary}

\begin{proof}
    We first restrict the Hilbert space $\cH_{AB}$ to the support of $\tau_A \otimes \1_B$ and get that this is a Hilbert space $\widetilde \cH_{AB}$ on which $\tau_A \otimes \1$ is positive definite. Clearly $\rho_{AB}, \sigma_{AB} \in \cS(\widetilde{\cH}_{AB})$ as $\ker \tau_A \subseteq \ker \rho_A \cap \ker \sigma_A$, a fact that can be verified using a purification of $\rho_{AB}$ and an appropriate Schmidt decomposition that for $P_A$ the orthogonal projection onto the support of $\tau_A$, it holds that $(P_A \otimes \mathds 1) \rho_{AB} (P_A \otimes \mathds 1) = \rho_{AB}$. Hence
    \begin{equation}
        \mathcal C(\tau_A) := \left\{\tau_A \otimes \sigma_B: \sigma_B \in \mathcal S(\mathcal H_B)\right\}
    \end{equation}
    is a compact convex subset of the state space $\mathcal S(\mathcal H_{AB})$ which contains a positive definite state (e.g.~$\tau_A \otimes \frac{\1_B}{d_B}$). Finally, we have that 
    \begin{align}
       0 \le \widetilde I^\uparrow_\alpha(\rho_{AB}\|\tau_A) & \leq  \inf_{\sigma_B} \widetilde D_\alpha(\rho_{AB}\|\widetilde m_{\tau} \mathds{1}_A \otimes \sigma_B) \\
        & = - \log(\widetilde m_{\tau}) + \inf_{\sigma_B} \widetilde D_\alpha(\rho_{AB}\| \mathds{1}_A \otimes \sigma_B)
    \end{align}
    where we have used Lemma 4.3 of \cite{Tomamichel_2016} in the inequality. Thus
    \begin{equation}
        \sup_{\rho_{AB}} |\widetilde I^\uparrow_\alpha(\rho_{AB}\|\tau_A)| \leq - \log(\widetilde m_{\tau}) - \widetilde H^\uparrow_\alpha(A|B)_\rho \leq \log \frac{m}{\widetilde m_{\tau}}.
    \end{equation}
    The assertion now follows from applying \Cref{theo:minimial-distance-cc-set-almost-additive-approach}, \Cref{theo:minimial-distance-cc-set-operator-approach} and \ref{theo:minimial-distance-cc-set-mixed-approach}.
\end{proof}

For $\tau_A = \mathds 1/d_A$, we retrieve the bounds on the sandwiched R{\'e}nyi conditional entropy in Section \ref{sec:sand-cond-entropy} as expected. Therefore, they also have an interpretation in a hypothesis-testing scenario.

\begin{corollary}[Limits]
    Let $\rho_{AB}$, $\sigma_{AB} \in \mathcal S(\mathcal H_{AB})$ and $\tau \in \mathcal S(\mathcal H_A)$ such that $\ker\tau_A \subseteq \ker\rho_A \cap \ker \sigma_A$. Moreover, let $\frac{1}{2}\|\rho_{AB} - \sigma_{AB}\| \leq \epsilon$ and $\widetilde m_{\tau}$ be the minimal non-zero eigenvalue of $\tau_A$. Then we find 
    \begin{align}
        |I(\rho_{AB}\|\tau_A) - I(\sigma_{AB}\|\tau_A)| &\le \varepsilon\log \left(\frac{m}{\widetilde m_{\tau}}\right) + (1 + \varepsilon) h\Big(\frac{\varepsilon}{1 + \varepsilon}\Big) \, ,\\
        |\widetilde I^\uparrow_\infty(\rho_{AB}\|\tau_A) - \widetilde I^\uparrow_\infty(\sigma_{AB}\|\tau_A)| &\le \log(1 + \varepsilon\frac{m}{\widetilde m_{\tau}})\, .
    \end{align}
\end{corollary}
\begin{proof}
    This follows from Lemmas \ref{lem:limit-almost-additive} and \ref{lem:limit-operator}.
\end{proof}

\section{\texorpdfstring{$\alpha$}{alpha}-Approximate quantum Markov chains}\label{sec:approxQMC}

The main aim of this section is to study $\alpha$-approximate quantum Markov chains, namely positive states $\rho_{ABC}$ on a tripartite space whose (non-variational) sandwiched Rényi conditional mutual information is small enough. Here, we show that this notion is equivalent to that of approximate quantum Markov chains, i.e.~states for which the conditional mutual information is small enough. Beforehand, we need to introduce some technical results, which concern continuity bounds for sandwiched Rényi divergences in both inputs.

\subsection{Continuity bounds for non-variational Rényi divergences via the ALAFF method}

In contrast to our main results, we prove in this section some continuity bounds for the sandwiched Rényi divergence and its derived quantities with respect to both inputs. The drawback of this more general approach is that the bounds obtained are less tight compared to those previously proven by techniques tailored to continuity bounds where the second input is fixed or optimized over.

In comparison to the quantities studied before, we no longer optimise over the second state but consider it to be the marginal of the input. Let us consider a bipartite Hilbert space $\cH_{AB}=\cH_A \otimes \cH_B$ and $\rho_{AB} \in \cS(\cH_{AB})$. Then, for $\alpha \in [1/2, 1) \cup (1, \infty)$, the \emph{(non-variational) sandwiched R{\'e}nyi conditional entropy} is given by 
    \begin{equation}
        \widetilde H_\alpha(A|B)_\rho := \frac{1}{1 - \alpha} \log \widetilde Q_\alpha(\rho_{AB} \Vert \1_A \otimes \rho_B) \, .
    \end{equation}
Note that we will also use the notation $\alphaQ(A|B)_\rho := \alphaQ(\rho_{AB} \Vert \1 \otimes \rho_B)$. Analogously, for the \textit{(non-variational) sandwiched R{\'e}nyi mutual information} we set
\begin{equation}
    \widetilde I_\alpha(A:B)_\rho :=  \widetilde D_\alpha(\rho_{AB} \Vert \rho_A \otimes \rho_B) \, ,
\end{equation}
and lastly, we define the \textit{(non-variational) sandwiched R{\'e}nyi  conditional mutual information} for $\rho_{ABC} \in \cS(\cH_A \otimes \cH_B \otimes \cH_C)$ by
\begin{equation}
     \widetilde I_\alpha(A:C|B)_\rho := \widetilde H_\alpha(C|B)_\rho - \widetilde H_\alpha(C|AB)_\rho \, .
\end{equation}

The approach to derive continuity bounds will be based on the \textit{ALAFF method} \cite{Bluhm2022ContinuityBounds,Bluhm2022ContinuityBoundsShort}. This is a generalization of the \emph{Alicki-Fannes-Winter (AFW) method}, introduced under this name by Shirokov \cite{Shirokov-AFWmethod-2020, Shirokov-ContinuityReview-2022} and based on the seminal results for continuity bounds of entropies by the authors of \cite{AlickiFannes-2004,Winter-AlickiFannes-2016}.

Let $\cH$ denote a finite-dimensional Hilbert space and $f$ be a real-valued function on the convex set $\cS_0 \subseteq \cS(\cH)$. We say that $f$ is \textit{almost locally affine (ALAFF)}, if there exist continuous functions $a_f, b_f : [0, 1] \to \R$, that are non decreasing on $[0, \tfrac{1}{2}]$, vanish as $p \to 0^+$ and satisfy
\begin{equation}\label{eq:alaff-property}
        -a_f(p) \le f(p \rho + (1 - p) \sigma) - p f(\rho) - (1 - p) f(\sigma) \le b_f(p)
\end{equation}
for all $p \in [0, 1]$ and $\rho, \sigma \in \cS_0$. Moreover, the set $\cS_0$  is called \textit{perturbed $\Delta$-invariant} with perturbation parameter $s \in [0, 1)$ if for all $\rho, \sigma \in \cS_0$ with $\rho \neq \sigma$, there is a state $\tau$ such that both states
\begin{equation}
    \Delta^\pm(\rho, \sigma, \tau) = s \tau + (1 - s) \varepsilon^{-1}[\rho - \sigma]_\pm\in\cS_0 \, ,
\end{equation}
where we fix $\varepsilon = \tfrac{1}{2}\norm{\rho - \sigma}_1$, and denote by $[\,\cdot\,]_\pm$ the positive and negative parts of a self-adjoint operator, respectively.

In the following result, we prove that $\widetilde{Q}_\alpha(\cdot \| \cdot)$ is almost locally affine for $\alpha \in [1/2,1)\cup(1, + \infty)$.
Note that we only need to prove the almost joint concavity part for $\alpha \in (1, + \infty)$ (respectively, convexity, for $\alpha \in [1/2,1)$), since  $\widetilde{Q}_\alpha(\cdot \| \cdot)$ is already jointly convex (resp. jointly concave).

\begin{theorem}\label{lemma:almost_concavity_convexity_Q}
     Let $(\rho_1, \sigma_1), (\rho_2, \sigma_2) \in \cS_{\ker}\coloneqq\{(\rho, \sigma) \in \cS(\cH) \times \cS(\cH) \;:\; \ker \sigma \subseteq \ker \rho\}$, $p \in [0,1]$, and define $\rho := p \rho_1 + (1-p ) \rho_2$ and $\sigma := p \sigma_1 + (1-p ) \sigma_2$, respectively. Let us denote by $m_{\sigma_1}$ and $m_{\sigma_2}$ the minimal non-zero eigenvalue of $\sigma_1$ and $\sigma_2$, respectively. Then, for $\alpha \in [1/2,1)$, we have 
     \begin{equation}\label{eq:xi-lower}
        \widetilde{Q}_\alpha (\rho \| \sigma) \leq p \, \widetilde{Q}_\alpha (\rho_1 \| \sigma_1) + (1-p)  \widetilde{Q}_\alpha (\rho_2 \| \sigma_2) + \xi (\alpha, p, \sigma_1, \sigma_2) \, ,
    \end{equation}
    and for $\alpha \in (1, +\infty)$,
    \begin{equation}\label{eq:xi-upper}
        \widetilde{Q}_\alpha (\rho \| \sigma) \geq p  \widetilde{Q}_\alpha (\rho_1 \| \sigma_1) + (1-p)  \widetilde{Q}_\alpha (\rho_2 \| \sigma_2) + \xi(\alpha, p, \sigma_1, \sigma_2) \, ,
    \end{equation}
    where for $\alpha\in[1/2,1)$ the error term $\xi$ is positive and if $\alpha\in(1,\infty)$ it is negative. Moreover,
    \begin{equation}
        \begin{aligned}
            \xi(\alpha, p, \sigma_1, \sigma_2)&\coloneqq -1+p^\alpha (p + (1-p) m^{-1}_{\sigma_1} )^{1-\alpha}    + (1-p)^{\alpha} (p m^{-1}_{\sigma_2} + (1-p)  )^{1-\alpha}\\
            &\leq (1-\alpha)\sqrt{p}\left(\left(\log(m^{-1}_{\sigma_1})+1\right)(m^{-1}_{\sigma_1} )^{1-\alpha}+\left(m^{-1}_{\sigma_2}+1\right)( m^{-1}_{\sigma_2})^{1-\alpha}\right)\eqqcolon u_\alpha(p)
        \end{aligned}
    \end{equation}
    for $\alpha \in [1/2,1)$
    and 
    \begin{equation}
        \begin{aligned}
            \xi(\alpha, p, \sigma_1, \sigma_2)&\coloneqq -\big(p-p^\alpha (p + (1-p) m^{-1}_{\sigma_1} )^{1-\alpha} \big) m^{1-\alpha}_{\sigma_1}  \\
            &\qquad-  \big((1-p) - (1-p)^{\alpha} (p m^{-1}_{\sigma_2} + (1-p)  )^{1-\alpha} \big)m^{1-\alpha}_{\sigma_2}\\
            &\geq(1-\alpha)\sqrt{p}\left(\left(\log(m^{-1}_{\sigma_1})+1\right)m^{1-\alpha}_{\sigma_1}+\left(m^{-1}_{\sigma_2}+1\right)m^{1-\alpha}_{\sigma_2}\right)\eqqcolon v_\alpha(p)
        \end{aligned}
    \end{equation}
    for $\alpha \in (1,\infty)$. 
\end{theorem}
\begin{proof}
    We start with the case $\alpha \in [1/2, 1)$. Here, we first separate $\rho$ using the superadditivity of $\alphaQ$ (cf. \Cref{lemma:Q-alpha-super/subadditivity})
    \begin{align}
        \widetilde{Q}_\alpha (\rho \| \sigma) & = \tr[ (\sigma^{\frac{1-\alpha}{2 \alpha}} \rho \sigma^{\frac{1-\alpha}{2 \alpha}}  )^\alpha]\leq p^\alpha \tr[ (\sigma^{\frac{1-\alpha}{2 \alpha}} \rho_1 \sigma^{\frac{1-\alpha}{2 \alpha}}  )^\alpha] + (1-p)^{\alpha} \tr[ (\sigma^{\frac{1-\alpha}{2 \alpha}} \rho_2 \sigma^{\frac{1-\alpha}{2 \alpha}}  )^\alpha]\,,
    \end{align}
    To split $\sigma$, let $P_1$ be the projection onto the support of $\sigma_1$. Then, we can upper bound $P_1\sigma P_1$ by $\sigma_1$ using $P_1 \sigma_2 P_1 \le P_1 \le m_{\sigma_1}^{-1} \sigma_1$ to obtain
    \begin{align}
        P_1\sigma P_1 & =  p \sigma_1 + (1-p) P_1 \sigma_2 P_1 \leq \big(p + (1-p) m_{\sigma_1}^{-1} \big) \sigma_1 =:c_{\sigma_1}\sigma_1\,.
    \end{align}
    Then, we rewrite and upper bound by
    \begin{equation}
    \begin{aligned}
          \tr[ (\sigma^{\frac{1-\alpha}{2 \alpha}} \rho_1 \sigma^{\frac{1-\alpha}{2 \alpha}}  )^\alpha]& =\tr[ ( \rho_1^{1/2} P_1\sigma^{\frac{1-\alpha}{ \alpha}}P_1 \rho_1^{1/2})^\alpha] \\
          & \leq \tr[ ( \rho_1^{1/2} (P_1\sigma P_1)^{\frac{1-\alpha}{ \alpha}} \rho_1^{1/2})^\alpha] \\
          &\leq c_{\sigma_1}^{1-\alpha}\tr[(\rho_1^{1/2} \sigma_1^{\frac{1-\alpha}{ \alpha}} \rho_1^{1/2})^\alpha]\,,  
    \end{aligned}
    \end{equation}
    which uses that $x\mapsto x^s$ is operator monotone and operator concave for $s\in[0,1]$ and, in particular, $\frac{1-\alpha}{\alpha}\in(0,1]$ for all $\alpha\in[\frac{1}{2},1)$. Therefore, the first inequality follows from \cite[Theorem V.2.3]{Bhatia1997} Repeating the same steps for $\sigma_2$, this time inserting the projection $P_2$ onto the support of $\sigma_2$, and using the inequalities $\widetilde{Q}_\alpha (\rho_1 \| \sigma_1) \leq 1$ and $\widetilde{Q}_\alpha (\rho_2 \| \sigma_2) \leq 1$,  we obtain
     \begin{align}
        \widetilde{Q}_\alpha (\rho \| \sigma) & \leq  p^\alpha c_{\sigma_1}^{1-\alpha} \widetilde{Q}_\alpha (\rho_1 \| \sigma_1) + (1-p)^{\alpha} c_{\sigma_2}^{1-\alpha} \widetilde{Q}_\alpha (\rho_2 \| \sigma_2) \\
        &\leq p \widetilde{Q}_\alpha (\rho_1 \| \sigma_1) + (1-p) \widetilde{Q}_\alpha (\rho_2 \| \sigma_2)  + \xi(\alpha, p, \sigma_1, \sigma_2)  \, ,
    \end{align}
        for 
        \begin{align}
            \xi(\alpha, p, \sigma_1, \sigma_2) &=( p^\alpha c_{\sigma_1}^{1-\alpha}-p  )  + ((1-p)^{\alpha} c_{\sigma_2}^{1-\alpha} - (1-p) )\\
            &\geq( p^\alpha c_{\sigma_1}^{1-\alpha}-p  ) \widetilde{Q}_\alpha (\rho_1 \| \sigma_1) + ((1-p)^{\alpha} c_{\sigma_2}^{1-\alpha} - (1-p) )\widetilde{Q}_\alpha (\rho_2 \| \sigma_2)\,,
        \end{align}
        which proves the first bound. For $p\in\{0,1\}$ the bound is clear so that w.l.o.g.~we consider $p\in(0,1)$ in the following. Next, we upper bound $\xi(\alpha, p, \sigma_1, \sigma_2)$ further:
    \begin{equation}
        \begin{aligned}
            \xi(\alpha, p, \sigma_1, \sigma_2)&=-p+p^\alpha (p + (1-p) m^{-1}_{\sigma_1} )^{1-\alpha} -(1-p)+ (1-p)^{\alpha} (p m^{-1}_{\sigma_2} + (1-p)  )^{1-\alpha}\\
            &= \int_0^1p\frac{d}{ds}(1 + p^{-1}(1-p)m^{-1}_{\sigma_1} )^{s(1-\alpha)}+(1-p)\frac{d}{ds} ((1-p)^{-1}p m^{-1}_{\sigma_2} + 1 )^{s(1-\alpha)}ds\\
            &= (1-\alpha)\int_0^1p^{1-s(1-\alpha)}\left(\log(p + (1-p)m^{-1}_{\sigma_1})-\log(p)\right)(p + (1-p)m^{-1}_{\sigma_1} )^{s(1-\alpha)}\\
            &\qquad+(1-p)^{1-s(1-\alpha)}\left(\log(p m^{-1}_{\sigma_2} + (1-p))-\log(1-p)\right)(p m^{-1}_{\sigma_2} + (1-p))^{s(1-\alpha)}ds\,,
        \end{aligned}
    \end{equation}
    where the natural logarithm is considered. Next, we bound the following term separately
    \begin{equation}
        \begin{aligned}
            -\int_0^1p^{1-s(1-\alpha)}\log(p)(p + (1-p)m^{-1}_{\sigma_1} )^{s(1-\alpha)}&\leq-(p + (1-p)m^{-1}_{\sigma_1} )^{1-\alpha}\int_0^1p^{1-s/2}\log(p)ds\\
            &{\leq}2(m^{-1}_{\sigma_1} )^{1-\alpha}p(\sqrt{p}-p)\\
            &{\leq}(m^{-1}_{\sigma_1} )^{1-\alpha}\sqrt{p}
        \end{aligned}
    \end{equation}
    where we used $p(\sqrt{p}-p)\leq\frac{1}{2}\sqrt{p}$. 
    Similarly,
    \begin{equation}
        \begin{aligned}
            -\int_0^1(1-p)^{1-s(1-\alpha)}&\log(1-p)(pm^{-1}_{\sigma_2}+(1-p) )^{s(1-\alpha)}\\
            &\leq-(pm^{-1}_{\sigma_2}+(1-p) )^{1-\alpha}\int_0^1(1-p)^{1-s/2}\log(1-p)ds\\
            &{\leq}2(m^{-1}_{\sigma_1} )^{1-\alpha}(1-p)(\sqrt{1-p}-(1-p))\\
            &{\leq}(m^{-1}_{\sigma_1} )^{1-\alpha}\sqrt{p}
        \end{aligned}
    \end{equation}
    which uses that $\sqrt{1-p}-(1-p)\leq\frac{1}{2}\sqrt{p}$ in the last inequality. 
    Therefore, 
    \begin{equation}
        \begin{aligned}
            \xi(&\alpha, p, \sigma_1, \sigma_2)\\
            &= (1-\alpha)\int_0^1p^{1-s(1-\alpha)}\left(\log(p + (1-p)m^{-1}_{\sigma_1})-\log(p)\right)(p + (1-p)m^{-1}_{\sigma_1} )^{s(1-\alpha)}\\
            &\qquad\qquad\qquad+(1-p)^{1-s(1-\alpha)}\left(\log(p m^{-1}_{\sigma_2} + (1-p))-\log(1-p)\right)(p m^{-1}_{\sigma_2} + (1-p))^{s(1-\alpha)}ds\\
            &\leq (1-\alpha)\left(\sqrt{p}\left(\log(m^{-1}_{\sigma_1})+1\right)(m^{-1}_{\sigma_1} )^{1-\alpha}+\left(\log(p m^{-1}_{\sigma_2} + (1-p))+\sqrt{p}\right)( m^{-1}_{\sigma_2})^{1-\alpha}\right)\\
            &\leq (1-\alpha)\left(\sqrt{p}\left(\log(m^{-1}_{\sigma_1})+1\right)(m^{-1}_{\sigma_1} )^{1-\alpha}+\left(\log(p m^{-1}_{\sigma_2} + 1)+\sqrt{p}\right)( m^{-1}_{\sigma_2})^{1-\alpha}\right)\\
            &\leq (1-\alpha)\sqrt{p}\left(\left(\log(m^{-1}_{\sigma_1})+1\right)(m^{-1}_{\sigma_1} )^{1-\alpha}+\left(m^{-1}_{\sigma_2}+1\right)( m^{-1}_{\sigma_2})^{1-\alpha}\right)\,
        \end{aligned}
    \end{equation}
    finishes the bounds for the case $\alpha\in[1/2,1)$. Next, we consider the case of $\alpha \in (1, +\infty)$, which follows a similar line of reasoning. 
    The superadditivity of $\alphaQ$ (cf. \Cref{lemma:Q-alpha-super/subadditivity}) gives
    \begin{align}
        \widetilde{Q}_\alpha (\rho \| \sigma) & \geq p^\alpha \tr[ (\sigma^{\frac{1-\alpha}{2 \alpha}} \rho_1 \sigma^{\frac{1-\alpha}{2 \alpha}}  )^\alpha] + (1-p)^{\alpha} \tr[ (\sigma^{\frac{1-\alpha}{2 \alpha}} \rho_2 \sigma^{\frac{1-\alpha}{2 \alpha}}  )^\alpha]\,.
    \end{align}
    Then, we use the bounds $P_1\sigma P_1\leq c_{\sigma_1}\sigma_1$ and $P_2\sigma P_2\leq c_{\sigma_2}\sigma_2$ again. Since $\alpha \in (1, +\infty)$, the fraction $\frac{1-\alpha }{\alpha} \in (-1,0)$ so that $x \mapsto - x^{\frac{1-\alpha }{\alpha}}$ is operator monotone and operator concave, allowing us to use \cite[Exercise V.2.2]{Bhatia1997} and find
     \begin{align}
        \widetilde{Q}_\alpha (\rho \| \sigma) & \geq  p^\alpha c_{\sigma_1}^{1-\alpha} \widetilde{Q}_\alpha (\rho_1 \| \sigma_1) + (1-p)^{\alpha} c_{\sigma_2}^{1-\alpha} \widetilde{Q}_\alpha (\rho_2 \| \sigma_2) \\
        &= p \widetilde{Q}_\alpha (\rho_1 \| \sigma_1) + (1-p) \widetilde{Q}_\alpha (\rho_2 \| \sigma_2)  + \xi(\alpha, p, \sigma_1, \sigma_2)  \, ,
    \end{align}
    Since $\widetilde{Q}_\alpha (\rho_i \| \sigma_i) \leq m_{\sigma_i}^{1-\alpha}$ for $i = 1, 2$, we find that 
    \begin{align}
        \xi(\alpha, p, \sigma_1, \sigma_2) 
        &=-[ (p - p^\alpha c_{\sigma_1}^{1-\alpha} ) m^{1-\alpha}_{\sigma_1} + ((1-p) - (1-p)^{\alpha} c_{\sigma_2}^{1-\alpha} )m^{1-\alpha}_{\sigma_2}]\\
        &\leq-[ (p - p^\alpha c_{\sigma_1}^{1-\alpha} ) \widetilde{Q}_\alpha (\rho_1 \| \sigma_1) + ((1-p) - (1-p)^{\alpha} c_{\sigma_2}^{1-\alpha} )\widetilde{Q}_\alpha (\rho_2 \| \sigma_2)]\,,
    \end{align}
    giving the third bound in the assertion. At last, we simplify this bound as follows;
    \begin{equation}
        \begin{aligned}
            -\xi(&\alpha, p, \sigma_1, \sigma_2) \\
            &=\big(p-p^\alpha (p + (1-p) m^{-1}_{\sigma_1} )^{1-\alpha} \big) m^{1-\alpha}_{\sigma_1}+ \big((1-p) - (1-p)^{\alpha} (p m^{-1}_{\sigma_2} + (1-p)  )^{1-\alpha} \big)m^{1-\alpha}_{\sigma_2}\\
            &= (\alpha-1)\int_0^1p^{1-s(1-\alpha)}\left(\log(p + (1-p)m^{-1}_{\sigma_1})-\log(p)\right)(p + (1-p)m^{-1}_{\sigma_1} )^{s(1-\alpha)}m^{1-\alpha}_{\sigma_1}\\
            &\qquad\qquad+(1-p)^{1-s(1-\alpha)}\left(\log(p m^{-1}_{\sigma_2} + (1-p))-\log(1-p)\right)(p m^{-1}_{\sigma_2} + (1-p))^{s(1-\alpha)}m^{1-\alpha}_{\sigma_2}ds\\
            &\overset{(i)}{\leq}(\alpha-1)\left(p\left(\log(m^{-1}_{\sigma_1})-\log(p)\right)m^{1-\alpha}_{\sigma_1}+\left(\log(p m^{-1}_{\sigma_2} + (1-p))+\sqrt{p}\right)m^{1-\alpha}_{\sigma_2}\right)\\
            &\overset{(ii)}{\leq}
            (\alpha-1)\left(\sqrt{p}\left(\log(m^{-1}_{\sigma_1})+1\right)m^{1-\alpha}_{\sigma_1}+\left(\log(p m^{-1}_{\sigma_2} + 1)+\sqrt{p}\right)m^{1-\alpha}_{\sigma_2}\right)\\
            &\leq(\alpha-1)\sqrt{p}\left(\left(\log(m^{-1}_{\sigma_1})+1\right)m^{1-\alpha}_{\sigma_1}+\left( m^{-1}_{\sigma_2}+1\right)m^{1-\alpha}_{\sigma_2}\right)\,,
        \end{aligned}
    \end{equation}
    additionally, to the inequalities exploited before, we used in (i), the inequality $-(1-p)\log(1-p)\leq\sqrt{p}$ and in (ii) $-p\log p \le \sqrt{p}$. This finishes the proof.
\end{proof}

\begin{remark}\label{rem:behaviour_error_almost_convexity_Q}
    The function $\xi(\alpha, p, \sigma_1, \sigma_2)$ has the following behaviour:
    \begin{itemize}
        \item If $p=0,1$, we have $\xi(\alpha, p, \sigma_1, \sigma_2)=0$. For the bounds $u_\alpha(p), v_\alpha(p)$ from \Cref{lemma:almost_concavity_convexity_Q} only $p = 0$ makes them vanish.
        \item If $\sigma_1=\sigma_2$, we can replace $\xi(\alpha, p, \sigma_1, \sigma_1)$  by the simpler function $\xi(\alpha, p, \sigma_1)=\big(-1 + p^\alpha+ (1-p)^{\alpha}  \big)$, for $\alpha \in [1/2,1)$, and $\xi(\alpha, p, \sigma_1)=-\big(1 - p^\alpha- (1-p)^{\alpha}  \big)m^{1-\alpha}_{\sigma_1}$, for $\alpha \in (1,+\infty)$, as in these cases we can take $c_{\sigma_1}$ and $c_{\sigma_2}$ to be $1$ in the proof. 
        \item It can be seen from the statement of \Cref{lemma:almost_concavity_convexity_Q} that $p \mapsto u_\alpha(p)$ and $p \mapsto v_\alpha(p)$ are non-decreasing on $[0,1/2]$, since the square root is.
    \end{itemize}
\end{remark}
We are now in position of using the findings of Lemma \ref{lemma:almost_concavity_convexity_Q} for $\widetilde{Q}_\alpha (\cdot \| \cdot)$, jointly with a suitable $\Delta$-invariant set, to apply the ALAFF method (cf. \cite[Theorem 4.6]{Bluhm2022ContinuityBounds}) and obtain continuity bounds. For the time being, we focus on $\widetilde H_\alpha(A|B)$ and $\widetilde I_\alpha(A:C|B)$, however, prove more general continuity bounds for $\widetilde{Q}_\alpha (\cdot \| \cdot)$, and thus $\widetilde{D}_\alpha (\cdot \| \cdot)$, in Section \ref{sec:general_continuity_bounds}. 

\begin{corollary}\label{cor:CB_nonvar_Renyi_CondEnt}
    Let $d_{AB}^{-1} > m > 0$ and $\cS_0 := \{\rho : \rho \in \cS(\cH_{AB}), m_\rho \ge m\}$ with $m_{\rho}$ the minimal eigenvalue of $\rho$. Then for $\rho, \sigma \in \cS_0$ with $\frac{1}{2}\norm{\rho - \sigma}_1 \le \varepsilon$, we find for $\alpha \in [1/2, 1) \cup (1, \infty)$
    \begin{equation}
        |\widetilde{H}_\alpha (A|B)_\rho - \widetilde{H}_\alpha (A|B)_\sigma| \le c(\alpha, m, d_A, d_{AB}) \sqrt{\varepsilon} \, . 
    \end{equation}
    with $d_A, d_{AB}$ the dimensions of $\cH_A$ and $\cH_{AB}$ respectively and
    \begin{equation}
        c(\alpha, m, d_A, d_{AB}) \coloneqq  \left\{\begin{array}{ll}
            \frac{d_A^{2(1 - \alpha)} - 1}{1 - \alpha} \frac{1}{1 - m d_{AB}} + \frac{\sqrt{2} \, \widetilde c(\alpha, m)}{(1 - m d_{AB})} d_A^{2(1 - \alpha)} & \alpha \in [1/2, 1)\\
            2\log d_A \frac{1}{1 - m d_{AB}} + \frac{\sqrt{2} \, \widetilde c(1, m)}{(1 - m d_{AB})} & \alpha = 1\\
            \frac{d_A^{2(\alpha - 1)} - 1}{\alpha - 1} \frac{1}{1 - m d_{AB}} + \frac{\sqrt{2}\, \widetilde c(\alpha, m)}{(1 - m d_{AB})} & \alpha \in (1, \infty) 
        \end{array}\right.
    \end{equation}
    where $\widetilde c(\alpha, m) = \left\{\begin{array}{ll}
         (\log(m^{-1}) + m^{-1} + 2) m^{\alpha - 1} & \alpha \in [1/2, 1]\\
        (\log(m^{-1}) + m^{-1} + 2) m^{1 - \alpha} & \alpha \in (1, \infty)
    \end{array}\right.$.
\end{corollary}
\begin{proof}
    We will only demonstrate the proof for $\alpha > 1$ as the one for $\alpha < 1$ is completely analogous. First note that $\cS_0$ is a convex, $m d_{AB}$-perturbed $\Delta$-invariant set\footnote{One can for example use $\frac{\1}{d_{AB}}$ to perturb.} (see \cite[Corollary 6.8]{Bluhm2022ContinuityBounds} for comparison). We further have that $\rho \mapsto \alphaQ(A|B)_\rho$ as a map from $\cS_0 \to [0, \infty)$ is convex and almost concave with remainder $v_\alpha(p)$ (cf. \Cref{lemma:almost_concavity_convexity_Q}). Lastly, we get 
    \begin{equation}
        \sup\limits_{\mu, \nu \in \cS(\cH)}|\alphaQ(A|B)_\mu - \alphaQ(A|B)_\nu | \le d_A^{2(\alpha - 1)} - 1\, , 
    \end{equation}
    and that $p \mapsto \frac{v_\alpha(p)}{1 - p}$ is monotone on $[0, 1)$ as clearly $\frac{\sqrt{p}}{1 - p}$ is. Employing \cite[Theorem 4.6]{Bluhm2022ContinuityBounds} gives 
    \begin{align}
        |\alphaQ(A|B)_\rho - \alphaQ(A|B)_\sigma| &\le  (d_A^{2(\alpha - 1)} - 1)\frac{\varepsilon}{1 - m d_{AB}} + (\alpha - 1) \frac{\sqrt{1 - m d_{AB} + \varepsilon}}{(1 - m d_{AB})}  \widetilde c(\alpha, m) \sqrt{\varepsilon} \\
        &\le  (d_A^{2(\alpha - 1)} - 1)\frac{\varepsilon}{1 - m d_{AB}} + (\alpha - 1) \frac{\sqrt{2} \,  \widetilde c(\alpha, m)}{(1 - m d_{AB})} \sqrt{\varepsilon}\\
        &\le (\alpha - 1)c(\alpha, m, d_A, d_{AB}) \sqrt{\varepsilon} \, .
    \end{align}
    Note that the normalization of $\alphaQ(A|B)_\rho$ is cancelled in the difference so that the bounds in \Cref{lemma:almost_concavity_convexity_Q} can be applied. In the above estimations, we used $\varepsilon < 1$ twice. If we now assume that w.l.o.g. we have  $\widetilde{H}_\alpha (A|B)_\rho \ge \widetilde{H}_\alpha (A|B)_\sigma$, we can deduce the following from the bound above:
    \begin{align}
         |\widetilde{H}_\alpha (A|B)_\rho - \widetilde{H}_\alpha (A|B)_\sigma| &= \widetilde{H}_\alpha (A|B)_\rho - \widetilde{H}_\alpha (A|B)_\sigma\\
         &=\frac{1}{\alpha - 1} \log \frac{\alphaQ(A|B)_\rho}{\alphaQ(A|B)_\sigma}\\
         &= \frac{1}{\alpha - 1} \log (\frac{\alphaQ(A|B)_\rho - \alphaQ(A|B)_\sigma}{\alphaQ(A|B)_\sigma} + 1)\\
         &\le c(\alpha, m, d_A, d_{AB}) \sqrt{\varepsilon}
    \end{align}
    where we employed $\log(x + 1) \le x$ for $x \ge 0$ and finally $1 \le \alphaQ(A|B)_\sigma$. This concludes the claim.
\end{proof}

It is straightforward to derive a continuity bound for the non-variational sandwiched Rényi conditional mutual information as a consequence of the previous result.

\begin{corollary}\label{cor:CB_nonvar_Renyi_CMI}
    Let $\mathcal{H}_{ABC} = \HH_A \otimes \HH_B \otimes \HH_C$, $d_{ABC}^{-1} > m > 0$, $\rho, \sigma \in \cS(\cH_{ABC})$ with $\rho, \sigma \ge m \1$. If $\frac{1}{2}\norm{\rho - \sigma}_1 \le \varepsilon$, then for $\alpha \in [1/2, 1) \cup (1, \infty)$
    \begin{equation}
        |\widetilde I_\alpha(A:C|B)_\rho - \widetilde I_\alpha(A:C|B)_\sigma| \leq 2c(\alpha, m, d_{C}, d_{ABC}) \sqrt{\varepsilon} \, . 
    \end{equation}
    with $c(\alpha, m, d_C, d_{ABC})$ from \Cref{cor:CB_nonvar_Renyi_CondEnt}.
\end{corollary}
\begin{proof}
    This is a direct consequence of \Cref{cor:CB_nonvar_Renyi_CondEnt}, since we can write
    \begin{equation}
         |\widetilde I_\alpha(A:C|B)_\rho - \widetilde I_\alpha(A:C|B)_\sigma| \leq  |\widetilde H_\alpha(C|B)_\rho - \widetilde H_\alpha(C|B)_\sigma| + |\widetilde H_\alpha(C|AB)_\rho - \widetilde H_\alpha(C|AB)_\sigma| \, .
    \end{equation}
    and $\frac{1}{2}\norm{\rho_{ABC} - \sigma_{ABC}}_1 \le \varepsilon$ implies $\frac{1}{2}\norm{\rho_{BC} - \sigma_{BC}}_1 \le \varepsilon$ by DPI of $\norm{\cdot}_1$. Similarly $\rho_{ABC}, \sigma_{ABC} \ge m\1$ implies $\rho_{BC}, \sigma_{BC} \ge m \1$.
\end{proof}

\subsection{Application: Approximate quantum Markov chains}

In this section, we use the continuity bound for the (non-variational) sandwiched Rényi conditional mutual information to derive a stability result for approximate quantum Markov chains.  Consider a tripartite Hilbert space $\mathcal{H}_{ABC} = \HH_A \otimes \HH_B \otimes \HH_C$ and  $\rho_{ABC} \in \mathcal{S}(\HH_{ABC})$. Let us denote by $\mathcal{R}^\rho_{B\rightarrow BC}(\cdot)$ the Petz recovery map for the partial trace in $B$, given for $X \in \mathcal{B}(\mathcal{H}_{BC})$ by 
\begin{equation}
    \mathcal{R}^\rho_{B\rightarrow BC}(X) := \rho_{BC}^{1/2} (\rho_B^{-1/2} \tr_C[X] \rho_B^{-1/2} \otimes \1_C ) \rho_{BC}^{1/2} \, .
\end{equation}
This can be lifted to a map on $\cB(\cH_{ABC})$ as $\operatorname{id}_A\otimes \mathcal{R}^\rho_{B \to BC}$. To enhance readability, we will omit $\operatorname{id}_A$, $\1_C$, and other identity operators whenever their inclusion is clear from the context. We remind the reader, however, that all matrices in each product act on the same space and are implicitly extended with identities as needed. It is well known \cite{Petz-MonotonicityRelativeEntropy-2003,HaydenJozsaPetzWinter-StrongSubadditivity-2004} that 
\begin{equation}\label{eq:conditions-QMC}
    I(A:C|B)_\rho=0 \; \Leftrightarrow \; \rho_{ABC}= \mathcal{R}^\rho_{B\rightarrow BC} (\rho_{AB}) \; \Leftrightarrow \; \widetilde{I}_\alpha (A:C |B)_\rho = 0  \, \text{ for any }\alpha \in (1/2,1) \cup (1, \infty),
\end{equation}
where the last equivalence can be found, e.g.~in \cite[Corollary 4.23]{GaoWilde-Recoverabilityfdivergences-2021}. We can further replace the Petz recovery map in the previous equivalences by the \textit{universal Petz recovery map} \cite{SutterFawziRenner2016UniversalRecoveryMaps,JungeRennerSutterWildeWinter2018UniversalRecoveryMaps} $\mathcal{R}^{\rho,u}_{B\rightarrow BC}(\cdot)$, given by
\begin{equation}
    \mathcal{R}^{\rho,u}_{B\rightarrow BC}(X) := \int_{\mathbb{R}} \mathcal{R}^{\rho,t/2}_{B\rightarrow BC}(X)  \, \beta_0 (t)  dt\,,\quad\text{with}\quad\beta_0 (t) = \frac{\pi}{2} (\text{cosh} (\pi t) + 1)^{-1} \, ,
\end{equation}
where  $\mathcal{R}^{\rho,t}_{B\rightarrow BC}(\cdot)$ is the \textit{rotated Petz recovery map}, namely
\begin{equation}
    \mathcal{R}^{\rho,t}_{B\rightarrow BC}(X) := \rho_{BC}^{\frac{1}{2}+i t} \rho_B^{-\frac{1}{2}-i t} \tr_C[X] \rho_B^{-\frac{1}{2}-i t} \rho_{BC}^{\frac{1}{2}+i t} \, .
\end{equation}
A state $\rho_{ABC}$ satisfying Equation \eqref{eq:conditions-QMC} is called a \textit{quantum Markov chain}. This notion can be extended to an approximate version in the following way: Given a small $\varepsilon>0$, a state  $\rho_{ABC} \in \cS(\cH_{A} \otimes \cH_B \otimes \cH_C)$ is an \textit{approximate quantum Markov chain} \cite{Sutter-ApproximateQMC-2018} if, and only if, $ I_\rho(A:C | B) < \varepsilon$. In an analogous way, for $\alpha \in (1/2, 1) \cup (1, \infty)$, we can say that $\rho_{ABC} \in \cS(\cH_{ABC})$ is an $\alpha$-\textit{approximate quantum Markov chain} whenever $\widetilde{I}_\alpha (A:C |B)_\rho < \varepsilon$.

In \cite[Section 7.3]{Bluhm2022ContinuityBounds}, some of the authors of the current manuscript proved that a state $\rho_{ABC} \in \cS(\cH_{A} \otimes \cH_B \otimes \cH_C)$ is an approximate quantum Markov chain if, and only if, it is close to its reconstructed state under the Petz recovery map. As a consequence of our new continuity bounds for sandwiched Rényi divergences, we can extend now that result to the case of $\alpha$-approximate quantum Markov chains in the following way.
\begin{proposition}
    Let $\rho_{ABC} \in \mathcal{S}(\HH_{A} \otimes \HH_B \otimes \HH_C)$ be positive definite. Given $\alpha \in (1/2,1) \cup (1,\infty)$, $\rho_{ABC}$ is an $\alpha$-approximate quantum Markov chain if, and only if, it is close to its (rotated, universal) Petz recovery. More specifically,  we have for $\alpha \in (1/2,1)$
    \begin{align}
        & \frac{\alpha}{1-\alpha} \log \left(  1 + \left(  \mathcal{K} \, \norm{\rho_{ABC} - \rho_{BC}^{\frac{1}{2}+i t} \rho_B^{-\frac{1}{2}-i t} \rho_{AB} \rho_B^{-\frac{1}{2}-i t} \rho_{BC}^{\frac{1}{2}+i t}}_1 \right)^{\frac{1}{1 - \frac{1}{2 \alpha} - \varepsilon}} \right) \\
        &\hspace{4cm} \leq \widetilde{I}_\alpha (A:C |B)_\rho \\
        &\hspace{4cm}  \leq c\left(\alpha, \norm{\rho_{ABC}^{-1}}_\infty^{-1}, d_C, d_{ABC}\right) \norm{\rho_{ABC} - \rho_{BC}^{1/2} \rho_B^{-1/2} \rho_{AB} \rho_B^{-1/2} \rho_{BC}^{1/2}}_1^{1/2} \, ,
    \end{align}
   for any $\varepsilon \in \left( 0, 1-\frac{1}{2 \alpha}\right)$, with 
    \begin{equation}
        \mathcal{K} = \left(\left(  \frac{\pi}{e \varepsilon \sin (\pi \frac{1-\alpha}{\alpha})}\right)^{1/2} + 8  \right) \frac{\pi}{2\cosh(\pi t)} \, ,
    \end{equation}
    and for $\alpha \in (1,\infty),$ 
    \begin{align}
        & \frac{\alpha}{\alpha-1} \log \left(  1 + \left(  \mathcal{K}' \, \norm{\rho_{ABC} - \rho_{BC}^{\frac{1}{2}+i t} \rho_B^{-\frac{1}{2}-i t} \rho_{AB} \rho_B^{-\frac{1}{2}-i t} \rho_{BC}^{\frac{1}{2}+i t}}_1 \right)^{\frac{1}{\frac{1}{2 \alpha} - \varepsilon}} \right) \\
        & \hspace{4cm} \leq \widetilde{I}_\alpha (A:C |B)_\rho \\
        & \hspace{4cm} \leq c\left(\alpha, \norm{\rho_{ABC}^{-1}}_\infty^{-1}, d_C, d_{ABC}\right) \norm{\rho_{ABC} - \rho_{BC}^{1/2} \rho_B^{-1/2} \rho_{AB} \rho_B^{-1/2} \rho_{BC}^{1/2}}_1^{1/2} \, ,
    \end{align}
   for any $\varepsilon \in \left( 0, \frac{1}{2 \alpha}\right)$, with 
    \begin{equation}
        \mathcal{K}'=d_C^{\frac{2(1-\alpha)}{\alpha}} \left(\left(  \frac{\pi}{e \varepsilon \sin (\pi \frac{\alpha - 1}{\alpha})}\right)^{1/2} + 8\right) \frac{\pi}{2\cosh(\pi t)} \, ,
    \end{equation}
    and $c(\alpha, \cdot, \cdot, \cdot)$ the function from \Cref{cor:CB_nonvar_Renyi_CondEnt}.
\end{proposition}

\begin{proof}
    The lower bounds appear in \cite[Corollary 4.21]{GaoWilde-Recoverabilityfdivergences-2021}, where the only difference is a term\footnote{Note that the following is the notation from \cite{GaoWilde-Recoverabilityfdivergences-2021} where $\widetilde{Q}_\infty(\rho\Vert \sigma) := \norm{\rho^{1/2}\sigma^{-1} \rho^{1/2}}_\infty$.} $\widetilde{Q}_\infty(\rho_{ABC} \| \rho_{AB} \otimes \1_C/d_C)$ which we lower bounded by one in $\cK$ and $\cK'$ as well as upper bounded it by $d_C^2$ in $\cK'$. For the RHS we first note that 
    \begin{equation}
        \widetilde{I}_\alpha (A:C |B)_\rho = \widetilde H_\alpha(C|B)_\rho - \widetilde H_{\alpha}(C|AB)_\rho \le \widetilde H_\alpha(C|AB)_{\mathcal{R}^{\rho}_{B \to BC}(\rho)} - \widetilde H_{\alpha}(C|AB)_\rho 
    \end{equation}
    by the data processing inequality. An application of \Cref{cor:CB_nonvar_Renyi_CondEnt} proves the claim.
\end{proof}

\subsection{General continuity bounds via the ALAFF method}\label{sec:general_continuity_bounds}

We conclude this section by deriving some continuity bounds for sandwiched Rényi divergences for both inputs. To that end, we first prove a continuity bound for $\widetilde{Q}_\alpha (\cdot \| \cdot)$ which we will subsequently use to obtain one for $\widetilde{D}_\alpha (\cdot \| \cdot)$. For that, we consider a perturbed $\Delta$-invariant set $\cS_0$ which is a modification of the aforementioned $\cS_\text{ker}$.

\begin{theorem}\label{thm:continuity_bound_Q}
    Let $1 > 2 m > 0$ and 
    \begin{equation}
        \cS_{0}\coloneqq\{(\rho, \sigma) \in \cS(\cH) \times \cS(\cH) \;:\; \ker \sigma \subseteq \ker \rho \; , \; 2m \leq  m_{\sigma} \} \, ,
    \end{equation}
    where $m_{\sigma}$ is the minimal eigenvalue of $\sigma$. Then, $\widetilde{Q}_\alpha (\cdot \| \cdot)$ is uniformly continuous on $\cS_0$. For $(\rho_1, \sigma_1), (\rho_2, \sigma_2) \in \cS_{0}$ with $\frac{1}{2} \norm{\rho_1 - \rho_2}_1 \leq \varepsilon \leq 1 $ and $\frac{1}{2} \norm{\sigma_1 - \sigma_2}_1 \leq \delta \leq 1 $, we have for $\alpha \in [1/2, 1)$
    \begin{align}\label{eq:continuity_bound_Q_12}
        | \widetilde{Q}_\alpha (\rho_1 \| \sigma_1) -  \widetilde{Q}_\alpha (\rho_2 \| \sigma_2) | \leq  (1+ \sqrt{2}) \sqrt{\varepsilon} + 2 c(\alpha , m , d_{\mathcal{H}}) \sqrt{\delta}  \, ,
    \end{align}
    and for $\alpha \in ( 1 , \infty)$
    \begin{align}\label{eq:continuity_bound_Q_infty}
        | \widetilde{Q}_\alpha (\rho_1 \| \sigma_1) -  \widetilde{Q}_\alpha (\rho_2 \| \sigma_2) | \leq  (1+ \sqrt{2}) m^{1-\alpha} \sqrt{\varepsilon} + 2 c(\alpha , m , d_{\mathcal{H}}) \sqrt{\delta} \, .
    \end{align}
    with 
    \begin{equation}
        c(\alpha, m, d_{\mathcal{H}}) =  \left\{\begin{array}{ll}
            \frac{1+\sqrt{2}(1-\alpha) (\log(m^{-1}) + m^{-1} + 2) m^{ \alpha- 1} }{1-md_{\mathcal{H}}} & \alpha \in [1/2, 1)\\
            m^{1- \alpha} \frac{1+\sqrt{2}(\alpha-1) (\log(m^{-1}) + m^{-1} + 2) }{1-md_{\mathcal{H}}} & \alpha \in (1, \infty) 
        \end{array}\right.
    \end{equation}
\end{theorem}
\begin{proof}
    Let $(\rho_1, \sigma_1), (\rho_2, \sigma_2) \in \cS_0$ with $\frac{1}{2}\norm{\rho_1 - \rho_2} \le \varepsilon \le 1$ and $\frac{1}{2}\norm{\sigma_1 - \sigma_2}\le \delta \le 1$. We define
    \begin{equation}\label{eq:interpolation_sigmas_RE}
        \overline{\sigma} = \frac{1}{2}\sigma_1 + \frac{1}{2}\sigma_2 \, ,
    \end{equation}
    and obtain
    \begin{equation}
        \begin{aligned}
            \frac{1}{2}\norm{\overline{\sigma} - \sigma_1}_1 &= \frac{1}{4}\norm{\sigma_1 - \sigma_2}_1 \le \frac{\delta}{2} \le 1 \, ,\\
            \frac{1}{2}\norm{\overline{\sigma} - \sigma_2}_1 &= \frac{1}{4}\norm{\sigma_1 - \sigma_2}_1 \le  \frac{\delta}{2} \le 1 \, .
        \end{aligned}
    \end{equation}
        Using this, the triangle inequality shows
    \begin{equation}\label{eq:triangle_ineq_Q_continuity_bound}
        \begin{aligned}
            & |\widetilde{Q}_{\alpha}(\rho_1 \Vert \sigma_1) - \widetilde{Q}_{\alpha}(\rho_2 \Vert \sigma_2)| \\
            &\quad \quad \quad \quad \quad \le \underbrace{|\widetilde{Q}_{\alpha}(\rho_1\Vert \sigma_1) - \widetilde{Q}_{\alpha}(\rho_1 \Vert \overline{\sigma})|}_{(I)} + \underbrace{|\widetilde{Q}_{\alpha}(\rho_1 \Vert \overline{\sigma}) - \widetilde{Q}_{\alpha}(\rho_2 \Vert \overline{\sigma})|}_{(II)} + \underbrace{|\widetilde{Q}_{\alpha}(\rho_2 \Vert \overline{\sigma}) - \widetilde{Q}_{\alpha}(\rho_2 \Vert \sigma_2)|}_{(III)} \, .
        \end{aligned}
    \end{equation}
    In the following, we bound each of these terms separately for the two cases $\alpha\in[1/2,1)$ and $\alpha\in(1,\infty)$. Let us begin with the case $\alpha \in [1/2,1)$:

    \begin{itemize}
        \item For (II), we require a continuity bound for $\widetilde{Q}_{\alpha}(\cdot \| \cdot)$ in the first argument:
        Note that $\rho \mapsto \alphaQ(\cdot\Vert \overline{\sigma})$ as a map from the $0$-perturbed $\Delta$-invariant set $\cS(\cH)$ to the reals is ALAFF using \Cref{lemma:almost_concavity_convexity_Q} and \Cref{rem:behaviour_error_almost_convexity_Q} with $a = 0$ and $b = \big(-1 + p^\alpha+ (1-p)^{\alpha}  \big)$. Moreover, we have that 
        \begin{equation}
            \sup\limits_{\twoline{\rho_1, \rho_2 \in \cS(\cH)}{\frac{1}{2}\|\rho_1- \rho_2\| = 1}} |\widetilde{Q}_{\alpha}(\rho_1\Vert \overline{\sigma}) - \widetilde{Q}_{\alpha}(\rho_2 \Vert \overline{\sigma})| \le 1 \, ,
        \end{equation}
        where we used that $0\leq\widetilde{Q}_{\alpha}(\rho_1\Vert \overline{\sigma}) \leq 1$. Employing \cite[Theorem 4.6]{Bluhm2022ContinuityBounds}, we conclude:
        \begin{equation}\label{eq:Qalpha_II_12}
            \begin{aligned}
                |\widetilde{Q}_{\alpha}(\rho_1 \Vert \overline{\sigma}) - \widetilde{Q}_{\alpha}(\rho_2 \Vert \overline{\sigma})| \leq  \varepsilon + (1+\varepsilon) \left( -1 +\left( \frac{\varepsilon}{1+ \varepsilon} \right)^\alpha + \left( \frac{1}{1+ \varepsilon} \right)^\alpha \right) \, .
            \end{aligned}
        \end{equation}

        \item  For (I) and (III), we need a continuity bound for $\widetilde{Q}_{\alpha}(\cdot \| \cdot)$ in the second argument. We argue for (I) and (III) completely analogous:
 
        By \Cref{lemma:almost_concavity_convexity_Q}, we find that  $\sigma \mapsto \widetilde{Q}_{\alpha}(\rho_1 \| \cdot)$ on $\cS_{\ge m}(\cH) = \{\rho \in \cS(\cH) : \sigma \ge m \1\}$ is ALAFF with $b = \xi(\alpha, p, m\1, m\1)$ and $a = 0$. Moreover, we have that 
        \begin{equation}
            \begin{aligned}
                \sup\limits_{\twoline{\sigma_1, \sigma_2 \in \cS_{\ge m}(\cH)}{\frac{1}{2}\norm{\sigma_1 - \sigma_2}_1 = 1 - m}}|\widetilde{Q}_{\alpha}(\rho_1 \| \sigma_1) - \widetilde{Q}_{\alpha}(\rho_1 \| \sigma_2) | \le 1 \, .
            \end{aligned}
        \end{equation}
        and that $\cS_{\ge m}(\cH)$ is $m d_{\cH}$-perturbed $\Delta$-invariant. Hence \cite[Theorem 4.6]{Bluhm2022ContinuityBounds} and the fact that $\sigma_1, \overline{\sigma} \in \cS_{\ge m}(\cH)$ allows us to conclude: 
        \begin{equation}\label{eq:Qalpha_I_12}
            \begin{aligned}
                |\widetilde{Q}_{\alpha}(\rho_1 \| \sigma_1) - \widetilde{Q}_{\alpha}(\rho_1 \| \overline{\sigma})| 
                &\leq \frac{\delta}{{1-md_{\mathcal{H}}}} + \frac{1 -md_{\mathcal{H}} + \delta}{1 -md_{\mathcal{H}}}  \xi\left(\alpha, \frac{\delta}{1-md_{\mathcal{H}}+\delta}, m\1,m \1 \right)   \\
                &\leq \frac{\delta}{{1-md_{\mathcal{H}}}} + \frac{\sqrt{1 -md_{\mathcal{H}} + \delta}}{1 -md_{\mathcal{H}}} (1-\alpha) (\log(m^{-1}) + m^{-1} + 2) m^{ \alpha- 1} \sqrt{\delta} \\
                &\leq \underbrace{\frac{1 +\sqrt{2}(1-\alpha) (\log(m^{-1}) + m^{-1} + 2) m^{ \alpha- 1} }{1-md_{\mathcal{H}}}}_{c(\alpha , m , d_{\mathcal{H}})} \sqrt{\delta} \, ,
            \end{aligned}
        \end{equation} 
    where we are using $\delta/2 \le \delta \leq \sqrt{\delta}$ for $ \delta \in [0,1]$.
    \end{itemize}
Merging these bounds, we find
\begin{equation}
\begin{aligned}
   |\widetilde{Q}_{\alpha}(\rho_1 \Vert \sigma_1) - \widetilde{Q}_{\alpha}(\rho_2 \Vert \sigma_2)|  &\leq  \varepsilon  + (1+\varepsilon) \left( -1 +\left( \frac{\varepsilon}{1+ \varepsilon} \right)^\alpha + \left( \frac{1}{1+ \varepsilon} \right)^\alpha \right)  + 2 c(\alpha , m , d_{\mathcal{H}}) \sqrt{\delta} \\
     & \leq  (1+ \sqrt{2})  \sqrt{\varepsilon} + 2  c(\alpha , m , d_{\mathcal{H}}) \sqrt{\delta}\, ,
\end{aligned}
\end{equation}
where we are using the following inequality, proven with elementary calculus: 
\begin{equation}
    \begin{aligned}
         (1+\varepsilon) \left( -1 +\left( \frac{\varepsilon}{1+ \varepsilon} \right)^\alpha + \left( \frac{1}{1+ \varepsilon} \right)^\alpha \right) & \leq \sqrt{2 \varepsilon} \, .
    \end{aligned}
\end{equation}
This concludes the case $\alpha \in [1/2,1)$. For $\alpha \in (1, \infty)$ the bound is obtained similarly with the only difference that the uniform bounds in (I), (II), (III) are given by $m^{\alpha - 1}$ instead of 1 (e.g.~for (I) we have $\sup\limits_{\rho_1, \rho_2 \in \cS(\cH)} |\alphaQ(\rho_1\Vert \overline{\sigma}) - \alphaQ(\rho_2\Vert \overline{\sigma})| \le m^{\alpha - 1}$).
\end{proof}

From this result, we can derive the following continuity bound for sandwiched Rényi divergences with respect to the first and second input. 

\begin{corollary}\label{thm:continuity_bound_sandwiched_Renyi}
Let $1 > 2 m > 0$ and 
\begin{equation}
    \cS_{0}\coloneqq\{(\rho, \sigma) \in \cS(\cH) \times \cS(\cH) \;:\; 2m \leq  m_{\sigma} \} \, ,
\end{equation}
where $m_{\sigma}$ is the minimal eigenvalue of $\sigma$. Then, $\widetilde{Q}_\alpha (\cdot \| \cdot)$ is uniformly continuous on $\cS_0$. For $(\rho_1, \sigma_1), (\rho_2, \sigma_2) \in \cS_{0}$ with $\frac{1}{2} \norm{\rho_1 - \rho_2}_1 \leq \varepsilon \leq 1 $ and $\frac{1}{2} \norm{\sigma_1 - \sigma_2}_1 \leq \delta \leq 1 $, we have for $\alpha \in [1/2, 1)$
\begin{align}\label{eq:continuity_bound_D_12}
    | \widetilde{D}_\alpha (\rho_1 \| \sigma_1) -  \widetilde{D}_\alpha (\rho_2 \| \sigma_2) | \leq \frac{m^{\alpha-1}}{1-\alpha } \left[ (1+ \sqrt{2}) \sqrt{\varepsilon} + 2 c(\alpha , m , d_{\mathcal{H}}) \sqrt{\delta}  \right] \, ,
\end{align}
and for $\alpha \in ( 1 , \infty)$
\begin{align}\label{eq:continuity_bound_D_infty}
    | \widetilde{D}_\alpha (\rho_1 \| \sigma_1) -  \widetilde{D}_\alpha (\rho_2 \| \sigma_2) | \leq \frac{1}{\alpha -1} \left[ (1+ \sqrt{2}) m^{1-\alpha} \sqrt{\varepsilon} + 2 c(\alpha , m , d_{\mathcal{H}}) \sqrt{\delta} \right]  \, ,
\end{align}
with $c(\alpha, m , d_{\mathcal{H}})$ as in Theorem \ref{thm:continuity_bound_Q}.
\end{corollary}

\begin{proof}
    For $\alpha \in (1, \infty)$, note that
    \begin{align}
        \widetilde{D}_\alpha (\rho_1 \| \sigma_1) -  \widetilde{D}_\alpha (\rho_2 \| \sigma_2)  & = \frac{1}{\alpha -1} \log \left( \frac{\widetilde{Q}_\alpha (\rho_1 \| \sigma_1)}{\widetilde{Q}_\alpha (\rho_2 \| \sigma_2)} \right) \\ 
        & = \frac{1}{\alpha -1} \log \left( \frac{\widetilde{Q}_\alpha (\rho_1 \| \sigma_1)- \widetilde{Q}_\alpha (\rho_2 \| \sigma_2)}{\widetilde{Q}_\alpha (\rho_2 \| \sigma_2)} + 1 \right) \\
        & \leq \frac{1}{\alpha -1} \frac{\widetilde{Q}_\alpha (\rho_1 \| \sigma_1)- \widetilde{Q}_\alpha (\rho_2 \| \sigma_2)}{\widetilde{Q}_\alpha (\rho_2 \| \sigma_2)} \\
        & \leq \frac{1}{\alpha -1} \left( \widetilde{Q}_\alpha (\rho_1 \| \sigma_1)- \widetilde{Q}_\alpha (\rho_2 \| \sigma_2) \right) \, ,
    \end{align}
    where we are using $\widetilde{Q}_\alpha (\rho_2 \| \sigma_2) \geq 1$ and $\log(x+1)\leq x$ for every $x> -1$. Exchanging the roles of $\rho_1, \sigma_1$ with $\rho_2,\sigma_2$, we conclude
    \begin{align}
        |\widetilde{D}_\alpha (\rho_1 \| \sigma_1) -  \widetilde{D}_\alpha (\rho_2 \| \sigma_2) | &\leq \frac{1}{\alpha -1} \left| \widetilde{Q}_\alpha (\rho_1 \| \sigma_1)- \widetilde{Q}_\alpha (\rho_2 \| \sigma_2) \right| \\
        & \leq    \frac{1}{\alpha -1} \left[ (1+ \sqrt{2}) m^{1-\alpha} \sqrt{\varepsilon} + 2 c(\alpha , m , d_{\mathcal{H}}) \sqrt{\delta} \right] \, .
    \end{align}
    For $\alpha \in [1/2,1)$ the proof follows the same lines.
\end{proof}

\section{Discussion} \label{sec:discussion}
In this paper, we presented a framework for proving uniform continuity bounds for sandwiched Rényi divergencies and presented a comprehensive analysis of the continuity properties of the sandwiched R\'{e}nyi conditional entropy, the sandwiched R\'{e}nyi mutual information and the sandwiched R\'{e}nyi divergence with fixed second argument. While our almost additive approach drew inspiration from \cite{marwah2022} and our operator space approach from \cite{Beigi2022}, we further developed and extended the methodologies introduced in both papers. This extension enabled us to enhance the bounds for the sandwiched Rényi conditional entropy and broaden the applicability of these methods to encompass other entropic measures, which could find practical use in resource theories, for example. Combining the two other approaches, the mixed approach yields bounds that perform well in all regimes and are optimal for large $d_A$ and small $\alpha$. Additionally, we explored the ALAFF method (\cite{Bluhm2022ContinuityBounds}), devised by some of the authors. However, it yielded bounds that underperformed in comparison to those presented here. Comparing the bounds obtained by the operator space, the almost additive and mixed approaches we highlight their strengths and weaknesses. We find that the operator-space approach gives the best bound for large $\alpha$, while the mixed and almost additive methods are optimal for low $\alpha$. To the best of our knowledge, we provide the tightest bounds known for the sandwiched R\'{e}nyi conditional entropy in the regime $\alpha \in (1, \infty)$, the sandwiched R\'{e}nyi mutual information in the range $\alpha \in [1/2, 1) \cup (1, \infty)$ and the sandwiched R\'{e}nyi divergence with fixed second argument in the range $\alpha \in [1/2, 1) \cup (1, \infty)$. Finally, we provide an application of continuity bounds to quantum Markov states. Here, we resort to the ALAFF method, since the almost additive, operator space and mixed approaches rely on the optimization in the second argument and are hence not applicable in this context. 

\vspace{0.2cm}

\emph{Acknowledgments:} We thank Li Gao and Marius Lemm for the valuable feedback and fruitful discussions. A.B. acknowledges the support of the French National Research Agency in the framework of the ``France 2030” program (ANR-11-LABX-0025-01) for the LabEx PERSYVAL,  A.C., P.G.~and T.M.~that of the Deutsche Forschungsgemeinschaft (DFG, German Research Foundation) – Project-ID 470903074 – TRR 352, and T.M. that of the Munich Center for Quantum Sciences and Technology.

\emph{Conflict of interest:} The authors declare that there are no conflicts of interest regarding the publication of this paper. No financial, personal, or professional relationships with other people or organizations have influenced the content or outcomes of this work. All affiliations and sources of financial support have been disclosed, and the research was conducted independently and impartially.

\bibliographystyle{abbrv}
\bibliography{bibliography}
\end{document}